\documentclass[11pt]{article}
\ifdefined\pdfpagewidth \pdfpagewidth=8.5in \pdfpageheight=11in \fi
\ifdefined\pagewidth \pagewidth=8.5in \pageheight=11in \fi
\usepackage{amsmath,amssymb,amsthm,mathtools}
\usepackage{tabularx,array,longtable}
\usepackage{enumitem}
\usepackage{tikz}
\usetikzlibrary{arrows.meta,calc}
\usepackage{pgfplots}
\pgfplotsset{compat=1.18}
\makeatletter
\def\tc@slash{/}
\newcommand{\link}[1]{\texttt{\@tfor\tc@t:=#1\do{\tc@t\ifx\tc@t\tc@slash\allowbreak\fi}}}
\makeatother

\newtheorem{theorem}{Theorem}[section]
\newtheorem{proposition}[theorem]{Proposition}
\newtheorem{lemma}[theorem]{Lemma}
\newtheorem{corollary}[theorem]{Corollary}
\theoremstyle{definition}
\newtheorem{definition}[theorem]{Definition}
\newtheorem{construction}[theorem]{Construction}
\newtheorem{example}[theorem]{Example}
\theoremstyle{remark}
\newtheorem{remark}[theorem]{Remark}

\newcommand{\B}{\{0,1\}}
\newcommand{\Z}{\mathbb{Z}}

\newcommand{\R}{\mathbb{R}}
\newcommand{\Tern}{\{-1,0,1\}}
\newcommand{\cat}{\mathbin{\Vert}}
\newcommand{\Hv}{\mathrm{H}}
\newcommand{\Nhost}{N_{\mathrm{host}}}
\newcommand{\field}[1]{\langle #1\rangle}
\newcommand{\cell}[1]{\texttt{#1}}
\newcommand{\mem}{\mathrm{mem}}
\newcommand{\Th}{\mathrm{Th}}
\newcommand{\nlit}{n_{\mathrm{lit}}}
\newcommand{\nrep}{n_{\mathrm{rep}}}
\newcommand{\phstart}{\textsc{start}}
\newcommand{\phcodon}{\textsc{codon}}
\newcommand{\phgate}{\textsc{gate}}
\newcommand{\phaddr}{\textsc{addr}}
\newcommand{\rocq}[1]{\texttt{#1}}
\DeclareMathOperator{\Lev}{Lev}
\DeclareMathOperator{\Out}{Out}
\DeclareMathOperator{\ser}{ser}
\DeclareMathOperator{\inst}{inst}
\DeclareMathOperator{\depth}{depth}
\DeclareMathOperator{\describe}{describe}
\DeclareMathOperator{\code}{code}
\DeclareMathOperator{\erfc}{erfc}
\DeclareMathOperator{\popcount}{popcount}

\title{Universal Construction and Exact Self-Reproduction\\ in Ternary McCulloch--Pitts Networks}
\author{Charles C. Norton}
\date{}

\begin{document}
\maketitle

\begin{abstract}
A fixed network of McCulloch--Pitts threshold units with weights in $\{-1,0,1\}$
can hold other threshold networks in its state and run them: the state is a ring of
banks of records, each record a unit of a stored network, and each step evaluates
one record. We use such a network to carry out von Neumann's universal construction
and self-reproduction exactly. As a cellular automaton keeps its rule, the fixed
network keeps its weights, and what reproduces is a stored network. A constructor
of $143$ records reads the description
of a network from its tape, builds that network in the next bank, copies the
description onto the next tape and hands control to what it built; started on its
own description, it rebuilds itself, weights included, in every generation. The
scheme scales to a universal computer. A SUBLEQ computer of $17{,}598$ ternary
units, stored as $36{,}080$ records and running a program of $27$ instructions,
builds any network that fits a bank and reproduces itself in the same way, at every
word width from eight bits on; run directly, it emits the serialization of its own
weights, memory and tape. Integer pre-activations give every orbit a margin of
$1/2$, and replicating each unit $r$ times multiplies it by $r$. That suffices
against noise of any size on the pre-activations, but against von Neumann's output
flips only below a threshold inversely proportional to the fan-in. These results
are proved in Rocq.
\end{abstract}

\section{Introduction}\label{sec:intro}

McCulloch and Pitts introduced networks of binary threshold units as a logical
calculus of neural activity~\cite{McCullochPitts1943}, and von Neumann used their
notation for the logical design of the stored-program
computer~\cite{vonNeumann1945}. He later asked how a machine can build a copy of
itself, and answered with a universal constructor: given the description of a
machine, it builds that machine, and, supplied with its own description and a
mechanism that copies descriptions, it builds a copy of itself together with a copy
of the description~\cite{vonNeumann1966}. The self-reproducing automata that
followed~\cite{Arbib1966,Codd1968,Langton1984,Pesavento1995,PerrierSipperZahnd1996},
surveyed in~\cite{Sipper1998}, reproduce configurations of a cellular space under a
fixed transition rule.

This paper asks the same question of threshold networks: whether the weights of a
threshold network can be reproduced exactly by the dynamics of a threshold
network. The networks have weights in $\Tern$ and integer thresholds and are
iterated as a fixed composition of layers, so that each is a map
$F\colon\B^n\to\B^n$ whose behaviour is the orbit of an initial state. Exactness
comes from the integers. A trained network that outputs an approximation of its
weights loses information in every generation, while an orbit of integer
pre-activations carries a weight set bit for bit and survives every perturbation of
a pre-activation below one half.

We answer it with a fixed ternary network that holds other networks as data. Its
state is a ring of banks of records with a tape for each bank. A record is a
two-input unit of a stored network, and each step evaluates one record of the active
bank. The fixed network plays the part of von Neumann's transition rule and the
stored networks the part of his configurations. As in a cellular automaton, the
transition rule keeps its weights, and what reproduces is the stored network, its
weights and biases included.

These objects exist by classical results. By Kleene's recursion theorem a universal
computer has programs that print their own text~\cite{Kleene1938}; universal
circuits and fixed-weight networks run any circuit or network whose code they are
given~\cite{Valiant1976,DonnarummaPreveteTrautteur2012}; and von Neumann's design
separates construction from copying~\cite{vonNeumann1966}. This paper builds them
as ternary threshold networks and proves what they do. It gives exact sizes and
depths for every parameter, shows where ternary weights and threshold units cost a
level of depth, reproduces a stored network bit for bit, weights and biases
included, and bounds the effect of noise.

\subsection{Contributions}

\begin{enumerate}[label=(\arabic*),leftmargin=2em]
\item \emph{A network that runs stored networks.} The ring machine holds $K$ banks of
$S$ records and $K$ tapes of $L$ cells, with heads at absolute positions, and a
fixed ternary netlist of depth $19$ computes its step from every state, for every
window, ring, bank and tape (Theorem~\ref{thm:physnet}). One pass over a bank
evaluates the netlist stored in it, and every ternary unit compiles into records
(Lemmas~\ref{lem:feedback} and~\ref{cor:compile}).
\item \emph{A self-reproducing constructor.} An organism of $143$ records with a
description code of its own builds every network whose description fits a tape
into the next bank, copies the description onto the next tape and hands control to
what it built. Started on its own description it reproduces: the copies accumulate
until every bank holds one, and from then on the state repeats every $K$
generations (Theorems~\ref{thm:construct} and~\ref{thm:organism}).
\item \emph{A universal computer.} A family $H_A$ of ternary networks executes
SUBLEQ, a computer with a single instruction~\cite{MavaddatParhami1988}, on $2^A$
words of $A$ bits. Its step has depth $13$ at every word width, and the member with
eight-bit words has $17{,}598$ units (Proposition~\ref{prop:host}); a read and the
branch each need two levels of ternary units (Lemma~\ref{lem:muxdepth},
Remark~\ref{rem:ternarydepth}). A program of $21$ instructions builds every byte
string from a run-length description of it, and one of $27$, which separates
construction from copying as von Neumann did, makes the host emit the
serialization of itself, its memory and its tape (Theorems~\ref{thm:univ}
and~\ref{thm:selfrep}). This is the quine scheme on a universal machine: a map
outside the network turns the emitted bytes back into a running network.
\item \emph{A universal computer that reproduces by the dynamics.} Stored in the
ring machine, the host of every word width from $8$ on runs the second program,
builds every network that fits a bank and reproduces in every ring as the organism
does, with no map outside the network; with eight-bit words it is $36{,}080$
records (Theorems~\ref{thm:storedconstruct} to~\ref{thm:widths}).
\item \emph{Reproduction under noise.} Every orbit survives perturbations below $1/2$ of each
pre-activation. Replicating every unit $r$ times, Pierce's interwoven
redundancy~\cite{Pierce1965} with threshold units, multiplies that margin by $r$.
This suffices against Gaussian noise of any size on the pre-activations. Against
von Neumann's output flips with probability $\varepsilon$ it suffices only for
$\varepsilon<1/(2ek^\ast)$, $k^\ast$ the largest fan-in, and above $1/(2k^\ast)$ a
disjunction of fan-in $k^\ast$ fails with a probability that tends to one
(Theorem~\ref{thm:amplify}, Proposition~\ref{prop:flip}). A bound computed from the
orbit is the exact probability that some unit of a step fails, as measured, and
exceeds the rate at which the host leaves its orbit by the $5$ to $6\%$ of failures
that the state does not show; a second-order bound certifies the complete run of the
threefold host where the first-order bound exceeds one
(Propositions~\ref{prop:profile} and~\ref{prop:restore}).
\end{enumerate}

Every result is proved in Rocq for every word width, ring and size it names, except
Lemma~\ref{lem:generations}, which concerns the serialization format. Extracted to
OCaml, the Rocq definitions write the published files. The sizes of particular byte
strings, the values of the profile and restoration bounds along
particular orbits, and the measurements are computed by the artifact.
Appendix~\ref{sec:verification} lists the Rocq theorems behind each result and the
computations and runs that corroborate them, and the artifact~\cite{Norton2026}
holds the development, the generators, the evaluators and the record of every run.

Section~\ref{sec:prelim} fixes threshold units, layered maps and netlists.
Section~\ref{sec:interp} gives the ring machine and Section~\ref{sec:organism} the
organism. Sections~\ref{sec:host} to~\ref{sec:storedhost} build the SUBLEQ host, its
programs and its stored form, and Section~\ref{sec:robust} bounds the noise that a reproducing run tolerates.
Section~\ref{sec:related} discusses related work and Section~\ref{sec:conclusion}
concludes. The appendices give the serialization format, the sizes of the
constructions, the records of the organism, the restoration bound, the measurements
under noise and the verification.

\section{Ternary threshold networks}\label{sec:prelim}

This section fixes the networks of the paper: ternary threshold units, the layered
maps they form, and netlists, which describe such a map unit by unit and which
levelization turns into layers. Throughout, $\Sigma=\{0,1,\dots,255\}$ is the set of
bytes, $\Sigma^*$ the set of finite byte strings, $|u|$ the length of
$u\in\Sigma^*$ and $u\cat v$ concatenation. The Heaviside function is
$\Hv\colon\R\to\B$, with $\Hv(z)=1$ if $z\ge 0$ and $\Hv(z)=0$ otherwise, applied
coordinatewise to vectors.

\subsection{Units and layered maps}

\begin{definition}[Ternary threshold unit]\label{def:unit}
Let $k\ge 0$. A \emph{ternary threshold unit} with $k$ inputs is a map
$u\colon\B^k\to\B$ of the form
\[
u(x)=\Hv(\langle w,x\rangle+\theta),\qquad w\in\Tern^k,\ \theta\in\Z.
\]
The integer $\langle w,x\rangle+\theta$ is the \emph{pre-activation}. The
\emph{effective fan-in} of $u$ is $|\{j: w_j\neq 0\}|$.
\end{definition}

\begin{example}\label{ex:basis}
$\mathrm{AND}(x,y)=\Hv(x+y-2)$, $\mathrm{OR}(x,y)=\Hv(x+y-1)$,
$\mathrm{NAND}(x,y)=\Hv(1-x-y)$, $\mathrm{NOT}(x)=\Hv(-x)$, and the identity
$x=\Hv(x-1)$ on $\B$. More generally $\Hv(\sum_{i\in P}x_i-\sum_{i\in
N}x_i-|P|)$ is the conjunction of the literals $x_i$ ($i\in P$) and $\neg x_i$
($i\in N$), and $\Hv(\sum_{i\in P}x_i-\sum_{i\in N}x_i+|N|-1)$ is their
disjunction, so a negated literal costs no unit. Ternary threshold units therefore
form a functionally complete basis. A function that is not linearly separable, such
as parity, needs depth at least two~\cite{Muroga1971}; the exclusive or is
$\mathrm{XOR}(x,y)=\mathrm{AND}(\mathrm{OR}(x,y),\mathrm{NAND}(x,y))$, of depth two.
\end{example}

\begin{definition}[Layered ternary threshold map]\label{def:lttm}
Let $n_0,\dots,n_d$ be positive integers with $n_0=n_d=n$, and for
$\ell=1,\dots,d$ let $W_\ell\in\Tern^{n_\ell\times n_{\ell-1}}$ and
$\theta_\ell\in\Z^{n_\ell}$. The \emph{layered ternary threshold map} (LTTM) with these
parameters is $F=\Phi_d\circ\dots\circ\Phi_1\colon\B^n\to\B^n$, where
$\Phi_\ell(y)=\Hv(W_\ell y+\theta_\ell)$. Its \emph{depth} is $d$, its \emph{size} is
$M=\sum_{\ell=1}^d n_\ell$, the number of its \emph{rows}, and its dynamics is
$x_{t+1}=F(x_t)$ on $\B^n$.
\end{definition}

Every LTTM is a finite automaton with $2^n$ states, and every orbit is eventually
periodic; on threshold networks as finite automata
see~\cite{Kleene1956,Minsky1967}.

\subsection{Netlists and levelization}

A network is easier to design unit by unit than layer by layer. A netlist lists its
units in an order in which every unit comes after the units it reads, and
levelization turns the list into layers.

\begin{definition}[Netlist]\label{def:netlist}
A \emph{ternary threshold netlist} consists of a finite set $I$ of input signals
and a finite sequence $V$ of units. Each unit $v\in V$ has a bias in $\Z$ and a
list of distinct predecessors, each carrying a weight in $\{-1,1\}$; a predecessor
is an input signal or a unit that precedes $v$ in the sequence. An input that a
unit weights by $0$ is not one of its predecessors. The depth of an input signal is
$0$, and the depth of a unit is one more than the largest depth of its
predecessors, or $1$ if it has none. A \emph{clocked netlist} on $n$ state bits is
a netlist with $I=\{x_1,\dots,x_n\}$ in which each $x_i$ is assigned the unit or
input signal that computes its next value.
\end{definition}

\begin{lemma}[Levelization]\label{lem:lev}
Let $N$ be a clocked netlist on $n$ state bits whose units have maximal depth
$d\ge 1$. There is an LTTM $\Lev(N)$ of depth $d$ on $\B^n$ such that
$\Lev(N)(x)$ is the next state computed by $N$ for every $x\in\B^n$. Every unit
of $N$ that contributes to a next-state value occurs in $\Lev(N)$ with its
weights and bias unchanged, exactly once if its depth is less than $d$ and once
for each state bit whose next value it is if its depth is $d$; all other units
of $\Lev(N)$ are identity units $\Hv(x-1)$.
\end{lemma}

\begin{proof}
Remove the units that no next-state value depends on, and place each unit $v$ in
layer $\depth(v)$ and the state bits in layer $0$. For a signal $u$, let
$\lambda(u)$ be the largest layer of a unit that reads it, a signal assigned as
the next value of a state bit counting as read at layer $d+1$. In every layer
$\ell$ with $\depth(u)<\ell<\lambda(u)$, place an identity unit that reads the copy
of $u$ in layer $\ell-1$. A unit $v$ then finds each predecessor $u$ in layer
$\depth(v)-1$, since $\depth(u)\le\depth(v)-1<\lambda(u)$. Layer $d$ holds one row
for each state coordinate, in order: the unit assigned to it if that unit has depth
$d$, and otherwise the copy of the assigned signal. Since $\Hv(x-1)=x$ on $\B$,
induction on $\ell$ shows that each row of layer $\ell$ outputs the value of its
signal of $N$.
\end{proof}

\section{Weights as state: the ring machine}\label{sec:interp}

This section gives a fixed network whose state holds threshold networks as data.
The state is a window of cells, in which a stored network keeps its state, and a
ring of banks and tapes; a bank lists the units of a stored network as records. Each
step evaluates the record under the slot head of the active bank on the window and
moves that head on by one slot. Three more heads read one tape and write the next
bank and the next tape, and this is how one stored network builds another
(Figure~\ref{fig:machine}). Every head moves by absolute positions, so a bank or a
tape need not have the size of what it holds.

\begin{definition}[The ring machine]\label{def:twohead}
Let $\alpha\ge4$, $K\ge2$ and $S,L\ge1$, and put $\beta=8+3\alpha$. A state of the
machine $\mathcal M^\alpha_{K,S,L}$ consists of a window $w\in\B^{2^\alpha}$, a bank index
$\kappa<K$, banks $Q_0,\dots,Q_{K-1}$ of $S$ slots of $\beta$ bits, each with a slot head and a
bit head, and tapes $D_0,\dots,D_{K-1}$ of $L$ cells, each with a head. Indices of banks and tapes are taken modulo $K$: bank
$\kappa$ is \emph{active} and bank $\kappa+1$ \emph{idle}, tape $\kappa$ is the
\emph{read tape} and tape $\kappa+1$ the \emph{copy tape}; the \emph{arm} is the bit of the
idle bank under its bit head, in the slot under its slot head.

Window cell $0$ is a sink, which reads $0$ and keeps nothing written to it. Cells
$1$, $2$ and $3$ are not stored: they name the cell under the read head, under the
arm and under the copy head. Cells $4$ to $12$ are requests: AR, AA and AC move the
read head, the arm and the copy head on by one cell, RR, RA and RC return them to
cell $0$, CB and CT clear the idle bank and the copy tape, and HO hands over.

A \emph{record} is a word $(c_0,c_1,\theta,a_0,a_1,a_o)\in\B^{\beta}$: two weight codes
of two bits, read as $w=1$ for $01$, $w=-1$ for $10$ and $w=0$ otherwise, a bias
$\theta\in[-8,7]$ in two's complement on four bits, and three window addresses of
$\alpha$ bits, most significant first. One step reads the record under the slot head
of the active bank and the values $x_a$ of the cells it names, and
\begin{enumerate}[label=(\roman*)]
\item computes $v=\Hv(w_0x_{a_0}+w_1x_{a_1}+\theta)$;
\item writes $v$ into window cell $a_o$ if $a_o\ge4$, under the arm if $a_o=2$,
under the copy head if $a_o=3$, and nowhere if $a_o\le1$;
\item returns the read head to cell $0$ if RR holds $1$ after the write, and
otherwise moves it on by one cell if AR does; clears the idle bank if CB holds $1$,
and then returns the arm to cell $0$ if RA does and otherwise moves it on by one
bit if AA does; acts on the copy tape likewise with CT, RC and AC;
\item moves the slot head of the active bank on by one slot, and sets
$\kappa\leftarrow\kappa+1$ if HO holds $1$;
\item clears the requests.
\end{enumerate}
Every head moves cyclically, from the last slot or cell to the first. We write
$\mathcal M_{K,S,L}$ for $\mathcal M^6_{K,S,L}$, whose window has $64$ cells and whose records have
$26$ bits.
\end{definition}

\begin{figure}[t]
\centering
\begin{tikzpicture}[font=\footnotesize,>=Stealth,
  c/.style={draw,minimum height=5.5mm,minimum width=9mm,inner sep=0pt,anchor=west,text height=1.5ex,text depth=0.3ex},
  s/.style={draw,minimum height=5.5mm,inner sep=0pt,anchor=west,text height=1.5ex,text depth=0.3ex},
  lab/.style={anchor=east},
  note/.style={font=\scriptsize,text=black!70}]
\node[lab] at (0,0) {window $w$};
\foreach \i/\t in {0/{sink},1/{rd},2/{arm},3/{cpy}} {
  \node[c] (w\i) at (0.15+0.9*\i,0) {\t};
  \node[font=\tiny,text=black!60] at ($(w\i.north)+(0,0.13)$) {\i};
}
\node[s,minimum width=56mm] (wq) at (w3.east) {AR AA AC RR RA RC CB CT HO};
\node[font=\tiny,text=black!60] at ($(wq.north)+(0,0.13)$) {4--12};
\node[s,minimum width=40mm] (wr) at (wq.east) {cells $13,\dots,2^\alpha-1$};
\node[lab] at (0,-1.45) {active bank $Q_\kappa$};
\node[s,minimum width=14mm] (r0) at (0.15,-1.45) {};
\node[s,fill=black!10,minimum width=40mm] (r1) at (r0.east) {record under the slot head};
\node[s,minimum width=20mm] (r2) at (r1.east) {next slot};
\node[s,minimum width=58mm] (r3) at (r2.east) {};
\node[note,anchor=east] at ($(r3.east)+(-0.1,0)$) {$S$ slots of $\beta$ bits};
\draw[->,rounded corners=3pt] (r1.north) -- ++(0,0.42) -| node[pos=0.25,fill=white,inner sep=2pt]
  {$v=\Hv(w_0x_{a_0}+w_1x_{a_1}+\theta)$ into cell $a_o$} (wr.south);
\foreach \y/\n/\l/\k/\q in {-2.35/rt/{read tape $D_\kappa$}/rd/{AR, RR}, -3.05/ib/{idle bank $Q_{\kappa+1}$}/arm/{AA, RA, CB}, -3.75/ct/{copy tape $D_{\kappa+1}$}/cpy/{AC, RC, CT}} {
  \node[lab] at (0,\y) {\l};
  \node[s,minimum width=30mm] (\n0) at (0.15,\y) {};
  \node[s,fill=black!10,minimum width=9mm] (\n1) at (\n0.east) {\k};
  \node[s,minimum width=93mm] (\n2) at (\n1.east) {};
  \node[note,anchor=east] at ($(\n2.east)+(-0.1,0)$) {moved by \q};
}
\end{tikzpicture}
\caption{The ring machine $\mathcal M^\alpha_{K,S,L}$ at bank index $\kappa$. A step
evaluates the record under the slot head of the active bank on the cells $a_0$ and
$a_1$, writes the result into cell $a_o$ and moves the slot head to the next slot.
Window cells $1$, $2$ and $3$ are the cells under the read head, the arm and the
copy head, shaded; the requests in cells $4$ to $12$ move these heads on or return
them to cell $0$, clear the idle bank and the copy tape, and hand over to bank
$\kappa+1$. The other $K-2$ banks and tapes are left as they are.}
\label{fig:machine}
\end{figure}
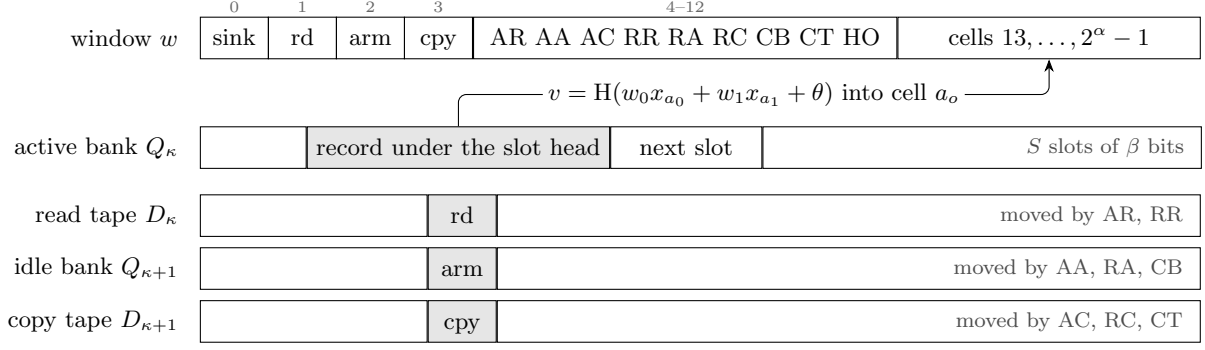

A \emph{pass} is $S$ steps: the slot head of the active bank comes back to the slot
it started from, so a pass evaluates the records of the bank once each, in order,
each on the values the records before it left. The \emph{empty} record, of $\beta$
zeros, reads the sink with weight $0$ and writes the sink, so it only moves the slot
head on. An \emph{image} is the sequence of the bits of a list of records, and a
bank holds an image \emph{ready} when it holds it from slot $0$, followed by empty
slots, with its slot head at slot $0$. A bank that holds an image of $G\le S$
records ready evaluates them in $G$ steps and then idles for $S-G$. The bank index is read at the start of a step,
so a record that sets HO passes control to the next bank from the next step on, at
that bank's current slot.

\begin{theorem}[The machine as a ternary threshold network]\label{thm:physnet}
Let $\alpha\ge4$, $K\ge2$ and $S,L\ge1$, and put $W=2^\alpha$ and $N=S\beta$. There is
a clocked ternary threshold netlist $N^\alpha_{K,S,L}$ of depth $19$, every weight
$\pm1$, on $W-4+K+K(N+S+\beta)+2KL$ state bits: the stored window cells $4,\dots,W-1$,
the bank index in one-hot form, the bits of each bank with its slot and bit heads in
one-hot form, and the cells of each tape with its head in one-hot form. For every
state $x$ of $\mathcal M^\alpha_{K,S,L}$, the next state that $N^\alpha_{K,S,L}$ computes from
the bits of $x$ is the encoding of the next state of $x$.
\end{theorem}

\begin{proof}
The netlist has thirty-two blocks. For each bank, conjunctions of the slot head's
lines with the bits select the record under the slot head, and a disjunction over
the banks weighted by the lines of the bank index selects the record of the active
bank. Conjunctions of head lines with cells give the cell under each head, and the
lines of the bank index pick out the read tape, the idle bank and the copy tape.
Three blocks of $W$ units decode $a_0$, $a_1$ and $a_o$, the line of address $j$
being the conjunction of the address bits that are $1$ in $j$ and the negations of
those that are $0$ (Example~\ref{ex:basis}). A disjunction of the conjunctions of
each decode line with the cell it names reads $x_{a_0}$ and $x_{a_1}$, cells $1$ to
$3$ through the selected heads.

From the weight codes, units for $[w_k=1]\wedge x_{a_k}$ and $[w_k=-1]\wedge
x_{a_k}$ give $p=w_0x_{a_0}+w_1x_{a_1}\in[-2,2]$. Five units test $p\ge c$ for
$c=-2,\dots,2$ and five test $\theta\ge-c$ on the four bits of $\theta$, and $v$ is
the disjunction of the five conjunctions, since $p+\theta\ge0$ exactly when some $c$,
namely $c=p$, has $p\ge c$ and $\theta\ge-c$. A multiplexer on the decode line of
$a_o$ writes $v$ into each stored window cell, and each request takes as its next
value a unit that never fires. A bank bit is the multiplexer of $v$ and its old
value on the conjunction of the bank being idle, the decode line of $2$ and the
bit's slot and bit lines, cleared when CB holds $1$. A head line moves to the next
line when its head advances, to the first when it returns, and stays otherwise; the
slot head of the active bank always advances, and that of the idle bank advances
when the arm advances from the last bit of a slot. The tapes are alike. The depth
of a unit depends only on its block, the requests being decided at depth $14$ and
the slot head lines at depth $19$.
\end{proof}

Appendix~\ref{app:sizes} counts the units and edges of $N^\alpha_{K,S,L}$ and the
rows of its levelization.

\subsection{Netlists in the ring}

A stored network runs in the ring machine as a netlist of units with at most two
predecessors, one record for each unit. The next two lemmas say when one pass over a
bank evaluates such a netlist, and how every ternary unit becomes such units.

\begin{lemma}[Stored netlists]\label{lem:feedback}
If the active bank holds records followed by empty slots and no record but the last
sets HO, a pass evaluates the records in order, each on the window and the cells
under the heads as the records before it left them. Let a clocked netlist have
units of at most two predecessors and its state bits in window cells, and list its
units as records in an order in which every unit follows its predecessors. Let each
record write a cell that holds the unit's value until the last record that reads it,
and let the record of the next value of each state bit write that bit after every
record that reads it. Then one pass over these records evaluates the netlist.
\end{lemma}

\begin{proof}
The slot head of the active bank moves on by one slot per step, the bank index, read
at the start of a step, selects the bank, and the empty slots change nothing.
Induction over the records shows that each reads the values of its predecessors on
the current state, a state bit being unchanged until the record of its next value,
which follows its last reader.
\end{proof}

Lemma~\ref{lem:feedback} is the rank-ordered schedule of gate-level simulation
engines~\cite{Denneau1982} and the fixed-order sequential schedule of automata
networks~\cite{RiosWilsonTheyssier2024c}.

\begin{lemma}[Compilation]\label{cor:compile}
A ternary threshold unit $\Hv(\sum_{i\in P}x_i-\sum_{i\in N}x_i+\theta)$ is $1$
exactly when at least $t=|N|-\theta$ of its $k$ literals $x_i$ ($i\in P$) and $\neg x_i$
($i\in N$) are true. It is a constant for $t\le0$ and for $t>k$, a disjunction
for $t=1$ and a conjunction for $t=k$, each a copy of the literal when $k=1$ and a
chain of $k-1$ two-input units otherwise, and for $1<t<k$ the function $\Th_t$, which
is $1$ exactly when at least $t$ of its arguments are, satisfies
\[
\Th_t(x_1,\dots,x_k)=\Th_t(x_2,\dots,x_k)\vee\bigl(x_1\wedge \Th_{t-1}(x_2,\dots,x_k)\bigr).
\]
Every unit so obtained is a record.
\end{lemma}

\begin{proof}
The pre-activation is the number of true literals minus $t$. The unit
$\Hv(\pm x\pm y+c)$ with $c\in\{0,-1,-2\}$ is the conjunction of any two literals
and with $c\in\{-1,0,1\}$ their disjunction, a copy is $\Hv(x-1)$ or $\Hv(-x)$, and
the constants are the disjunction and the conjunction of a literal and its
negation; a chain conjoins or disjoins each further literal with the result so
far. The recursion holds because at least $t$ of $x_1,\dots,x_k$ are true exactly
when at least $t$ of $x_2,\dots,x_k$ are, or $x_1$ is and at least $t-1$ of the
others are, and its base cases are the chains and the constants.
\end{proof}

\begin{example}[Exclusive or]\label{ex:xor}
Let the active bank of $\mathcal M_{2,4,L}$, whose window has $64$ cells and whose
records have $26$ bits, hold the three records
\[
(1,1,-1;\,13,14,15),\qquad(-1,-1,1;\,13,14,16),\qquad(1,1,-2;\,15,16,13),
\]
written $(w_0,w_1,\theta;\,a_0,a_1,a_o)$, ready, so that its fourth slot is empty.
The first record writes the disjunction of cells $13$ and $14$ into cell $15$, the
second their negated conjunction into cell $16$, and the third the conjunction of
cells $15$ and $16$ into cell $13$. From a window with $x_{13}=x_{14}=1$, the four
steps of a pass write $1$ into cell $15$, $0$ into cell $16$ and $0$ into cell $13$,
and move the slot head from the empty slot back to slot $0$. From every window the
pass replaces $x_{13}$ by $x_{13}\oplus x_{14}$ and keeps $x_{14}$: this is
Lemma~\ref{lem:feedback} for the clocked netlist on the state bits $x_{13}$ and
$x_{14}$ whose next values are $x_{13}\oplus x_{14}$ and $x_{14}$, its units
compiled into records as Example~\ref{ex:basis} and Lemma~\ref{cor:compile}
prescribe. A record that set HO in the fourth slot would pass control to the other
bank at the end of the pass, and the other bank would then run whatever it holds.
\end{example}

\section{A self-reproducing constructor}\label{sec:organism}

This section gives the main example, the organism: a stored network of $143$
records that carries out von Neumann's cycle in $\mathcal M_{K,S,L}$, the ring machine
with $64$ window cells. It reads the description on the read tape twice. In
construct mode it decodes the description into the next bank, in copy mode it copies
the description onto the next tape, and then it hands control to what it built.
Descriptions are written in a code of its own, in which a two-input gate costs four
bits besides its addresses.

\begin{definition}[The description code]\label{def:code}
A \emph{gate} is a record that computes the conjunction or the disjunction of
two literals $\ell_0,\ell_1$, each a window cell or its negation, or a copy of one
literal. Its weights are $\pm1$ by the polarities $u_0,u_1$ of the literals, with
bias $-u_0-u_1$ for a conjunction and $1-u_0-u_1$ for a disjunction; a copy has
weights $(\pm1,0)$ and bias $-u_0$. Its \emph{codon} is the four bits
$(k_1,k_0,u_0,u_1)$, with $k_1k_0$ equal to $00$, $01$ and $10$ for the three
kinds and $u_1=0$ for a copy, and it determines the first eight bits of the record,
its weight codes and its bias. The codons $1100$, $1101$ and $1110$ are
\emph{escapes} and $1111$ is \emph{end}. A \emph{description} of an image $X$ is a
word that the following reading carries to $X$: from a gate codon it emits the
gate's eight bits and the next $18$ bits, from an escape the next $26$ bits, and at
the end codon it stops, having read the whole word. The description $\code(X)$ of an
image $X\in\B^{26n}$ lists, record by record, the codon and the $18$ address bits of
each record whose first eight bits are those of a gate, and the escape $1100$
followed by the $26$ bits of every other record, and ends with $1111$.
\end{definition}

For an image $X$ of $n$ records of which $e$ are not gates, $\code(X)$ has
$22n+8e+4$ bits. Every unit that Lemma~\ref{cor:compile} produces is a gate.

\begin{lemma}[Completeness of the code]\label{lem:code}
For every image $X$ of whole records, $\code(X)$ is a description of $X$.
\end{lemma}

\begin{proof}
By induction on the number of records. The reading of a gate's codon and address
bits gives back the record, since the codon determines the first eight bits; the
reading of an escape and the record's $26$ bits gives them back unchanged; and the
end codon closes the induction.
\end{proof}

\begin{construction}[The organism]\label{con:organism}
The organism $\rho$ is an image of $143$ records, $3718$ bits, listed in
Appendix~\ref{app:organism}. It uses $45$ window cells: the sink, the three head
cells, the nine requests, fourteen cells $13$ to $26$ that hold the state of a
decoder, and eighteen temporaries in cells $27$ to $44$, each written in a pass
before it is read in that pass. The decoder's state is a mode $m$, construct ($m=0$)
or copy, in cell $13$; a phase, one cell each from $14$ to $17$: \phstart{},
\phcodon{} for a codon, \phgate{} for the first eight bits of a record and \phaddr{}
for its addresses; a codon register $\xi_3\xi_2\xi_1\xi_0$, $\xi_i$ in cell $18+i$; and a
counter $z=z_4\cdots z_0$, $z_i$ in cell $22+i$. One pass carries out one step of
the decoder. In phases \phcodon{}, \phgate{} and \phaddr{}, while $z>0$, the step
decrements $z$ and keeps the phase; at $z=0$ it enters the phase given below. With
$\mathrm{rd}$ the bit under the read head, the step acts as follows.
\begin{enumerate}[label=(\roman*),leftmargin=2em]
\item[(o)] \phstart: the read head, the arm and the copy head return to cell $0$, the
idle bank and the copy tape are cleared, and the decoder enters \phcodon{} with
$z=3$ in construct mode.
\item \phcodon: $\mathrm{rd}$ is consumed and shifted into the codon register, and in copy
mode it is also copied under the copy head. At $z=0$ the codon is complete. At the
end codon in construct mode, the read head returns to cell $0$ and the decoder
enters \phcodon{} again, with $z=3$, in copy mode. At the end codon in copy mode,
the arm and the copy head return to cell $0$, the decoder enters \phstart{}, and the
pass hands over in its last record. At a gate codon in copy mode the decoder enters
\phaddr{} with $z=17$, and in every other case \phgate{} with $z=7$.
\item \phgate: in construct mode, bit $7-z$ of the gate's eight bits is written
under the arm for a gate codon, and the consumed bit $\mathrm{rd}$ at an escape; in copy mode,
at an escape, $\mathrm{rd}$ is consumed and copied. At $z=0$ the decoder enters \phaddr{} with
$z=17$.
\item \phaddr: $\mathrm{rd}$ is consumed and written under the arm in construct mode, and
under the copy head in copy mode. At $z=0$ the decoder enters \phcodon{} with $z=3$.
\end{enumerate}
A consumed bit moves the read head on, and a bit written under the arm or the copy
head moves that head on, by the requests the pass sets.
\end{construction}

A record is thus built in $30$ passes, $4$ in \phcodon{}, $8$ in \phgate{} and $18$
in \phaddr{}, and copied in one pass per bit of its description
(Figure~\ref{fig:decoder}).

\begin{figure}[t]
\centering
\begin{tikzpicture}[font=\footnotesize,>=Stealth,
  ph/.style={draw,rounded corners=3pt,minimum height=8mm,minimum width=13mm,align=center},
  lab/.style={font=\scriptsize}]
\node[ph] (I) at (0,0) {\phstart};
\node[ph] (C) at (4.0,0) {\phcodon\\[-2pt]{\scriptsize 4 bits}};
\node[ph] (G) at (8.4,0.8) {\phgate\\[-2pt]{\scriptsize 8 bits}};
\node[ph] (A) at (8.4,-0.8) {\phaddr\\[-2pt]{\scriptsize 18 bits}};
\draw[->] (I) -- node[lab,above]{$z=3$} (C);
\path[->] (C) edge[loop above,looseness=4] node[lab,above]{first end codon: copy mode} (C);
\draw[->] (C) -- node[lab,above,sloped]{$z=7$} (G);
\draw[->] (C) -- node[lab,below,sloped]{$z=17$} (A);
\draw[->] (G) -- node[lab,right]{$z=17$} (A);
\draw[->] (A) to[bend left=28] node[lab,below]{$z=3$} (C);
\draw[->] (C) to[bend left=22] node[lab,below]{second end codon: hand over} (I);
\end{tikzpicture}
\caption{The decoder of the organism, one step per pass, $z$ counting down the bits
of a phase: a codon in \phcodon, the eight bits a gate's codon determines or an
escaped record's first eight in \phgate, the address bits in \phaddr. Construct mode
writes each record under the arm, and copy mode copies every bit read under the copy
head.}
\label{fig:decoder}
\end{figure}
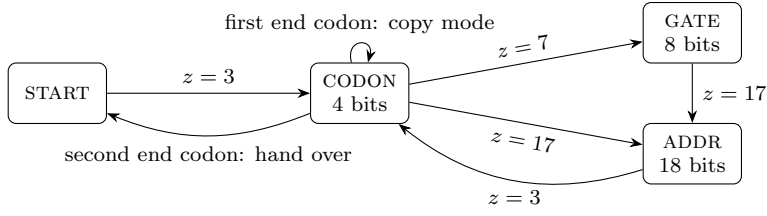

\begin{proposition}[One pass]\label{prop:pass}
Let a state of $\mathcal M_{K,S,L}$, $S\ge143$, hold $\rho$ ready in its active bank,
and in the window clear requests and a decoder state whose counter lies in the range
of its phase, at most $31$, $3$, $7$ or $17$ in \phstart, \phcodon, \phgate{} or
\phaddr. The first $143$ steps carry out the decoder's step of
Construction~\ref{con:organism} on the read tape, the idle bank and the copy tape at
their heads, and leave in the window clear requests and the next decoder state, whose
counter again lies in the range of its phase, whatever the other window cells hold.
They hand over at the $143$rd step exactly when the step ends copy mode; otherwise
the remaining $S-143$ steps bring the slot head back to slot $0$ and change nothing
else.
\end{proposition}

\begin{proof}
A pass of $\mathcal M_{K,S,L}$ is refined by a pass of an abstract machine that sees,
of each head, the cell under it and what the pass does to it: the value it writes
there, the cells it moves, whether it returns to cell $0$ and whether its bank or
tape is cleared. Every temporary is written before it is read, so the abstract pass
from a decoder state does not depend on what the temporaries hold. The proposition
is therefore a statement about the $1984$ decoder states of its hypothesis and the
eight values of the cells under the three heads, and evaluating the abstract pass in
each of these $15{,}872$ cases shows that it carries out the decoder's step, acts on
each head as the step prescribes, and hands over only in its last record. An empty
slot changes nothing but the slot head.
\end{proof}

A state of $\mathcal M_{K,S,L}$ \emph{starts the organism on} a word $\tau$ when its
active bank holds $\rho$ ready, its read tape holds $\tau$ from cell $0$, and its window
holds clear requests and the decoder in \phstart; the other banks, tapes and window
cells may hold anything.

\begin{theorem}[Universal construction by the organism]\label{thm:construct}
Let $K\ge2$ and $S\ge143$, let $X$ be an image of $n\le S$ records, and let $\tau$ be a
description of $X$ with $|\tau|\le L$. From a state that starts the organism on $\tau$,
after $30n+5+|\tau|$ passes, $(30n+4+|\tau|)S+143$ steps:
\begin{enumerate}[label=(\roman*),leftmargin=2em]
\item the next bank holds $X$ ready, with its bit head at $0$, and the next tape holds
$\tau$ from cell $0$, followed by zeros, with its head at $0$;
\item the active bank holds $\rho$ and the read tape $\tau$, the window holds the
decoder in \phstart{} with $z=1$ and the codon register $1111$ in construct mode,
and no other bank or tape has changed;
\item the bank index has moved to the next bank, at the last step.
\end{enumerate}
From then on $\mathcal M_{K,S,L}$ evaluates the records of $X$.
\end{theorem}

\begin{proof}
By Proposition~\ref{prop:pass} it suffices to follow the decoder over whole tapes.
The pass of \phstart{} clears the next bank and the next tape and returns every head
to cell $0$. In construct mode each record of $X$ costs $30$ passes and is written
under the arm, and the end codon costs $4$ passes, the last returning the read head
to cell $0$ and switching to copy mode; so the arm has written the $26n$ bits of $X$
in order from bit $0$, moving on once per bit. In copy mode every pass consumes one
bit and copies it, so after $|\tau|$ passes the copy tape holds $\tau$ from cell $0$; the
last pass, on the end codon, returns the arm and the copy head to cell $0$ and hands
over. Only that pass hands over, so every pass before it has $S$ steps and it has
$143$.
\end{proof}

With Lemma~\ref{lem:code}, the organism therefore builds in $\mathcal M_{K,S,L}$
every image of at most $S$ records whose description fits the tape, whatever its
size.

\begin{theorem}[Self-reproduction of the organism]\label{thm:organism}
Let $K\ge2$, $S\ge143$ and $L\ge3150=|\code(\rho)|$, and put
$T=(30\cdot143+4+3150)S+143=7444S+143$. Let $x_0$ start the organism on
$\code(\rho)$, followed by zeros, with the decoder in the state in which a generation
ends, \phstart{} with $z=1$ and the codon register $1111$ in construct mode. For every
$g\ge1$:
\begin{enumerate}[label=(\roman*),leftmargin=2em]
\item after $gT$ steps the bank index has moved $g$ banks on, and for every
$a\le\min(g,K-1)$ the bank that was active $a$ generations before holds $\rho$ and
its tape $\code(\rho)$, so that $\min(g+1,K)$ banks hold the organism; the banks
and tapes not yet reached are those of $x_0$;
\item for $g\ge K$, the state after $(g+K)T$ steps equals the state after $gT$
steps, the window included;
\item the network $N^6_{K,S,L}$ started from the bits of $x_0$ holds the bits of
these states at the same steps, and its orbit is periodic from the $K$-th generation
on, with period $KT$.
\end{enumerate}
\end{theorem}

\begin{proof}
The first generation is Theorem~\ref{thm:construct} with $X=\rho$ and
$\tau=\code(\rho)$, a description of $\rho$ by Lemma~\ref{lem:code}, without
escapes. After it the next bank holds $\rho$ and the next tape $\tau$ followed by zeros,
and the window holds the decoder in \phstart{} with the register $1111$, so the
hypothesis holds again for the next bank. Induction over the generations, keeping
for each bank and tape the number of generations since it was active, gives the
state after every generation bank by bank and tape by tape. After $K$ generations
every bank and tape has been rebuilt, and from then on they depend only on the bank
index. In the window, the cells no record writes keep their values, and every other
cell from $4$ on ends each generation as a function of the decoder's state at the
pass that hands over. That state is the same in every generation, since after the
pass of \phstart{} the decoder sees the same tapes. So for $g\ge K$ the state after
generation $g+K$ equals the state after generation $g$, and
Theorem~\ref{thm:physnet} carries the statement to the network.
\end{proof}

The description of $\rho$ has $3150$ bits, against $3718$ for the image: a gate is
named by four bits, and the organism writes its eight bits itself.

\section{A universal computer: the SUBLEQ host}\label{sec:host}

The organism builds and copies stored networks and computes nothing else. The next
three sections put a universal computer in its place: a ternary network that
executes SUBLEQ, a computer with a single instruction, with memory-mapped tape and
output devices. This section builds the network, one member $H_A$ for each word
width $A$; Section~\ref{sec:univ} programs it, and Section~\ref{sec:storedhost}
stores it in the ring machine, where it reproduces as the organism does.

\subsection{The machine and its devices}

A network reads and writes bytes through a tape transducer: a device outside the
network that acts on designated coordinates of the state between steps.

\begin{definition}[Tape transducer]\label{def:transducer}
A \emph{tape transducer} consists of an LTTM $F$ on $\B^n$ and disjoint designated
coordinate sets: an input register $R_{\mathrm{in}}$ and an output register
$R_{\mathrm{out}}$ of the same width $A\ge8$, and single bits $c_{\mathrm{rd}}$ (read
request), $c_{\mathrm{rw}}$ (rewind request), $c_{\mathrm{wr}}$ (write request),
$c_{\mathrm{eot}}$ (end-of-tape status) and $c_{\mathrm{halt}}$. An environment
state is a triple $(\tau,h,\omega)$ with tape $\tau\in\Sigma^*$, head position
$h\in\{0,\dots,|\tau|\}$ and output word $\omega\in\Sigma^*$. One step maps
$(x,\tau,h,\omega)$ to $(x',\tau,h',\omega')$ as follows. Let $y=F(x)$ and
initialize $x'=y$.
\begin{enumerate}[label=(\roman*)]
\item If $y_{c_{\mathrm{wr}}}=1$, then $\omega'$ is $\omega$ followed by the low
eight bits of $y_{R_{\mathrm{out}}}$, and $x'_{R_{\mathrm{out}}}=0$; otherwise
$\omega'=\omega$.
\item If $y_{c_{\mathrm{rd}}}=1$ and $h<|\tau|$, then $x'_{R_{\mathrm{in}}}$ holds
$\tau_h$ in its low eight bits and copies of its sign bit above them,
$x'_{c_{\mathrm{eot}}}=0$, and $h'=h+1$. If
$y_{c_{\mathrm{rd}}}=1$ and $h=|\tau|$, then $x'_{c_{\mathrm{eot}}}=1$ and $h'=h$.
If $y_{c_{\mathrm{rd}}}=0$, then $h'=h$.
\item If $y_{c_{\mathrm{rw}}}=1$, then $h'=0$.
\item Finally $x'_{c_{\mathrm{rd}}}=x'_{c_{\mathrm{rw}}}=x'_{c_{\mathrm{wr}}}=0$.
\end{enumerate}
The transducer halts at the first step with $y_{c_{\mathrm{halt}}}=1$, that step
included. If it halts from $(x_0,\tau,0,\varepsilon)$ with final output word
$\omega$, we write $\Out(x_0,\tau)=\omega$.
\end{definition}

\begin{definition}[SUBLEQ]\label{def:subleq}
Let $A\ge3$. The memory is $\mem\in\{0,\dots,2^A-1\}^{2^A}$ with arithmetic
modulo $2^A$, and the program counter is $p\in\{0,\dots,2^A-1\}$. One step reads
$a=\mem[p]$, $b=\mem[p+1]$ and $c=\mem[p+2]$ and sets $\mem[b]\leftarrow \mem[b]-\mem[a]\bmod 2^A$. It
then sets $p\leftarrow c$ if the new value of $\mem[b]$, read in two's complement on
$A$ bits, is at most $0$, and $p\leftarrow p+3$ otherwise~\cite{MavaddatParhami1988}.
Indices are taken modulo $2^A$, and the machine halts when $p=2^A-1$. For $A\ge8$ the
registers and the request and status bits of Definition~\ref{def:transducer} are
memory words at the seven top addresses, in the order of Table~\ref{tab:device}; a
request fires on the value $1$, a read loads a tape byte into $R_{\mathrm{in}}$
with its sign bit copied into the bits above the eighth, and an emission appends
the low byte of $R_{\mathrm{out}}$. For $3\le A<8$ a word does not hold a byte, and
only the step is considered. We write $\mathrm{SUBLEQ}_A$ for this machine. The host
is $A=8$: $256$ words of one byte.
\end{definition}

The bit $c_{\mathrm{halt}}$ is a state bit of the network that latches $p=2^A-1$ for
the state the step produces, so a halted state is a fixed point, and the network
treats the device words as ordinary memory.

\begin{table}[ht]
\centering\small
\begin{tabular}{@{}llll@{}}
\hline
Address & Symbol & Role & Effect \\
\hline
\cell{0xF9} & $c_{\mathrm{wr}}$  & write request     & value $1$ emits $R_{\mathrm{out}}$ and clears it\\
\cell{0xFA} & $c_{\mathrm{eot}}$ & end-of-tape status& set by a read at $h=|\tau|$, cleared by a read at $h<|\tau|$\\
\cell{0xFB} & $c_{\mathrm{rw}}$  & rewind request    & value $1$ sets $h=0$\\
\cell{0xFC} & $c_{\mathrm{rd}}$  & read request      & value $1$ loads $R_{\mathrm{in}}$ and advances $h$\\
\cell{0xFD} & $R_{\mathrm{in}}$  & input register    & written by the device\\
\cell{0xFE} & $R_{\mathrm{out}}$ & output register   & written by the program, cleared on emit\\
\cell{0xFF} & --                 & halt              & the machine halts when $p=255$\\
\hline
\end{tabular}
\caption{The memory-mapped device cells of the host, $A=8$.}
\label{tab:device}
\end{table}

\subsection{Cells}

A \emph{cell} below is a fixed small netlist; value bits are numbered from the least
significant. The cells are a subtractor, the branch and the memory. The branch needs two levels of
ternary units because the weights are ternary (Lemma~\ref{lem:branch} and
Remark~\ref{rem:ternarydepth}), and a read needs two levels of threshold units
whatever the weights (Lemma~\ref{lem:muxdepth}).

\begin{lemma}[Carry-lookahead subtractor]\label{lem:sub}
Let $x,y\in\{0,\dots,2^A-1\}$ have value bits $x_i,y_i$. Put
$p_i=\Hv(y_i-x_i)$ for $i<A-1$, and for $1\le k<A$
\[
\psi_{k,\ast}=\Hv\Bigl(\textstyle\sum_{i<k}p_i-k\Bigr),\qquad
\psi_{k,j}=\Hv\Bigl(y_j-x_j+\textstyle\sum_{j<i<k}p_i-(k-j)\Bigr)\quad(0\le j<k),
\]
$C_k=\Hv\bigl(\psi_{k,\ast}+\sum_{j<k}\psi_{k,j}-1\bigr)$ and
$h_{m,k}=\Hv(y_k-x_k+C_k+1-m)$ for $m=1,2,3$. Let $v_0=\mathrm{XOR}(x_0,y_0)$ and
$v_k=\Hv(h_{1,k}-h_{2,k}+h_{3,k}-1)$ for $1\le k<A$. Then
$\sum_kv_k2^k=(y-x)\bmod 2^A$, every unit is ternary, and the cell has depth five
over $x$ and $y$, whatever $A$.
\end{lemma}

\begin{proof}
$y-x\equiv y+\bar x+1\pmod{2^A}$, where $\bar x$ complements the bits of $x$. The
unit $p_i$ is $y_i\vee\neg x_i$, $\psi_{k,j}$ the conjunction of $y_j$, $\neg x_j$ and
$p_{j+1},\dots,p_{k-1}$, and $\psi_{k,\ast}$ that of $p_0,\dots,p_{k-1}$, so $C_k$
unrolls the recursion $C_k=(y_{k-1}\wedge\neg x_{k-1})\vee(p_{k-1}\wedge C_{k-1})$,
$C_0=1$, of the carry into position $k$ of $y+\bar x+1$. The sum bit $k\ge1$ is the
parity of $y_k+(1-x_k)+C_k\in\{0,1,2,3\}$, and $h_{m,k}$ is $1$ exactly when this sum
is at least $m$, so $v_k$ reads the parity; the carry-in being $1$,
$v_0=x_0\oplus y_0$. The depths are one for the $p_i$, two for the terms, three for
the carries, four for the $h_{m,k}$ and five for the sum bits.
\end{proof}

\begin{lemma}[Branch unit]\label{lem:branch}
For $v=(v_0,\dots,v_{n-1})\in\B^n$ in two's complement with sign bit $v_{n-1}$, let
$z(v)=\Hv\bigl(-\sum_{i<n}v_i\bigr)$. Then $z(v)=1$ if and only if $v=0$, and the
branch condition $[v\le 0]$ equals $\Hv(v_{n-1}+z(v)-1)$.
\end{lemma}

\begin{proof}
$-\sum_i v_i\ge 0$ holds exactly when every $v_i=0$. The value $v$ is at most $0$
exactly when $v=0$ or $v_{n-1}=1$, so $[v\le 0]=\mathrm{OR}(v_{n-1},z(v))$.
\end{proof}

\begin{remark}[Ternary weights cost one level at the branch]\label{rem:ternarydepth}
The branch predicate on $n$ bits in two's complement, $f_n(v)=1$ exactly when
$\sum_{i<n-1}2^iv_i-2^{n-1}v_{n-1}\le0$, is the single unit
$\Hv\bigl(2^{n-1}v_{n-1}-\sum_{i<n-1}2^iv_i\bigr)$ with integer weights, but no
ternary unit computes it when $n\ge3$. If $\Hv(\langle w,v\rangle+\theta)$ with
$w\in\Tern^n$ did, then $v=0$ would give $\theta\ge0$, each unit vector $e_i$ with
$i<n-1$ would give $w_i=-1$ and then $\theta=0$, and the vector of ones, whose value is
negative, would force $w_{n-1}\ge n-1\ge2$. So the cell of Lemma~\ref{lem:branch}
has the least ternary depth.
\end{remark}

\begin{lemma}[Address decoder and memory]\label{lem:mem}
For $j<2^A$ with bits $j_{A-1},\dots,j_0$ let
\[
\mathrm{sel}_j(\alpha)=\Hv\Bigl(\textstyle\sum_{i<A}(-1)^{1-j_i}\alpha_i-\popcount(j)\Bigr).
\]
Then $\mathrm{sel}_j(\alpha)=1$ if and only if $\alpha=j$. Writing $\mu_{j,t}$ for
bit $t$ of memory word $j$, the read multiplexer
$\mathrm{OR}_{j}\mathrm{AND}(\mu_{j,t},\mathrm{sel}_j(\alpha))$ returns bit $t$ of
$\mem[\alpha]$, and for a data word $\delta$ and a halt bit $h$ the cell
\[
\mu'_{j,t}=\mathrm{OR}\bigl(\mathrm{AND}(\mu_{j,t},\neg w_j),\ \mathrm{AND}(\delta_t,w_j)\bigr),
\qquad w_j=\mathrm{AND}(\mathrm{sel}_j(\alpha),\neg h),
\]
leaves every word unchanged when $h=1$ and otherwise writes $\delta$ into
$\mem[\alpha]$ and leaves the other words unchanged.
\end{lemma}

\begin{proof}
The sum $\sum_i(-1)^{1-j_i}\alpha_i$ equals
$|\{i:j_i=1,\alpha_i=1\}|-|\{i:j_i=0,\alpha_i=1\}|$, which is at most
$\popcount(j)$, with equality exactly when $\alpha=j$. The rest follows from the
definitions of $\mathrm{AND}$ and $\mathrm{OR}$, the negated literals being
carried by weights $-1$, and from the fact that exactly one $\mathrm{sel}_j$ is
$1$.
\end{proof}

\begin{lemma}[A read needs two levels]\label{lem:muxdepth}
No threshold unit, whatever its real weights, computes the selection
$(\alpha,\mu)\mapsto\mu_\alpha$ for $\alpha\in\B^A$ and $\mu\in\B^{2^A}$, $A\ge1$,
or computes $\bigvee_j(\mu_j\wedge P_j)$ on the inputs where $P$ is one-hot. So the
read multiplexer, two levels deep over the decode lines, has the least possible
depth.
\end{lemma}

\begin{proof}
Fixing inputs of a threshold unit leaves a threshold unit, so it suffices to fix
all but $\alpha_0$, $\mu_0$ and $\mu_1$, with the other address bits $0$. The
selection is $1$ at $(\alpha_0,\mu_0,\mu_1)=(0,1,0)$ and $(1,0,1)$ and $0$ at
$(0,0,1)$ and $(1,1,0)$. If $\Hv(\langle w,z\rangle+\theta)$ computed it, summing the
inequalities at the two points where it is $1$ would give
$\langle w,(1,1,1)\rangle+2\theta\ge0$, and at the two points where it is $0$ the
strict opposite, a contradiction. The one-hot case is the same argument on
$(P_0,P_1,\mu_0,\mu_1)$ at $(1,0,1,0)$ and $(0,1,0,1)$ against $(1,0,0,1)$ and
$(0,1,1,0)$.
\end{proof}

\subsection{The host family}

\begin{construction}[The host family]\label{con:host}
For $A\ge3$ the netlist $H_A$ consists of the blocks of Table~\ref{tab:blocks}, in
that order, each unit after its predecessors. Its state is, in this order, the $A$
counter bits, a halt bit and the $A\cdot2^A$ memory bits word by word, the counter
and each word most significant bit first, so $n=A+1+A2^A$. For $A=8$ the counter
occupies coordinates $0$ to $7$, the halt bit $8$, and bit $k$ of word $j$, $k=0$
the most significant, coordinate $9+8j+k$. A step is computed in four parts.
\begin{enumerate}[label=(\roman*),leftmargin=2em]
\item The decode lines $P_j$ of $p$ fetch the three operands by reindexing,
$\mem[p+o]_t=\bigvee_j(\mu_{j+o,t}\wedge P_j)$, and give each bit of $p+3$ as the
disjunction of the lines $P_j$ for which that bit of $j+3\bmod 2^A$ is $1$.
\item The decode lines of $a$ and $b$ read $x=\mem[a]$ and $y=\mem[b]$ and give the write
enables.
\item Lemma~\ref{lem:sub} gives $v=y-x$. The branch is the single unit
$\Hv\bigl(h_{1,A-1}-h_{2,A-1}+h_{3,A-1}+e-1\bigr)$, where $e$ is the conjunction of
the bitwise equivalences of $x$ and $y$, since $v\le0$ exactly when its sign bit is
$1$ or $x=y$.
\item With $q=p+3$, the next counter bit is
$(h\wedge p_i)\vee(\neg h\wedge\mathrm{leq}\wedge c_i)\vee(\neg h\wedge\neg\mathrm{leq}\wedge q_i)$
and the next halt bit is
$h\vee(\neg h\wedge\mathrm{leq}\wedge[c=2^A-1])\vee(\neg h\wedge\neg\mathrm{leq}\wedge[p=2^A-4])$.
The next memory is the write cell of Lemma~\ref{lem:mem} with address $b$ and data
$v$.
\end{enumerate}
The host is $\Nhost=H_8$.
\end{construction}

\begin{table}[ht]
\centering\small
\begin{tabular}{@{}lrcl@{}}
\hline
Block & Units & Layer & Computes\\
\hline
$P_j$ & $2^A$ & 1 & decode lines of $p$\\
$HP_i$ & $A$ & 1 & $h\wedge p_i$\\
$FA_{o,t,j}$ & $3A2^A$ & 2 & $\mu_{j+o,t}\wedge P_j$\\
$Q_i$ & $A$ & 2 & bit $i$ of $p+3$\\
$F_{o,t}$ & $3A$ & 3 & the operands $a$, $b$, $c$\\
$PA_j$, $PB_j$ & $2\cdot2^A$ & 4 & decode lines of $a$ and $b$\\
$AC$ & $1$ & 4 & $[c=2^A-1]$\\
$W_j$ & $2^A$ & 5 & write enables\\
$RX_{t,j}$, $RY_{t,j}$ & $2A2^A$ & 5 & read conjunctions\\
$KP_{j,t}$ & $A2^A$ & 6 & $\mu_{j,t}\wedge\neg W_j$\\
$X_t$, $Y_t$ & $2A$ & 6 & $x=\mem[a]$, $y=\mem[b]$\\
subtractor & $\tfrac12(A-1)(A+14)+3$ & 7--11 & $v=y-x$ (Lemma~\ref{lem:sub})\\
equality & $3A+1$ & 7--9 & $e=[x=y]$\\
$\mathrm{LEQ}$ & $1$ & 11 & $[v\le0]$\\
$TA_i$, $TB_i$, $NP_i$ & $3A$ & 12--13 & next counter\\
$HA$, $HB$, $NH$ & $3$ & 12--13 & next halt bit\\
$WR_{j,t}$, $NM_{j,t}$ & $2A2^A$ & 12--13 & next memory (9--10 for $t=0$)\\
\hline
\end{tabular}
\caption{The blocks of $H_A$ in the order of the netlist. The layer is the depth
of the block's units.}
\label{tab:blocks}
\end{table}

\begin{proposition}[Constant-depth host]\label{prop:host}
For every $A\ge3$, one step of the netlist $H_A$, from every state, is one step of
$\mathrm{SUBLEQ}_A$ with the halt latch. Its weights are ternary, its biases are
integers in $[-A,1]$, and its step has depth $13$. For $A\ge8$, with the device on
its seven top words, it realizes the transducer of Definitions~\ref{def:transducer}
and~\ref{def:subleq}. The host $\Nhost=H_8$ has $17{,}598$ units.
\end{proposition}

\begin{proof}
Each block computes what Table~\ref{tab:blocks} states: the decoders, the reads and
the write by Lemma~\ref{lem:mem}, the subtraction by Lemma~\ref{lem:sub}, and the
rest by the identities of Construction~\ref{con:host}. The three terms of the next
counter are mutually exclusive. The next counter is $2^A-1$ exactly when one of the
two branch cases leads there or the machine is already halted, and $p+3=2^A-1$
exactly when $p=2^A-4$. When the halt bit is $1$, the next counter is $p$, the write
enables are $0$ and the halt bit stays $1$, so a halted state is a fixed point. The
layers follow from those of the predecessors, and the longest paths end at the next
counter, the next halt bit and the next memory in layer $13$, whatever $A$.
\end{proof}

Appendix~\ref{app:sizes} gives the numbers of units, edges and parameters of $H_A$
for every $A$, the sizes of $\Lev(\Nhost)$, and the length of its serialization.

\begin{remark}[Depth $11$]\label{rem:fold}
Writing every fetch conjunction $FA_{o,t,j}$, every read conjunction $RX_{t,j}$,
$RY_{t,j}$ and every write enable $W_j$ as a single unit that reads the bits of its
address directly, of fan-in $A+1$, removes the decode lines of $a$ and $b$ and one
level each from the fetch and from the operand reads. The folded netlist computes
the step of $\mathrm{SUBLEQ}_A$ from every state with depth $11$ for every $A\ge3$,
each of its units computing the value of the unit of $H_A$ in its place. For $A=8$ it
has $17{,}086$ units and $120{,}508$ edges, against $17{,}598$ and $51{,}132$ for
$\Nhost$.
\end{remark}

\section{Programs for the host}\label{sec:univ}

This section programs the host. Section~\ref{sec:instances} says what it means for
the host, with a memory image and a tape, to reproduce itself. A program $P$ of $21$
instructions decodes a run-length recipe into the byte string it describes, which
makes the host a universal constructor of byte strings
(Section~\ref{sec:univconst}). A program $P^\ast$ of $27$ instructions also copies
the recipe, so that the host emits a serialization of itself from which it is
rebuilt (Section~\ref{sec:selfrep}). Both programs are stated for $\Nhost=H_8$ and
run on $H_A$ at every word width from $8$ on (Remark~\ref{rem:wider}).

\subsection{Instances and reproduction}\label{sec:instances}

An instance pairs a network with an initial memory and a tape, and it reproduces
itself when its output rebuilds it. The \emph{serialization} $\sigma(N)$ of a
clocked netlist $N$ is a byte string from which $N$, the order of its units
included, is recovered; Appendix~\ref{app:sigma} specifies it. Netlists with the
same serialization differ only by renaming of units, and identifying such netlists,
$\sigma$ is injective.

\begin{definition}[Instances, framing, instantiation]\label{def:instance}
Let $N$ be a clocked netlist realizing the SUBLEQ transducer with $A=8$. For
$u\in\Sigma^*$ let $\field{u}=\lambda(|u|)\cat u$, where $\lambda(k)$ is the
eight-byte little-endian encoding of $k$. An \emph{instance} is a triple
$\Omega=(N,m,\tau)$ with initial memory image $m\in\Sigma^{256}$ and tape
$\tau\in\Sigma^*$. Its initial state $x_0(m)\in\B^n$ has memory contents $m$,
program counter $0$ and all other coordinates $0$, the counter and every memory byte
stored most significant bit first at the coordinates the state layout names. Define
\[
\ser(\Omega)=\field{\sigma(N)}\cat\field{m}\cat\tau,
\qquad
\Out(\Omega)=\Out(x_0(m),\tau)\ \text{for the transducer } \Lev(N),
\]
and let $\inst$ be the partial map
$\inst(\field{a}\cat\field{b}\cat c)=(\sigma^{-1}(a),b,c)$, defined when $a$ lies
in the image of $\sigma$ and $|b|=256$.
\end{definition}

\begin{definition}[Self-description and self-reproduction]\label{def:selfrep}
An instance $\Omega=(N,m,\tau)$ is \emph{self-describing} if it halts and
$\Out(\Omega)=\sigma(N)$. It is \emph{exactly self-reproducing} if it halts and
$\Out(\Omega)=\ser(\Omega)$.
\end{definition}

\begin{lemma}\label{lem:generations}
If $\Omega$ is exactly self-reproducing, then $\inst(\Out(\Omega))=\Omega$.
Consequently the sequence $\Omega_0=\Omega$, $\Omega_{k+1}=\inst(\Out(\Omega_k))$
is constant, and every $\Omega_k$ is exactly self-reproducing.
\end{lemma}

\begin{proof}
The length fields make the decomposition $\field{a}\cat\field{b}\cat c$ unique,
and $\sigma$ is injective, so $\inst\circ\ser$ is the identity on instances. The
second statement follows by induction on $k$.
\end{proof}

The maps $\inst$ and $\Lev$, together with the device of
Definition~\ref{def:transducer}, form the \emph{instantiation map}: the network
emits a description, and $\inst$, $\sigma^{-1}$ and $\Lev$ turn it into the next
running network. Section~\ref{sec:storedhost} removes this map: held as records in
the ring machine, the host carries out the same step by the dynamics, as the
organism does.

\subsection{Universal construction}\label{sec:univconst}

A universal constructor emits any prescribed object from a description of it. Here
the objects are byte strings and the descriptions are run-length recipes. Every
string has a recipe barely longer than itself, the encoder used later attains that
bound, and the program $P$ turns any recipe into the string it describes. For a byte
$x$ we write $\nu(x)=(-x)\bmod 256$.

\begin{definition}[Recipe language]\label{def:recipe}
The set $\mathcal{R}\subseteq\Sigma^*$ of \emph{recipes} is generated by the
grammar
\[
\mathcal{R}\ ::=\ T^*\,\texttt{END},\qquad
T\ ::=\ \mathrm{LIT}\mid\mathrm{REP},\qquad \texttt{END}=\cell{0x00},
\]
where a literal token $\mathrm{LIT}$ is a byte $n$ with $1\le n\le 127$ followed
by $n$ bytes, and a repeat token $\mathrm{REP}$ is a byte $t$ with
$129\le t\le 255$ followed by one byte. The tag $128$ occurs in no recipe. The
decoding map $\delta\colon\mathcal{R}\to\Sigma^*$ is defined token by token by
\[
\delta(\mathrm{LIT}: n\,x_1\cdots x_n)=\nu(x_1)\cdots\nu(x_n),\qquad
\delta(\mathrm{REP}: t\,x)=\nu(x)^{\,256-t},\qquad \delta(\texttt{END})=\varepsilon,
\]
and extended by concatenation. Every payload byte is stored negated, so that a
single subtraction from a zero register recovers it.
\end{definition}

\begin{lemma}[Completeness of the recipe language]\label{lem:complete}
For every $f\in\Sigma^*$ there is $\tau\in\mathcal{R}$ with $\delta(\tau)=f$ and
$|\tau|\le |f|+\lceil |f|/127\rceil+1$.
\end{lemma}

\begin{proof}
Cut $f$ into $q=\lceil |f|/127\rceil$ blocks of length $127$, the last possibly
shorter, and let $\tau$ be the literal tokens of the blocks followed by
\texttt{END}. Each block contributes one tag byte and its own length, so
$|\tau|=|f|+q+1$, and $\delta(\tau)=f$ because $\nu$ is an involution.
\end{proof}

\begin{proposition}[Optimality of the bound]\label{prop:optimal}
If $f\in\Sigma^*$ has no two equal adjacent bytes, every $\tau\in\mathcal{R}$ with
$\delta(\tau)=f$ satisfies $|\tau|\ge |f|+\lceil |f|/127\rceil+1$.
\end{proposition}

\begin{proof}
Write $\tau=t_1\cdots t_k\texttt{END}$. The tokens cover consecutive blocks of $f$,
which has no two equal adjacent bytes, so a repeat token covers one byte with two
and a literal token at most $127$ with one more: $|\tau|=|f|+k+1$ and
$k\ge\lceil|f|/127\rceil$.
\end{proof}

\begin{lemma}[The encoder]\label{lem:describe}
Let $\describe(f)$ scan $f$ from its first position. At position $i$, let $\ell$ be
the length of the run of bytes equal to $f_i$ that starts at $i$, capped at $127$.
\begin{enumerate}[label=(\roman*)]
\item If $\ell\ge4$, emit the repeat token $(256-\ell)\,\nu(f_i)$ and move to
$i+\ell$.
\item Otherwise let $k$ be the least of $|f|$, $i+127$ and the first position after
$i$ at which four equal bytes begin; emit the literal token covering
$f_i,\dots,f_{k-1}$ and move to $k$.
\end{enumerate}
After the last byte, emit \texttt{END}. Then $\delta(\describe(f))=f$ and
$|\describe(f)|\le|f|+\lceil|f|/127\rceil+1$ for every $f\in\Sigma^*$, with
equality exactly when $f$ contains no four equal consecutive bytes.
\end{lemma}

\begin{proof}
The tokens partition $f$ and each decodes to the bytes it covers, so
$\delta(\describe(f))=f$. Of its $\nlit$ literal and $\nrep$ repeat tokens, a literal
token occupies one byte more than it covers and a repeat token two bytes for at
least four, so $|\describe(f)|\le|f|+\nlit-2\nrep+1$. No run of four equal bytes
begins inside a literal token, since one would have been emitted as a repeat token
and a literal token ends where such a run begins; so $\nrep=0$ exactly when $f$ has
no four equal consecutive bytes, and then every literal token but the last covers
$127$ bytes and the bound holds with equality. Otherwise let $n_{127}$ literal tokens
cover $127$ bytes. A literal token of fewer bytes that is not the last ends where a
run of four begins, so a repeat token follows it, and $\nlit-n_{127}\le\nrep+1$; the
full literal tokens and the repeat tokens cover disjoint bytes, so
$127n_{127}+4\nrep\le|f|$. Hence $\nlit-2\nrep\le n_{127}\le(|f|-4)/127<\lceil
|f|/127\rceil$, and the bound is strict.
\end{proof}

\begin{construction}[The constructor $P$]\label{con:P}
Let $\cell{Z}=\cell{0xF0}$, $\cell{ONE}=\cell{0xF1}$, $\cell{T1}=\cell{0xF2}$,
$\cell{T2}=\cell{0xF3}$ and $\cell{NEG1}=\cell{0xF4}$ be memory cells holding
initially $0$, $1$, $0$, $0$ and $255$. The program $P$ consists of the $21$
instructions of Table~\ref{tab:P}, occupying the bytes $\cell{0x00}$ to
$\cell{0x3E}$. The memory image $m_P$ contains $P$ and these constants and is
zero elsewhere, apart from the cell $\cell{EOK}=\cell{0xF5}$ holding $1$, which
$P$ never reads.
\end{construction}

Since the device clears the request bits at the end of every step, a request
cell holds $0$ on entry, and the instruction $\texttt{NEG1}\ c\ \ell$ assigns
$\mem[c]\leftarrow 0-255=1$ and, the result being positive, passes control to the
next instruction: one instruction sets one request.

\begin{theorem}[Universal construction]\label{thm:univ}
For every $\tau\in\mathcal{R}$ the instance $(\Nhost,m_P,\tau)$ halts and
$\Out(\Nhost,m_P,\tau)=\delta(\tau)$; by Lemma~\ref{lem:complete}, every byte string is
the output of an instance with host $\Nhost$ and memory image $m_P$. Writing $\nlit$
and $\nrep$ for the numbers of literal and repeat tokens of $\tau$, and
$f_{\mathrm{lit}}$ and $f_{\mathrm{rep}}$ for the numbers of bytes each kind
contributes to $\delta(\tau)$, the run takes exactly
$5f_{\mathrm{lit}}+4f_{\mathrm{rep}}+7\nlit+9\nrep+8$ steps.
\end{theorem}

\begin{proof}
Write $\tau=t_1\cdots t_k\texttt{END}$, and let $I_j$ assert that control reaches
$\cell{0x00}$ with the head at the first byte of $t_{j+1}$ and the output word
$\delta(t_1\cdots t_j)$, with $\cell{Z}$, $\cell{ONE}$, $\cell{NEG1}$ and
$R_{\mathrm{out}}$ holding $0$, $1$, $255$ and $0$, all request bits $0$ and the bytes
of $P$ unmodified; $I_0$ holds at the start.

From $I_j$, instructions $0$ to $5$ read the tag $\chi$ of $t_{j+1}$, set $\cell{T1}$
to $\nu(\chi)$ and $\cell{T2}$ to $\chi$, and branch exactly when $\chi=0$ or
$\chi\ge128$. For a literal token, $1\le\chi\le127$, the loop at $\cell{0x1B}$ reads
each payload byte $x$, emits $\nu(x)$ through $R_{\mathrm{out}}$ and the write
request, and counts $\cell{T2}$ down to $0$. For a repeat token, $129\le\chi\le255$,
$\cell{T1}=256-\chi$ is positive, and the loop at $\cell{0x2D}$ reads $x$ once and
emits $\nu(x)$ while it counts $\cell{T1}$ down. Both loops return to $\cell{0x00}$
with $I_{j+1}$. At the end token $\cell{T1}$ is $0$, and at the tag $128$ it is
negative; both reach $\cell{0xFF}$, which halts with the output $\delta(\tau)$. Each
token consumes a tape byte, so the run halts.

For the count, a literal token of length $\chi$ executes instructions $0$ to $6$,
then $\chi$ passes of $9$ to $12$, all but the last also executing $13$, and then
$14$: $5\chi+7$ steps. A repeat token of run length $\ell=256-\chi$ executes $0$ to
$5$, $7$, $8$ and $15$, then $\ell$ passes of $16$ to $18$, all but the last also
executing $19$, and then $14$: $4\ell+9$ steps. The end token executes $0$ to $5$,
$7$ and $20$: eight steps.
\end{proof}

\begin{corollary}[Self-description]\label{cor:selfdesc}
For $\tau_H=\describe(\sigma(\Nhost))$ the instance $(\Nhost,m_P,\tau_H)$ is
self-describing, and its run takes $793{,}709$ steps.
\end{corollary}

\begin{proof}
Theorem~\ref{thm:univ} with $\delta(\tau_H)=\sigma(\Nhost)$ by
Lemma~\ref{lem:describe}; Appendix~\ref{app:sizes} gives the census of $\tau_H$.
\end{proof}

The memory image $m_P$ and the recipe $\tau_H$ do not occur in the output, so this
instance describes itself in the sense of the Turing machines of Lee and Thatcher
that print their own description~\cite{Lee1963,Thatcher1963}. The next program
closes the gap.

\subsection{Self-reproduction}\label{sec:selfrep}

Von Neumann's self-reproducing automaton uses its description twice: once as
instructions for building, and once as data to be copied. The program $P^\ast$
does the same on the host. Phase~(A) decodes a recipe into the string it
describes, phase~(R) rewinds the tape, and phase~(B) copies the recipe itself. Run
on a recipe for the host's own serialization and memory image, it emits everything
from which the instance is rebuilt.

\begin{construction}\label{con:selfrep}
The program $P^\ast$ consists of three phases.
\begin{enumerate}[label=(\Alph*)]
\item \emph{Construction.} The twenty instructions of the decoding loop of $P$,
unchanged, except that on reading the end token control passes to phase~(R)
instead of halting.
\item[(R)] \emph{Rewind.} One instruction $\texttt{NEG1}\ c_{\mathrm{rw}}$.
\item[(B)] \emph{Copying.} Six instructions: set the read request; subtract
$c_{\mathrm{eot}}$ from \cell{EOK}, halting when the result is $\le 0$; assign
$\mem[\cell{Z}]\leftarrow 0-\mem[R_{\mathrm{in}}]$; assign
$\mem[R_{\mathrm{out}}]\leftarrow 0-\mem[\cell{Z}]$; set the write request; and restore
$\mem[\cell{Z}]\leftarrow 0$ while jumping back to the head of the phase.
\end{enumerate}
It has $27$ instructions, occupying $\cell{0x00}$ to $\cell{0x50}$
(Table~\ref{tab:Pstar}). Let $m^\ast\in\Sigma^{256}$ hold $P^\ast$ and the
constants $\mem[\cell{ONE}]=1$, $\mem[\cell{NEG1}]=255$ and $\mem[\cell{EOK}]=1$, and be zero
elsewhere, and let
\[
\tau^\ast=\describe\bigl(\field{\sigma(\Nhost)}\cat\field{m^\ast}\bigr)\in\mathcal{R}.
\]
\end{construction}

\begin{table}[t]\centering\small
\begin{tabular}{@{}rlllll@{}}
\hline
$k$ & addr & $a$ & $b$ & $c$ & effect\\
\hline
0 & \cell{0x00} & \cell{T2} & \cell{T2} & \cell{0x03} & $\mem[\cell{T2}]\leftarrow 0$\\
1 & \cell{0x03} & \cell{T1} & \cell{T1} & \cell{0x06} & $\mem[\cell{T1}]\leftarrow 0$\\
2 & \cell{0x06} & \cell{NEG1} & $c_{\mathrm{rd}}$ & \cell{0x09} & read request\\
3 & \cell{0x09} & $R_{\mathrm{in}}$ & \cell{T1} & \cell{0x0C} & $\mem[\cell{T1}]\leftarrow\nu(\chi)$\\
4 & \cell{0x0C} & \cell{T1} & \cell{T2} & \cell{0x0F} & $\mem[\cell{T2}]\leftarrow \chi$\\
5 & \cell{0x0F} & \cell{Z} & \cell{T2} & \cell{0x15} & branch if $\chi=0$ or $\chi\ge128$\\
6 & \cell{0x12} & \cell{Z} & \cell{Z} & \cell{0x1B} & goto \textsc{lit}\\
7 & \cell{0x15} & \cell{Z} & \cell{T1} & \cell{0x3C} & branch to \textsc{end} if $\nu(\chi)\le0$\\
8 & \cell{0x18} & \cell{Z} & \cell{Z} & \cell{0x2D} & goto \textsc{rep}\\
9 & \cell{0x1B} & \cell{NEG1} & $c_{\mathrm{rd}}$ & \cell{0x1E} & \textsc{lit}: read request\\
10 & \cell{0x1E} & $R_{\mathrm{in}}$ & $R_{\mathrm{out}}$ & \cell{0x21} & $\mem[R_{\mathrm{out}}]\leftarrow\nu(x)$\\
11 & \cell{0x21} & \cell{NEG1} & $c_{\mathrm{wr}}$ & \cell{0x24} & write request\\
12 & \cell{0x24} & \cell{ONE} & \cell{T2} & \cell{0x2A} & $\mem[\cell{T2}]\mathrel{-}=1$; done $\to$ 14\\
13 & \cell{0x27} & \cell{Z} & \cell{Z} & \cell{0x1B} & next literal byte\\
14 & \cell{0x2A} & \cell{Z} & \cell{Z} & \cell{0x00} & goto 0\\
15 & \cell{0x2D} & \cell{NEG1} & $c_{\mathrm{rd}}$ & \cell{0x30} & \textsc{rep}: read request\\
16 & \cell{0x30} & $R_{\mathrm{in}}$ & $R_{\mathrm{out}}$ & \cell{0x33} & $\mem[R_{\mathrm{out}}]\leftarrow\nu(x)$\\
17 & \cell{0x33} & \cell{NEG1} & $c_{\mathrm{wr}}$ & \cell{0x36} & write request\\
18 & \cell{0x36} & \cell{ONE} & \cell{T1} & \cell{0x2A} & $\mem[\cell{T1}]\mathrel{-}=1$; done $\to$ 14\\
19 & \cell{0x39} & \cell{Z} & \cell{Z} & \cell{0x30} & emit again\\
20 ($P$) & \cell{0x3C} & \cell{Z} & \cell{Z} & \cell{0xFF} & \textsc{end}: halt\\
20 ($P^\ast$) & \cell{0x3C} & \cell{NEG1} & $c_{\mathrm{rw}}$ & \cell{0x3F} & \textsc{r}: rewind request\\
21 & \cell{0x3F} & \cell{NEG1} & $c_{\mathrm{rd}}$ & \cell{0x42} & \textsc{b}: read request\\
22 & \cell{0x42} & $c_{\mathrm{eot}}$ & \cell{EOK} & \cell{0xFF} & $\mem[\cell{EOK}]\leftarrow1-c_{\mathrm{eot}}$; halt if $\le0$\\
23 & \cell{0x45} & $R_{\mathrm{in}}$ & \cell{Z} & \cell{0x48} & $\mem[\cell{Z}]\leftarrow\nu(x)$\\
24 & \cell{0x48} & \cell{Z} & $R_{\mathrm{out}}$ & \cell{0x4B} & $\mem[R_{\mathrm{out}}]\leftarrow x$\\
25 & \cell{0x4B} & \cell{NEG1} & $c_{\mathrm{wr}}$ & \cell{0x4E} & write request\\
26 & \cell{0x4E} & \cell{Z} & \cell{Z} & \cell{0x3F} & $\mem[\cell{Z}]\leftarrow0$; goto \textsc{b}\\
\hline
\end{tabular}
\caption{The constructor $P$ of Construction~\ref{con:P}, instructions $0$ to $20$ at
\cell{0x00}--\cell{0x3E}, and the self-reproducing program $P^\ast$ of
Construction~\ref{con:selfrep}, instructions $0$ to $19$ of $P$ followed by $20$ to $26$ at
\cell{0x3C}--\cell{0x50}. An instruction $a\;b\;c$ assigns $\mem[b]\leftarrow \mem[b]-\mem[a]$ and
branches to $c$ when the result is at most $0$; $x$ is the tape byte just read.}
\label{tab:P}\label{tab:Pstar}\end{table}

\begin{theorem}[Exact self-reproduction]\label{thm:selfrep}
For every recipe $\tau$ the instance $(\Nhost,m^\ast,\tau)$ halts and
$\Out(\Nhost,m^\ast,\tau)=\delta(\tau)\cat\tau$; its run reads no byte at
$h=|\tau|$ before it requests the rewind, and requests it at $h=|\tau|$. In
particular $\Omega^\ast=(\Nhost,m^\ast,\tau^\ast)$ is exactly self-reproducing,
$\Out(\Omega^\ast)=\ser(\Omega^\ast)$, and the generations
$\Omega_{k+1}=\inst(\Out(\Omega_k))$ starting from $\Omega^\ast$ are constant.
\end{theorem}

\begin{proof}
Phase~(A) consists of the instructions of $P$ with the same memory layout, and they
do not address the cells $\cell{EOK}$ and $c_{\mathrm{eot}}$, so the proof of
Theorem~\ref{thm:univ} applies. Every token of $\tau$ is consumed in order, and no
byte is read after the end token. When the end token is read, the output word is
$\delta(\tau)$, the head is at $|\tau|$, $\mem[\cell{Z}]=0$, $\mem[R_{\mathrm{out}}]=0$,
$\mem[\cell{EOK}]=1$, and all request bits are $0$. Phase~(R) assigns
$\mem[c_{\mathrm{rw}}]\leftarrow 1$, and the device sets $h=0$.

Let $J_i$ assert that control reaches the head of phase~(B) with head position $i$
and output word $\delta(\tau)\cat\tau_0\cdots\tau_{i-1}$, with $\mem[\cell{Z}]=0$,
$\mem[\cell{EOK}]=1$, $\mem[R_{\mathrm{out}}]=0$ and all request bits $0$. $J_0$ holds.
Assume $J_i$; the first instruction requests a read. If $i<|\tau|$, the device
loads $\tau_i$, clears $c_{\mathrm{eot}}$ and advances the head. The second
instruction then leaves $\cell{EOK}=1$, the third and fourth assign
$\mem[\cell{Z}]\leftarrow\nu(\tau_i)$ and $\mem[R_{\mathrm{out}}]\leftarrow\tau_i$, the
fifth emits $\tau_i$, and the sixth restores $\mem[\cell{Z}]=0$ and branches back, so
$J_{i+1}$ holds. If $i=|\tau|$, the device sets $c_{\mathrm{eot}}=1$, and the second
instruction assigns $\mem[\cell{EOK}]\leftarrow0$ and halts with the output word
$\delta(\tau)\cat\tau$.

For $\tau=\tau^\ast$, $\delta(\tau^\ast)=\field{\sigma(\Nhost)}\cat\field{m^\ast}$
by Lemma~\ref{lem:describe}, so the output is $\ser(\Omega^\ast)$; the last
statement is Lemma~\ref{lem:generations}.
\end{proof}

\begin{proposition}[Running time of $P^\ast$]\label{prop:steps}
For $\tau\in\mathcal{R}$ with the census of Theorem~\ref{thm:univ}, the instance
$(\Nhost,m^\ast,\tau)$ halts after
\[
\underbrace{5f_{\mathrm{lit}}+4f_{\mathrm{rep}}+7\nlit+9\nrep+7}_{\text{phase (A)}}
\;+\;\underbrace{1}_{\text{phase (R)}}
\;+\;\underbrace{6|\tau|+2}_{\text{phase (B)}}
\]
steps. For $\tau=\tau^\ast$ the run takes $1{,}664{,}939$ steps.
\end{proposition}

\begin{proof}
Phase~(A) is the decoding loop of $P$ with its halt instruction replaced by the
rewind. In phase~(B) each of the $|\tau|$ positions executes the six instructions
$21$ to $26$ of Table~\ref{tab:Pstar}, and the last pass executes $21$ and $22$. The
census of $\tau^\ast$ is in Appendix~\ref{app:sizes}.
\end{proof}

\begin{remark}[Wider words]\label{rem:wider}
For $A\ge8$ let $m_{P,A}$ and $m^\ast_A$ hold $P$ and $P^\ast$ in $2^A$ words, the
cells $\cell{0xF0}$ to $\cell{0xFF}$ and the addresses naming them moved to the top
sixteen words, with $\mem[\cell{ONE}]=\mem[\cell{EOK}]=1$ and $\mem[\cell{NEG1}]=2^A-1$, so
that $m_{P,8}=m_P$ and $m^\ast_8=m^\ast$. Theorems~\ref{thm:univ}
and~\ref{thm:selfrep} and Proposition~\ref{prop:steps} hold for $H_A$, $m_{P,A}$ and
$m^\ast_A$ with the same proofs: a byte read with its sign bit copied has the same
value in two's complement on $A$ bits as on eight, and the low byte of
$R_{\mathrm{out}}$ is again $\nu(x)$ or $x$, since $256$ divides $2^A$. On every
recipe the run of $P^\ast$ reads no byte at the end of the tape before it requests
the rewind, and requests the rewind once, at the end of the tape.
\end{remark}

\section{The host reproduces itself by the dynamics}\label{sec:storedhost}

Inside the ring machine, the instantiation step of Section~\ref{sec:instances} is a
write into the next bank followed by a hand-over. This section stores the host, at
every word width from $8$ on, as a network of records in a bank of a ring machine
whose window holds its state bits. A device compiled into the same records writes
the bytes the host emits under the arm until the tape is rewound, and under the copy
head after it. Run on a recipe for its own records, $P^\ast$ then writes the host
into the next bank, copies the recipe onto the next tape and hands control to the
copy (Figure~\ref{fig:generation}), as the organism does with its own code. The
proofs reuse those of Section~\ref{sec:univ}: a pass of the stored host is one step
of the host followed by a device that agrees with the tape transducer on every run
of $P^\ast$.

\begin{construction}[The stored host]\label{con:stored}
Let $A\ge8$ and $w\ge1$. The window holds, after the thirteen cells of
Definition~\ref{def:twohead}, a cell $\mathrm{booted}$ at $13$, a counter
$\mathrm{cnt}$ of $w$ bits from $14$ on, least significant first, a cell
$\mathrm{dest}$, a limit $\mathrm{lim}$ of $w$ bits, the $n=A+1+A2^A$ state bits of
$H_A$ in the order of Construction~\ref{con:host}, a cell for each unit of $H_A$,
three cells for each level of the recursion of Lemma~\ref{cor:compile}, and $2w+9$
cells for the device. The stored host $\mathcal S_{A,w}$ lists the following records.
\begin{enumerate}[label=(\alph*),leftmargin=2em]
\item \emph{Startup}, $A2^A+A+2w+8$ records: for each memory bit $c$,
$c\leftarrow c\vee\neg\mathrm{booted}$ if its value in $m^\ast_A$ is $1$ and
$c\leftarrow c\wedge\mathrm{booted}$ if it is $0$; then
$c\leftarrow c\wedge\mathrm{booted}$ for the bits of the counter, the halt bit and
the bits of $\mathrm{cnt}$, $\mathrm{lim}$ and $\mathrm{dest}$; then five records
copying $\neg\mathrm{booted}$ into RR, CB, RA, CT and RC; then
$\mathrm{booted}\leftarrow1$.
\item \emph{Step}: the units of $H_A$ in the order of the netlist, each compiled by
Lemma~\ref{cor:compile} into its own cell, which also accumulates its chain; the
unit of the next value of a state bit writes that bit instead, with its last
record. A unit compiled by the recursion holds $\Th_t$, $\Th_{t-1}$ and the conjunction
of each level in the three cells of that level.
\item \emph{Device}, $11A+12w+75$ records: the device of Definition~\ref{def:ring}
on the state that~(b) leaves.
\item \emph{Hand-over}: the fewest empty records, at most seven, that make the image
a whole number of bytes; then $\mathrm{booted}\leftarrow\mathrm{booted}\wedge\neg h$,
where $h$ is the halt bit, and three records copying $h$ into RA, RC and HO, the
last record.
\end{enumerate}
The window is the least $2^\alpha$ that holds these cells. The stored host
$\mathcal S=\mathcal S_{8,20}$ has $36{,}080$ records and a window of $2^{15}$ cells,
and $\tau_{\mathcal S}=\describe(\mathcal S)$ is the recipe of its image, its records
in order read as bytes most significant bit first; Appendix~\ref{app:sizes} gives
the layout of its window and the lengths of its image and of $\tau_{\mathcal S}$.
\end{construction}

\begin{definition}[The device of the stored host]\label{def:ring}
The device acts on the state $y$ of $H_A$ after a step, the registers
$\mathrm{cnt}$, $\mathrm{lim}$ and $\mathrm{dest}$, the read tape, the idle bank
and the copy tape, each request firing on the value $1$:
\begin{enumerate}[label=(\roman*)]
\item if $y_{c_{\mathrm{wr}}}=1$, the low eight bits of $y_{R_{\mathrm{out}}}$, most
significant first, are written under the arm if $\mathrm{dest}=0$ and under the
copy head if $\mathrm{dest}=1$, that head moving on by one cell after each, and
$R_{\mathrm{out}}$ is cleared;
\item if $y_{c_{\mathrm{rd}}}=1$, then $c_{\mathrm{eot}}$ is set if
$\mathrm{dest}=1$ and $\mathrm{cnt}=\mathrm{lim}$; otherwise the eight bits under
the read head, the head moving on by one cell after each, are loaded into the low
byte of $R_{\mathrm{in}}$ most significant first, their sign bit is copied into the
bits above them, $c_{\mathrm{eot}}$ is cleared and $\mathrm{cnt}$ is incremented;
\item if $y_{c_{\mathrm{rw}}}=1$, then $\mathrm{lim}\leftarrow\mathrm{cnt}$,
$\mathrm{cnt}\leftarrow0$ and $\mathrm{dest}\leftarrow1$, and the read head returns
to cell $0$;
\item the request cells are cleared.
\end{enumerate}
\end{definition}

\begin{figure}[t]
\centering
\begin{tikzpicture}[font=\footnotesize,>=Stealth,
  s/.style={draw,minimum height=6mm,inner sep=1pt,anchor=west,text height=1.5ex,text depth=0.3ex},
  lab/.style={anchor=east},
  note/.style={font=\scriptsize,text=black!70}]
\node[lab] at (0,0) {active bank $Q_\kappa$};
\node[s,fill=black!10,minimum width=13mm] (a1) at (0.15,0) {startup};
\node[s,fill=black!10,minimum width=35mm] (a2) at (a1.east) {step of $\Nhost$};
\node[s,fill=black!10,minimum width=13mm] (a3) at (a2.east) {device};
\node[s,fill=black!10,minimum width=15mm] (a4) at (a3.east) {hand-over};
\node[note,anchor=south] at (a1.north) {$2104$};
\node[note,anchor=south] at (a2.north) {$33{,}568$ records};
\node[note,anchor=south] at (a3.north) {$403$};
\node[note,anchor=south] at (a4.north) {$1+4$};
\node[lab] at (0,-1.1) {read tape $D_\kappa$};
\node[s,minimum width=76mm] (d0) at (0.15,-1.1) {$\tau_{\mathcal S}=\describe(\mathcal S)$, $240{,}900$ bytes};
\node[lab] at (0,-2.2) {next bank $Q_{\kappa+1}$};
\node[s,densely dashed,minimum width=76mm] (q1) at (0.15,-2.2) {$\mathcal S$, written under the arm in phase~(A)};
\node[lab] at (0,-3.3) {next tape $D_{\kappa+1}$};
\node[s,densely dashed,minimum width=76mm] (d1) at (0.15,-3.3) {$\tau_{\mathcal S}$, written under the copy head in phase~(B)};
\node[draw,rounded corners,align=left,anchor=north west,text width=31mm,minimum height=39mm,inner sep=4pt]
  (win) at (9.75,0.3)
  {\textbf{window}\\[2pt] the host's state, with\\ $P^\ast$ and $m^\ast$ in its\\ memory; the device's\\ $\mathrm{cnt}$, $\mathrm{lim}$, $\mathrm{dest}$;\\ $\mathrm{booted}$; the requests,\\ HO set by the\\ hand-over at the halt};
\draw[->] (a4.east) -- node[note,above] {runs} (win.west|-a4.east);
\draw[->] (d0.east) -- node[note,above] {(A), (B)} (win.west|-d0.east);
\draw[<-] (q1.east) -- node[note,above] {(A)} (win.west|-q1.east);
\draw[<-] (d1.east) -- node[note,above] {(B)} (win.west|-d1.east);
\end{tikzpicture}
\caption{A generation of the stored host. Each pass evaluates the records of $\mathcal S$
and carries out one step of the host on the window. The first pass clears the next
bank and the next tape and returns every head to cell $0$. $P^\ast$ then reads
$\tau_{\mathcal S}$ twice: in phase~(A) it decodes it, and the device writes the decoded bytes,
the image of $\mathcal S$, under the arm; in phase~(B) it copies $\tau_{\mathcal S}$ under the copy head.
At the halt the hand-over passes control to the next bank, whose first pass loads
$m^\ast$ and runs the copy.}
\label{fig:generation}
\end{figure}
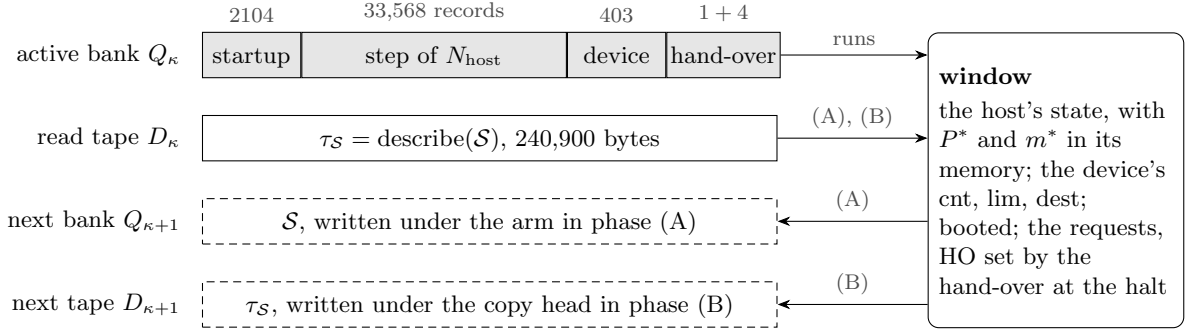

\begin{proposition}[One pass]\label{prop:storedpass}
Let a state of $\mathcal M^\alpha_{K,S,L}$, $S\ge G$, hold $\mathcal S_{A,w}$, of $G$
records, ready in its active bank, and in the window clear requests and a state $x$
of $H_A$. If
$\mathrm{booted}=0$, one pass replaces $x$ by $x_0(m^\ast_A)$, clears
$\mathrm{cnt}$, $\mathrm{lim}$ and $\mathrm{dest}$, returns the read head, the arm
and the copy head to cell $0$ and clears the idle bank and the copy tape. It then
carries out one step of $H_A$ followed by the device of Definition~\ref{def:ring},
and sets $\mathrm{booted}$ to the negation of the halt bit. If the halt bit is $1$,
it returns the arm and the copy head to cell $0$ and hands over at its $G$-th step;
otherwise its $S$ steps bring the slot head back to slot $0$. It leaves
$\mathcal S_{A,w}$ in the bank it started in and the requests clear, whatever the
temporaries held.
\end{proposition}

\begin{proof}
By Lemma~\ref{lem:feedback} the pass evaluates the blocks of
Construction~\ref{con:stored} in order. The startup records leave every cell
unchanged when $\mathrm{booted}=1$. Otherwise they set each memory bit to its value
in $m^\ast_A$, clear the other cells they name, and set the five requests that
return the heads and clear the idle bank and the copy tape. The step block is $H_A$
compiled by Lemma~\ref{cor:compile} in the order of the netlist, which
Lemma~\ref{lem:feedback} admits: every unit follows its predecessors, a unit that
reads a state bit comes no later than the unit of its next value, which reads the bit
before its last record writes it, and no unit reads the unit of a next value. So
after the block the state bits hold the next state of $H_A$.

The device records form the requests $[y_c=1]$ as conjunctions of $A$ literals, the
bit $[\mathrm{cnt}=\mathrm{lim}]$ as a conjunction of $w$ equivalences, and its
conjunction with $\mathrm{dest}$. They then carry out clauses (i) to (iv) of
Definition~\ref{def:ring} in order. Each bit of the low byte of $R_{\mathrm{out}}$
is written under the arm or the copy head by a multiplexer, followed by the request
that moves that head on, and $R_{\mathrm{out}}$ is cleared on a write. Each bit under
the read head is loaded into $R_{\mathrm{in}}$ by a multiplexer on the condition of a
read, followed by the request that moves the read head on, and the sign bit is
copied above them. Then $c_{\mathrm{eot}}$ is set, $\mathrm{cnt}$ is incremented by
carries formed before any of its bits changes, the rewind is carried out, its last
record setting RR, and the requests of the host are cleared.

The hand-over records set $\mathrm{booted}$ and, at the halt, the requests that
return the arm and the copy head and, last, HO. Every temporary is written in the
pass before it is read and is not written after its last read, so a cell shared by
temporaries holds, at each read, the value of the temporary read. No record writes
the active bank, and the empty slots after the records change nothing but the slot
head.
\end{proof}

The theorems below use what Proposition~\ref{prop:storedpass} states and one more
property of the image: every cell read by a record after the startup records, other
than the cells under the heads and the sink, is written earlier in the pass. The
image $\mathcal S_{A,w}$ has this property for every $A\ge8$ and $w\ge1$, and it has
$B=\beta G/8$ bytes, with $\beta=8+3\alpha$. The proofs apply to every image with these two
properties, so the step block may compile any netlist that computes the step of
$H_A$. A state of $\mathcal M^\alpha_{K,S,L}$ \emph{starts} $\mathcal S_{A,w}$
\emph{on} a recipe $\tau$ when its active bank holds $\mathcal S_{A,w}$ ready, its read
tape holds $\tau$ from cell $0$ followed by zeros, and its window holds
$\mathrm{booted}=0$ and clear requests; the other banks, tapes and cells may hold
anything.

\begin{theorem}[Universal construction by the dynamics]\label{thm:storedconstruct}
Let $A\ge8$ and $w\ge1$, let $\mathcal S_{A,w}$ have $G$ records, and let $K\ge2$
and $S\ge G$. Let $\tau$ be a recipe with $|\tau|<2^w$ whose string $X=\delta(\tau)$ has at
most $S\beta/8$ bytes, let $L\ge8|\tau|$, and let $P$ be the number of steps that
Proposition~\ref{prop:steps} gives for $\tau$. From a state that starts
$\mathcal S_{A,w}$ on $\tau$, after $P$ passes, each of $S$ steps but the last, which
has $G$:
\begin{enumerate}[label=(\roman*),leftmargin=2em]
\item the next bank holds the bits of $X$ from slot $0$, followed by zeros, with
both its heads at $0$, and the next tape holds $\tau$ from cell $0$, followed by zeros,
with its head at $0$;
\item the active bank still holds $\mathcal S_{A,w}$ and the read tape $\tau$, the
window holds $\mathrm{booted}=0$ and clear requests, and no other bank or tape has
changed;
\item the bank index has moved to the next bank.
\end{enumerate}
From then on $\mathcal M^\alpha_{K,S,L}$ evaluates the records of $X$.
\end{theorem}

\begin{proof}
By Proposition~\ref{prop:storedpass} the first pass loads $x_0(m^\ast_A)$, returns
the three heads to cell $0$, clears the next bank and the next tape, and carries out
the first step of the host. Each later pass carries out one more step, with the
device of Definition~\ref{def:ring}, up to the pass in which the halt bit is set;
that pass clears $\mathrm{booted}$, returns the arm and the copy head to cell $0$ and
hands over at its last step. By Theorem~\ref{thm:selfrep} and
Remark~\ref{rem:wider}, $P^\ast$ on $\tau$ with the device of
Definition~\ref{def:transducer} halts after the steps of
Proposition~\ref{prop:steps} with output $\delta(\tau)\cat \tau$, reads no byte at
$h=|\tau|$ before it requests the rewind, and requests the rewind once, at $h=|\tau|$.

On this run the two devices agree. As $|\tau|<2^w$, $\mathrm{cnt}$ does not wrap.
Before the rewind $\mathrm{dest}=0$, and each read loads the next byte of $\tau$ from
the read tape and increments $\mathrm{cnt}$, so $\mathrm{cnt}=h$. At the rewind the
read head returns to cell $0$ and $\mathrm{lim}$ becomes $|\tau|$. After it, each read
loads the next byte while $\mathrm{cnt}<\mathrm{lim}$ and sets $c_{\mathrm{eot}}$ at
$\mathrm{cnt}=|\tau|$, as Definition~\ref{def:transducer} does at $h=|\tau|$.

The bytes emitted before the rewind, $\delta(\tau)=X$, are written under the arm from
bit $0$ of the cleared bank, the arm moving on by one bit after each, and they fit
since $8|X|\le S\beta$. Those emitted after it, $\tau$, are written under the copy head
from cell $0$ of the cleared tape, and they fit since $8|\tau|\le L$. The active bank
and the read tape are not written, and no other bank or tape is addressed. From the
next step on, the bank holding $X$ is active, at slot $0$.
\end{proof}

By Lemma~\ref{lem:complete} every byte string has a recipe, so $\mathcal S_{A,w}$
builds every image that fits a bank from a recipe that fits a tape, itself among
them, whatever the size of the image.

\begin{theorem}[Self-reproduction of the host by the dynamics]\label{thm:storedrep}
Let $A\ge8$ and $w\ge1$, let $\mathcal S_{A,w}$ have $G$ records, let $\tau'$ be a
recipe of its image with $|\tau'|<2^w$, let $K\ge2$, $S\ge G$ and $L\ge8|\tau'|$,
and let $x_0$ start $\mathcal S_{A,w}$ on $\tau'$ with $\kappa=0$. A generation is
$(P-1)S+G$ steps, where $P$ is the number of steps that
Proposition~\ref{prop:steps} gives for $\tau'$. For every $g\ge1$:
\begin{enumerate}[label=(\roman*),leftmargin=2em]
\item after $g$ generations $\kappa=g\bmod K$, and the banks $0,\dots,\min(g,K-1)$
of the ring, counted back from bank $\kappa$, hold $\mathcal S_{A,w}$ from slot $0$
and their tapes $\tau'$ from cell $0$: the copies accumulate until every bank holds
one;
\item from the $K$-th generation on, the state repeats every $K$ generations;
\item the network $N^\alpha_{K,S,L}$ started from the bits of $x_0$ holds the bits
of these states at the same steps.
\end{enumerate}
\end{theorem}

\begin{proof}
The first generation is Theorem~\ref{thm:storedconstruct} with $X$ the image of
$\mathcal S_{A,w}$ and $\tau=\tau'$. After it the next bank holds $\mathcal S_{A,w}$ from
slot $0$ with its heads at $0$, the next tape $\tau'$ from cell $0$ followed by
zeros, and the window $\mathrm{booted}=0$ and clear requests, so the state satisfies
the hypothesis of Theorem~\ref{thm:storedconstruct} for the next bank. A generation
changes only the active bank's slot head, the read tape's head, and the next bank
and tape, which it rebuilds; induction over the generations gives $\kappa$ and the
copies.

For the period, let $g\ge K$. Every bank has then been rebuilt and has been active
since: the bank active $a$ generations ago, $1\le a<K$, holds $\mathcal S_{A,w}$
with its slot head where its hand-over pass left it, at slot $G\bmod S$, and its bit
head at $0$, and its tape holds $\tau'$ followed by zeros with the head where the
copying left it, after $\tau'$; the active bank and its tape have all heads at $0$.
So the banks, the tapes and $\kappa$ after generation $g$ depend only on $g\bmod K$.
In the window, the cells no record writes keep their values. The first pass of a
generation loads the same state of the host and the same registers into every
window with $\mathrm{booted}=0$, and the heads and the tapes they read are the same
in every generation; every cell a record after the startup reads is written earlier
in the pass, apart from the cells under the heads and the sink. So, pass by pass,
the windows of any two generations agree on every cell the records read or write,
and after a generation the window is the same in every generation. The state after
generation $g+K$ therefore equals the state after generation $g$, and
Theorem~\ref{thm:physnet} carries the statement to the network.
\end{proof}

For $\mathcal S$ with $K=3$, $S=G+5=36{,}085$ and $L=8|\tau_{\mathcal S}|+64$, a
generation is $2{,}653{,}813$ passes, $95{,}762{,}842{,}100$ applications of
$\mathcal M^{15}_{3,S,L}$; Appendix~\ref{app:sizes} gives the number of passes in closed
form.

\begin{theorem}[Every word width]\label{thm:widths}
For every $A\ge8$ there is a least $w_A\ge20$ for which the image of
$\mathcal S_A=\mathcal S_{A,w_A}$, of $B_A$ bytes, has $2B_A+1<2^{w_A}$; then its
recipe of literal tokens and $\describe(\mathcal S_A)$ are shorter than $2^{w_A}$.
The stored host $\mathcal S_A$ builds every image that fits a bank and reproduces
itself by the dynamics of $N^{\alpha_A}_{K,S,L}$ as
Theorems~\ref{thm:storedconstruct} and~\ref{thm:storedrep} state. It has
$\Theta(A2^A)$ records and a window of $2^{\alpha_A}$ cells with
$\alpha_A=A+\Theta(\log A)$, and in banks of $G_A$ slots, $G_A$ its number of
records, a generation takes $\Theta(A^34^A)$ applications of
$\mathcal M^{\alpha_A}_{K,G_A,L}$.
\end{theorem}

\begin{proof}
The records of $\mathcal S_{A,w}$ grow linearly in $w$ and its window
logarithmically, so $B$ grows at most quadratically in $w$ while $2^w$ grows
exponentially, and $w_A$ exists. The recipe of literal tokens of Lemma~\ref{lem:complete} has
$B+\lceil B/127\rceil+1\le2B+1$ bytes, and $\describe$ gives no more by
Lemma~\ref{lem:describe}. The window holds the $A2^A+A+1$ state bits, a cell for each
of the $\Theta(A2^A)$ units of $H_A$ and $O(A+w_A)$ further cells, so
$\alpha_A=A+\Theta(\log A)$; then $\beta=8+3\alpha_A=\Theta(A)$,
$B=\beta G_A/8=\Theta(A^22^A)$, and Theorem~\ref{thm:storedrep} gives a generation of
$\Theta(B)$ passes of $G_A$ steps. Appendix~\ref{app:sizes} gives $G_A$, $w_A$ and
the bounds in closed form.
\end{proof}

\section{Reproduction under noise}\label{sec:robust}

The reproductions of Sections~\ref{sec:organism} to~\ref{sec:storedhost} are exact
because every pre-activation is an integer: a run rebuilds the weights bit for bit
while no perturbation of a pre-activation reaches one half. This section asks how
much noise a reproducing run tolerates. It proves this margin of one half, bounds the
probability of leaving the orbit under random perturbations from the
pre-activations the orbit visits, and compares the bounds with the host under noise
(Section~\ref{sec:certified}). It shows that replicating every unit multiplies the
margin, which keeps a run exact under three models of noise, below a threshold in
one of them (Section~\ref{sec:amplify}), and it gives the noise under which a
generation of the stored host, which rebuilds the host's weights in the next bank,
runs exactly (Section~\ref{sec:storednoise}).

The probabilistic theorems are proved in Rocq on a probability space given by its
properties (Appendix~\ref{sec:verification}), and they assume independent
perturbations. Noise that is correlated between units or drifts over time is
covered only by Theorem~\ref{thm:margin}, Proposition~\ref{prop:mismatch} and
Theorem~\ref{thm:amplify}(i), which hold for every perturbation below the margin,
whatever its law. The union bounds, thresholds and redundancies are evaluated in
Rocq; the profile and restoration bounds along particular orbits are computed by
the artifact.

Let $\Hv_{1/2}(z)=1$ if $z\ge -\tfrac12$ and $\Hv_{1/2}(z)=0$ otherwise; on integer
arguments $\Hv_{1/2}$ and $\Hv$ agree. A \emph{layered integer threshold map} is
defined as in Definition~\ref{def:lttm} with $W_\ell\in\Z^{n_\ell\times
n_{\ell-1}}$; the results of this section hold in that generality.

\begin{definition}[Perturbed evaluation]\label{def:perturbed}
Let $F$ be a layered integer threshold map with parameters
$(W_\ell,\theta_\ell)_{\ell\le d}$ and let $\varepsilon\ge 0$. An
\emph{$\varepsilon$-perturbed trajectory} from $x_0$ is a sequence $(\tilde x_t)$
with $\tilde x_0=x_0$ and
$\tilde x_{t+1}=\tilde\Phi_{d,t}\circ\dots\circ\tilde\Phi_{1,t}(\tilde x_t)$, where
$\tilde\Phi_{\ell,t}(y)=\Hv_{1/2}(W_\ell y+\theta_\ell+e_{\ell,t})$ and each
$e_{\ell,t}\in\R^{n_\ell}$ satisfies $\|e_{\ell,t}\|_\infty\le\varepsilon$. The
perturbations may depend on $\ell$, $t$, the current input, and the past.
\end{definition}

For a transducer the device of Definition~\ref{def:transducer} acts on the
perturbed trajectory as on the exact one, between the steps; it is a function
of the state, so every statement of this section holds for the sequence of
states a run visits as it holds for an orbit of $F$.

\begin{theorem}[Integer margin]\label{thm:margin}
If $\varepsilon<\tfrac12$, every $\varepsilon$-perturbed trajectory from $x_0$
satisfies $\tilde x_t=F^t(x_0)$ for all $t\ge 0$. If some unit has
pre-activation $-1$ along the orbit, a perturbation of $+\tfrac12$ at that unit
changes its output, so the constant cannot be increased.
\end{theorem}

\begin{proof}
Suppose the input $y$ of layer $\ell$ at time $t$ equals the unperturbed input.
Then $z=W_\ell y+\theta_\ell\in\Z^{n_\ell}$. If $z_i\ge 0$ then
$z_i+e_i\ge-\varepsilon>-\tfrac12$, and if $z_i\le -1$ then
$z_i+e_i\le -1+\varepsilon<-\tfrac12$, so $\Hv_{1/2}(z_i+e_i)=\Hv(z_i)$ and the
layer's output is the unperturbed one. Induction over layers and time gives the
first statement, and $\Hv_{1/2}(-1+\tfrac12)=1\neq 0=\Hv(-1)$ the second.
\end{proof}

Along the first $4000$ steps of the self-describing instance an average of
$26{,}825$ rows of $\Lev(\Nhost)$ per step sit at pre-activation $-1$ and $4{,}066$
at $0$, so the
constant $\tfrac12$ is attained there.

\begin{proposition}[Weight and bias errors]\label{prop:mismatch}
Suppose each layer is realized with weights $W_\ell+\Delta_\ell$ and biases
$\theta_\ell+\Gamma_\ell$, where $|(\Delta_\ell)_{ij}|\le\varepsilon$ if
$(W_\ell)_{ij}\neq 0$, $|(\Delta_\ell)_{ij}|\le\gamma$ if $(W_\ell)_{ij}=0$, and
$|(\Gamma_\ell)_i|\le\eta$, with thresholding by $\Hv_{1/2}$.
\begin{enumerate}[label=(\roman*),leftmargin=2em]
\item If $\varepsilon k_{\ell,i}+\gamma(n_{\ell-1}-k_{\ell,i})+\eta<\tfrac12$ for
every layer $\ell$ and row $i$, where $k_{\ell,i}$ is the effective fan-in of the
row, every orbit of the realized system coincides with the orbit of $F$.
\item Let $(x_t)_{t\le T}$ be the states visited by a run of $F$ from $x_0$, and
let $K^\ast$ and $A^\ast$ be the largest numbers of inputs equal to $1$ that a row
reads at the nonzero and at the zero entries of its layer along the run. If
$\varepsilon K^\ast+\gamma A^\ast+\eta<\tfrac12$, the realized trajectory from
$x_0$ visits $x_0,\dots,x_T$.
\end{enumerate}
\end{proposition}

\begin{proof}
On an input $y\in\B^{n_{\ell-1}}$,
$|(\Delta_\ell y)_i+(\Gamma_\ell)_i|\le\sum_{j:y_j=1}|(\Delta_\ell)_{ij}|+\eta$,
which is at most the left side of~(i), and of~(ii) while the trajectories agree;
the induction of Theorem~\ref{thm:margin} applies.
\end{proof}

The term $\gamma(n_{\ell-1}-k_{\ell,i})$, the leakage of the entries the design
sets to zero, grows with the width of the layer; along an orbit only the inputs
equal to $1$ leak, and few bits of a SUBLEQ state are set at any time.

\begin{proposition}[Independent perturbations]\label{prop:random}
Suppose each pre-activation evaluation is perturbed by an independent random
variable $e$ with $\Pr(|e|\ge\tfrac12)\le p$, and let $M=\sum_\ell n_\ell$. The
probability that the perturbed trajectory differs from the orbit of $F$ at some
time $t\le T$ is at most $TMp$. For centred Gaussian perturbations of standard
deviation $s$, $p=\erfc\bigl(1/(2\sqrt2\,s)\bigr)$.
\end{proposition}

\begin{proof}
By the argument of Theorem~\ref{thm:margin}, the first deviation occurs at an
evaluation whose input is unperturbed and whose perturbation satisfies
$|e|\ge\tfrac12$. There are $TM$ evaluations up to time $T$.
\end{proof}

\begin{proposition}[Profile bound]\label{prop:profile}
Let the perturbations be independent centred Gaussians of standard deviation
$s$, let $(x_t)$ be the orbit of $F$ from $x_0$, and for $t\ge0$ and each row $i$
of each layer let $z_{t,i}$ be the pre-activation of row $i$ on the input
$x_t$. Put
\begin{gather*}
b'(t)=1-\prod_i\Bigl(1-Q\bigl(|z_{t,i}+\tfrac12|/s\bigr)\Bigr),\qquad
b(t)=\sum_i Q\bigl(|z_{t,i}+\tfrac12|/s\bigr),\\
Q(u)=\Pr(\mathcal N(0,1)>u)=\tfrac12\erfc(u/\sqrt2).
\end{gather*}
Given that a perturbed trajectory from $x_0$ agrees with the orbit up to step $t$,
the probability that it differs from it at step $t+1$ is at most $b'(t)$, the
probability that some row of the step fails while it reads the values of the orbit;
and $b'(t)\le b(t)\le\tfrac12Mp$. If copies of the trajectory are run from $x_0$,
each stopped at its first deviation or after $T$ steps, the expected number of
copies that deviate is at most the expectation of $\sum_c\sum_{t<e_c}b'(t)$, where
$e_c$ is the number of steps copy $c$ runs.
\end{proposition}

\begin{proof}
Condition on agreement up to step $t$, so the input of the first layer at step
$t+1$ is $x_t$. As in Theorem~\ref{thm:margin}, the output of the step differs
from $x_{t+1}$ only if some row, reading an unperturbed input, fails, that is, has
its perturbed pre-activation on the far side of $-\tfrac12$: for $z_{t,i}\ge0$ this
is the event $e<-(z_{t,i}+\tfrac12)$, and for $z_{t,i}\le-1$ the event
$e\ge-\tfrac12-z_{t,i}$, each of probability $Q(|z_{t,i}+\tfrac12|/s)$. While no row has
failed, every row of the next layer reads the values of the orbit, and the
perturbations of its rows are independent; so no row of the step fails with
probability $\prod_i(1-Q(|z_{t,i}+\tfrac12|/s))$, layer by layer, which gives $b'(t)$. The
union bound gives $b'(t)\le b(t)$, and $|z_{t,i}+\tfrac12|\ge\tfrac12$ gives
$Q(|z_{t,i}+\tfrac12|/s)\le Q(1/(2s))=p/2$. For the last statement, copy $c$ deviates at
step $t+1$ with probability at most $b'(t)$ given that it has run $t$ steps
without deviating, so its probability of deviating is at most
$\sum_t\Pr(e_c>t)\,b'(t)=\mathbb E\bigl[\sum_{t<e_c}b'(t)\bigr]$.
\end{proof}

\subsection{Tolerances of the host}\label{sec:certified}

For $\Lev(\Nhost)$, with $M=41{,}967$ rows, the largest effective fan-in is $256$,
attained by the $40$ units of the $256$-way disjunctions of the operand fetch and
the read multiplexer, and the widest layer has $8210$ rows.
Proposition~\ref{prop:mismatch}(i) therefore makes every orbit exact for
$\varepsilon<1/512$ when nothing leaks, and for $\gamma<6.09\times10^{-5}$ when the
weights are exact. Along the first $10{,}000$ steps of the self-describing instance
$K^\ast=8$ and $A^\ast=349$, so (ii) certifies $\gamma<1.43\times10^{-3}$ and
$\varepsilon<1/16$, each alone; along the self-reproducing run $A^\ast=441$, which
certifies $\gamma<1.13\times10^{-3}$. Under uniform random errors the first
$10{,}000$ steps stayed exact up to $\gamma=3.9\times10^{-3}$ and up to
$\varepsilon=1/16$.

Under Gaussian read noise of deviation $s=0.05$, $p=1.52\times10^{-23}$. Over the
$T=1{,}664{,}939$ steps of the self-reproducing run $TM=6.99\times10^{10}$, and
Proposition~\ref{prop:random} bounds the probability of a deviation by
$1.1\times10^{-12}$; at $s=0.10$ this bound exceeds one. Along the first $4000$
steps of the self-describing instance, $73.61\%$ of the rows sit at distance
$\tfrac12$ from the comparator at every step and the rest at distance at least
$\tfrac32$, so the sum $b$ of Proposition~\ref{prop:profile} lies a factor $2.7$
below $Mp$ wherever it is below one.

With $256$ noisy copies stepped beside the exact orbit, no copy left it in
$768{,}000$ step--copy pairs at $s\le0.06$, and failures were observed from
$s=0.08$ on. The product form $b'(t)$ is the exact probability that some row of
step $t$ fails while the state is on the orbit, and a failed row need not change
the state: with noisy copies restarted from the exact state at every step, the
steps with a failing row matched $b'$ within $1.3$ standard deviations at every
level, and $4.9$ to $6.3\%$ of them ended on the orbit at $s=0.09$ to $0.12$
(Appendix~\ref{app:measure}). So $b'$ exceeds the rate at which the host leaves its
orbit by that fraction, and at the sample size of Table~\ref{tab:curve} it lies
within the exact $95\%$ Poisson interval of the measured rate at every level; the
sum $b$ does so up to $s=0.11$ and lies above it from $s=0.12$ on, where several
rows of a step fail together (Figure~\ref{fig:noise}(a)).

\begin{figure}[t]
\centering
\begin{tikzpicture}
\begin{semilogyaxis}[name=a,width=7.4cm,height=5.0cm,font=\footnotesize,
  xlabel={noise deviation $s$},ylabel={per-step probability},
  xmin=0.055,xmax=0.155,ymin=1e-13,ymax=1e2,
  xtick={0.06,0.08,0.10,0.12,0.14},
  xticklabel style={/pgf/number format/fixed,/pgf/number format/fixed zerofill,/pgf/number format/precision=2},
  ytick={1e-12,1e-9,1e-6,1e-3,1},
  legend style={font=\scriptsize,at={(0.97,0.03)},anchor=south east,draw=none,fill=white,fill opacity=0.8,text opacity=1},
  legend cell align=left,title={(a) $\Lev(\Nhost)$},title style={font=\footnotesize}]
\addplot[densely dashed,thick,black!55] coordinates {(0.06,3.30e-12) (0.08,1.72e-5) (0.09,1.16e-3) (0.10,2.41e-2) (0.11,2.30e-1) (0.12,1.30) (0.15,3.60e1)};
\addlegendentry{$Mp$}
\addplot[densely dotted,thick,black] coordinates {(0.06,1.21e-12) (0.08,6.34e-6) (0.09,4.27e-4) (0.10,8.85e-3) (0.11,8.42e-2) (0.12,4.73e-1) (0.15,1.30e1)};
\addlegendentry{profile sum $b$}
\addplot[thick,black] coordinates {(0.06,1.21e-12) (0.08,6.34e-6) (0.09,4.27e-4) (0.10,8.81e-3) (0.11,8.07e-2) (0.12,3.77e-1) (0.15,1.00)};
\addlegendentry{profile product $b'$}
\addplot[only marks,mark=*,mark size=1.6pt,error bars/.cd,y dir=both,y explicit] coordinates {
  (0.08,6.94e-6) += (0,7.36e-6) -= (0,4.15e-6)
  (0.09,4.07e-4) += (0,6.01e-5) -= (0,5.43e-5)
  (0.10,8.71e-3) += (0,1.14e-3) -= (0,1.03e-3)
  (0.11,7.70e-2) += (0,1.00e-2) -= (0,9.14e-3)
  (0.12,3.59e-1) += (0,4.68e-2) -= (0,4.25e-2)
  (0.15,1.00) += (0,0) -= (0,1.19e-1)};
\addlegendentry{measured, $95\%$}
\end{semilogyaxis}
\begin{semilogyaxis}[at={(a.east)},anchor=west,xshift=1.0cm,width=7.4cm,height=5.0cm,font=\footnotesize,
  xlabel={noise deviation $s$},
  xmin=0.19,xmax=0.37,ymin=1e-13,ymax=1e2,
  xtick={0.20,0.25,0.30,0.35},
  xticklabel style={/pgf/number format/fixed,/pgf/number format/fixed zerofill,/pgf/number format/precision=2},
  ytick={1e-12,1e-9,1e-6,1e-3,1},
  legend style={font=\scriptsize,at={(0.97,0.03)},anchor=south east,draw=none,fill=white,fill opacity=0.8,text opacity=1},
  legend cell align=left,title={(b) $N_3$, read by majority},title style={font=\footnotesize}]
\addplot[densely dashed,thick,black!55] coordinates {(0.20,8.03e-9) (0.22,1.16e-6) (0.25,2.48e-4) (0.27,3.48e-3) (0.30,7.22e-2) (0.33,6.90e-1) (0.36,3.89)};
\addlegendentry{first order}
\addplot[thick,black] coordinates {(0.20,3.59e-13) (0.22,1.82e-10) (0.25,1.58e-7) (0.27,4.57e-6) (0.30,2.26e-4) (0.33,4.34e-3) (0.36,4.30e-2)};
\addlegendentry{restoration bound}
\addplot[only marks,mark=*,mark size=1.6pt,error bars/.cd,y dir=both,y explicit] coordinates {
  (0.27,4.04e-6) += (0,5.39e-6) -= (0,2.73e-6)
  (0.30,2.13e-4) += (0,5.97e-5) -= (0,4.94e-5)
  (0.33,3.66e-3) += (0,1.02e-3) -= (0,8.42e-4)
  (0.36,2.93e-2) += (0,8.10e-3) -= (0,6.72e-3)};
\addlegendentry{measured, $95\%$}
\addplot[only marks,mark=o,mark size=1.8pt] coordinates {(0.24,9.79e-7) (0.27,8.47e-5) (0.30,2.46e-3) (0.33,2.58e-2)};
\addlegendentry{a copy deviates}
\end{semilogyaxis}
\end{tikzpicture}
\caption{The per-step probability of leaving the orbit under Gaussian read noise
of deviation $s$, against the bounds of Section~\ref{sec:robust}; the lines join
the values computed at the deviations of Tables~\ref{tab:curve}
and~\ref{tab:bound3}, and the error bars are exact $95\%$ intervals for a Poisson
count. (a)~$\Lev(\Nhost)$: the rate at which noisy copies first leave the exact
orbit, the product form $b'$ and the sum $b$ of the profile bound of
Proposition~\ref{prop:profile}, and the bound $Mp$ of
Proposition~\ref{prop:random}. (b)~$N_3$: the rate at which the state read by
majority first leaves the orbit, the restoration bound of
Proposition~\ref{prop:restore} and the first-order bound of
Theorem~\ref{thm:amplify}(ii), with the rate at which some copy of some state bit
first deviates.}
\label{fig:noise}
\end{figure}
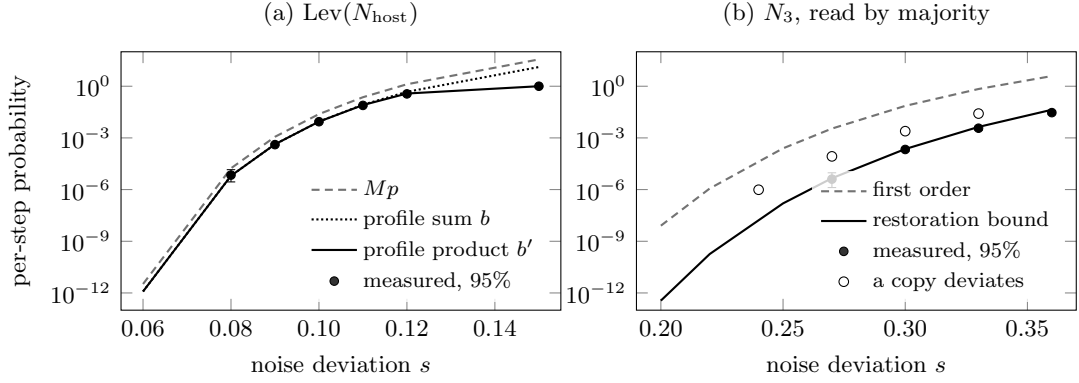

\subsection{Fault tolerance by amplification}\label{sec:amplify}

Theorem~\ref{thm:margin} tolerates perturbations below $\tfrac12$ because every
pre-activation is an integer. Replicating the state and letting every unit read
every copy multiplies the margin. Let $\rho_r\colon\B^n\to\B^{rn}$ be the
replication $\rho_r(x)=(x,\dots,x)$, copy $c$ occupying the coordinates
$cn,\dots,cn+n-1$, let $J_r$ be the all-ones $r\times r$ matrix and
$\mathbf 1_r$ the all-ones vector.

\begin{definition}[Amplification]\label{def:amplify}
For an LTTM $F$ with parameters $(W_\ell,\theta_\ell)_{\ell\le d}$ and an odd
$r\ge1$, the \emph{$r$-fold amplification} $F^{(r)}$ is the LTTM on $\B^{rn}$
with parameters
\[
W^{(r)}_\ell=J_r\otimes W_\ell\in\Tern^{rn_\ell\times rn_{\ell-1}},\qquad
\theta^{(r)}_\ell=\mathbf 1_r\otimes\bigl(r\,\theta_\ell+\tfrac{r-1}{2}\bigr)\in\Z^{rn_\ell}.
\]
Row $(c,i)$ of layer $\ell$ reads every copy of input $j$ with the weight
$(W_\ell)_{ij}$, so its effective fan-in is $r$ times that of row $i$; $F^{(r)}$
has $rM$ rows and $r^2$ times the nonzero entries of $F$, and $F^{(1)}=F$.
\end{definition}

\begin{proposition}[Amplified margin]\label{prop:amplify}
Let $F$ be an LTTM and $r$ odd.
\begin{enumerate}[label=(\roman*),leftmargin=2em]
\item $F^{(r)}\circ\rho_r=\rho_r\circ F$, and the same holds for each layer.
\item On the input $\rho_r(y)$, the pre-activation of row $(c,i)$ of layer
$\ell$ is $rz+\tfrac{r-1}{2}$, where $z\in\Z$ is the pre-activation of row $i$
of $\Phi_\ell$ on $y$, and its distance from $-\tfrac12$ is $r\,|z+\tfrac12|\ge
r/2$.
\item If the input differs from $\rho_r(y)$ in $f$ coordinates, every
pre-activation differs from the value in~(ii) by at most $f$; a row whose inputs
carry $f<r/2$ wrong copies therefore computes its value on $\rho_r(y)$ under
every perturbation of magnitude less than $r/2-f$.
\end{enumerate}
\end{proposition}

\begin{proof}
$(J_r\otimes W_\ell)\rho_r(y)=\mathbf 1_r\otimes rW_\ell y$, so row $(c,i)$ has
pre-activation $r(W_\ell y)_i+r(\theta_\ell)_i+\tfrac{r-1}{2}=rz+\tfrac{r-1}{2}$. If
$z\ge0$ this is at least $\tfrac{r-1}{2}\ge0$, the row outputs $1=\Hv(z)$, and
the distance from $-\tfrac12$ is $rz+\tfrac r2=r(z+\tfrac12)$; if $z\le-1$ it is
at most $-\tfrac{r+1}{2}<0$, the row outputs $0=\Hv(z)$, and the distance is
$-(rz+\tfrac r2)=r|z+\tfrac12|$. This is~(ii), and~(i) follows layer by layer.
For~(iii), every weight has magnitude at most $1$, so each wrong input moves the
pre-activation by at most $1$.
\end{proof}

\begin{theorem}[Fault tolerance by amplification]\label{thm:amplify}
Let $F$ be an LTTM with $M$ rows, $r$ odd and $x_0\in\B^n$.
\begin{enumerate}[label=(\roman*),leftmargin=2em]
\item Every $\varepsilon$-perturbed trajectory of $F^{(r)}$ from $\rho_r(x_0)$
with $\varepsilon<r/2$ equals $\bigl(\rho_r(F^t(x_0))\bigr)_{t\ge0}$.
\item Let the perturbations be independent with $\Pr(|e|\ge u)\le\varphi(u)$
for a non-increasing $\varphi$. The probability that the perturbed trajectory of
$F^{(r)}$ from $\rho_r(x_0)$ differs from $\rho_r(F^t(x_0))$ at some $t\le T$
is at most $T\,rM\,\varphi(r/2)$. For centred Gaussian perturbations of standard
deviation $s$ it is at most $T\,rM\erfc\bigl(r/(2\sqrt2\,s)\bigr)$, and at most
$r\sum_{t<T}\sum_iQ(r|z_{t,i}+\tfrac12|/s)$ in terms of the pre-activations
$z_{t,i}$ of $F$ of Proposition~\ref{prop:profile}.
\item Let $\delta\in(0,1)$. Under Gaussian perturbations the trajectory follows
the orbit for $T$ steps with probability at least $1-\delta$ whenever
$r\ge2\sqrt2\,s\sqrt{\ln(TrM/\delta)}$, so $r=O\bigl(s\sqrt{\log(TM/\delta)}\bigr)$
suffices. When $\varphi(u)\le e^{-u/\lambda}$ it does so whenever
$r\ge2\lambda\ln(TrM/\delta)$: the bound $p=e^{-1/(2\lambda)}$ at margin
$\tfrac12$ becomes $p^{\,r}$.
\end{enumerate}
\end{theorem}

\begin{proof}
By Proposition~\ref{prop:amplify}(ii), on an input equal to the replicated
state of the orbit every row of $F^{(r)}$ lies at distance at least $r/2$ from
the comparator, and the argument of Theorem~\ref{thm:margin} gives~(i). For~(ii),
the first step at which the perturbed trajectory differs from the orbit, if
there is one, reads an exact input, and a row of that step changes its output
only if its perturbation has magnitude at least its distance $r|z+\tfrac12|$
from the comparator, on the far side: the probability is at most
$\varphi(r/2)$, and for Gaussian perturbations exactly $Q(r|z+\tfrac12|/s)$,
which is at most $\tfrac12\erfc(r/(2\sqrt2\,s))$. The union bound over the $rM$
rows of a step and the $T$ steps gives the three bounds, the last summing the
exact probabilities over the profile, each row of $F$ contributing $r$ rows of
$F^{(r)}$. For~(iii), $\erfc(u)\le e^{-u^2}$ for $u\ge0$, so
$TrM\erfc(r/(2\sqrt2\,s))\le TrM\,e^{-r^2/(8s^2)}\le\delta$ when
$r^2\ge8s^2\ln(TrM/\delta)$; and $TrM\,e^{-r/(2\lambda)}\le\delta$ when
$r\ge2\lambda\ln(TrM/\delta)$, where $e^{-r/(2\lambda)}=(e^{-1/(2\lambda)})^r$.
\end{proof}

In (ii) and (iii) the perturbation of a row has the same distribution whatever the
row's fan-in, as for noise added at the input of a comparator, and every copy adds
margin against it. When the noise sits on every synapse instead, the error of a row
grows with its active inputs, and a copy buys less.

\begin{corollary}[Noise on every synapse]\label{cor:synapse}
Suppose that at every step each programmed entry of every layer of $F^{(r)}$ is
read with an independent centred Gaussian error of standard deviation $s_{\mathrm w}$,
the biases exactly, and let $K^\ast$ be the largest number of active programmed
inputs of a row of $F$ along the orbit of $x_0$ up to time $T$, as in
Proposition~\ref{prop:mismatch}(ii). The probability that the trajectory of $F^{(r)}$
from $\rho_r(x_0)$ differs from $\rho_r(F^t(x_0))$ at some $t\le T$ is at most
$T\,rM\erfc\bigl(\sqrt r/(2\sqrt{2K^\ast}\,s_{\mathrm w})\bigr)$, which is at most
$\delta$ when $r\ge8s_{\mathrm w}^2K^\ast\ln(TrM/\delta)$.
\end{corollary}

\begin{proof}
At the first step at which the trajectory differs from the orbit the input of
every layer is exact, and on $\rho_r(y)$ the error of row $(c,i)$ is the sum of the
errors at the $r$ copies of each of the at most $K^\ast$ active programmed inputs
of row $i$ of $F$, a centred Gaussian of deviation at most
$s_{\mathrm w}\sqrt{rK^\ast}$. Its distance from the comparator is at least $r/2$
(Proposition~\ref{prop:amplify}(ii)); the union bound over the $rM$ rows and the
$T$ steps, and $\erfc(u)\le e^{-u^2}$, give the two statements.
\end{proof}

A copy thus divides the effective noise by $r$ on the pre-activation and by
$\sqrt r$ on the synapses. Along the self-reproducing run of the host $K^\ast=8$,
and the bound of Corollary~\ref{cor:synapse} falls below $10^{-3}$ with $r=5$ at
$s_{\mathrm w}=0.05$ and with $r=21$ at $s_{\mathrm w}=0.10$.

In von Neumann's model of unreliable components~\cite{vonNeumann1956} the error
does not sit on a sum but on the output: every unit outputs the complement of its
value with a fixed probability, whatever its margin. Amplification then helps only
below a threshold.

\begin{proposition}[Output flips]\label{prop:flip}
Suppose that every evaluation of every row of $F^{(r)}$ outputs the complement of
its value with probability $\varepsilon$, independently, and let $k^\ast$ be the
largest effective fan-in of a row of $F$ and $n$ its number of state bits.
\begin{enumerate}[label=(\roman*),leftmargin=2em]
\item Started from $\rho_r(x_0)$, the probability that the state of $F^{(r)}$ read by
majority differs from $F^t(x_0)$ after one of the first $T$ steps is at most
$T(M+n)(2ek^\ast\varepsilon)^{(r+1)/2}$. For $\varepsilon<1/(2ek^\ast)$ it is at
most $\delta$ when $r+1\ge2\ln(T(M+n)/\delta)/\ln\bigl(1/(2ek^\ast\varepsilon)\bigr)$.
\item Let a row of $F$ be the disjunction of $k$ inputs that hold $0$ on the orbit,
and let $2k\varepsilon>1$. There is $\zeta<1$, depending on $k$ and $\varepsilon$
only, such that a row of $F^{(r)}$ that copies it, reading $kr$ input copies each
flipped with probability $\varepsilon$, independently, computes a wrong value with
probability at least $1-\zeta^{(r-1)/2}$.
\end{enumerate}
\end{proposition}

\begin{proof}
All $r$ copies of a row read the same copies of its inputs, so they compute the
same value before their outputs flip; call the row wrong in a step when that value
differs from the orbit's. Consider the first wrong row of the run. Every row it
reads is right, so each of its at most $k^\ast r$ input copies is wrong exactly when
its own output flipped, independently with probability $\varepsilon$, and by
Proposition~\ref{prop:amplify}(ii) and~(iii) the row is wrong only if at least
$(r+1)/2$ of them are. This has probability at most
$\binom{k^\ast r}{(r+1)/2}\varepsilon^{(r+1)/2}\le(2ek^\ast\varepsilon)^{(r+1)/2}$,
since $\binom mj\le(em/j)^j$. While no row is wrong, the majority of the $r$ copies
of a state bit is wrong only if at least $(r+1)/2$ of them flipped, with probability
at most $(2e\varepsilon)^{(r+1)/2}$. The union bound over the $M$ rows and the $n$
state bits of each of the $T$ steps gives the bound of~(i), and taking logarithms
gives the redundancy.

For~(ii), the copy has bias $-r+\tfrac{r-1}{2}=-\tfrac{r+1}{2}$ and reads its $kr$
input copies with weight $1$, so its pre-activation is the number of flipped copies
minus $\tfrac{r+1}{2}$, and it is wrong exactly when at least $\tfrac{r+1}{2}$ of
them flipped. For $0<x\le1$ and $c=1-\varepsilon+\varepsilon x$, fewer than $t\ge1$
of $m$ independent copies flip with probability at most $c^m/x^{t-1}$: by induction
on $m$, either the first copy flips and fewer than $t-1$ of the others do, or it does
not and fewer than $t$ do, which has probability at most
$\varepsilon c^{m-1}/x^{t-2}+(1-\varepsilon)c^{m-1}/x^{t-1}=c^m/x^{t-1}$, the first
term absent when $t=1$. For $r=2h+1$, $m=kr$ and $t=h+1$ this is
$c^k(c^{2k}/x)^h\le\zeta^h$ with $\zeta=c^{2k}/x$. Take
$u=(2k\varepsilon-1)/(4k\varepsilon)$, which lies in $(0,\tfrac12)$, and $x=1-u$, so
that $c=1-\varepsilon u$. Then $c^{2k}(1+\varepsilon u)^{2k}\le1$,
$(1+\varepsilon u)^{2k}\ge1+2k\varepsilon u$ and
$(1+2k\varepsilon u)(1-u)=(2k\varepsilon+1)^2/(8k\varepsilon)>1$, so $\zeta<1$.
\end{proof}

The threshold belongs to amplification, which keeps the fan-in of $F$ and adds no
restoring stages, and not to computation with noisy threshold units: formulas of
noisy three-input gates compute reliably for every
$\varepsilon<1/6$~\cite{HajekWeller1991}, and gates of fan-in $k$ cannot do so above
$1/2-1/(2\sqrt k)$~\cite{EvansSchulman1999}. In the host it comes from $48$
disjunctions, of fan-in $256$ in the operand fetch and the operand reads and of
fan-in $128$ in the next counter, whose inputs all hold $0$ when the bit they
compute is $0$. For $\Lev(\Nhost)$, with $k^\ast=256$, the amplification is
therefore reliable for $\varepsilon<7.2\times10^{-4}$, and for $\varepsilon>1/512$
such a disjunction is wrong with a probability that tends to one as $r$ grows. Every
other unit of the host has fan-in at most $8$. Split into trees of fan-in $8$, each
of the $48$ disjunctions would take three levels instead of one. The step would then
have depth $17$, since the fetch and the reads lie on its longest path, and
$1/(2ek^\ast)$ would rise to $0.023$.

Integer weights reach the same margin without copies: the layer with weights
$rW_\ell$ and biases $r\theta_\ell+\tfrac{r-1}{2}$ has pre-activation $rz+\tfrac{r-1}{2}$
on the input itself, so Theorem~\ref{thm:amplify} holds for it with $M$ rows in
place of $rM$. A ternary row cannot hold the weight $r$, and reading $r$ copies of
an input with weight $1$ is how it realizes one.

\begin{remark}[Thresholds and restoring organs]\label{rem:restoring}
In von Neumann's multiplexing~\cite{vonNeumann1956} every component fails with a
fixed probability $p$ whatever its inputs, and a restoring organ of majority
elements follows every level; reliable computation then needs redundancy
logarithmic in the size of the circuit~\cite{DobrushinOrtyukov1977,Pippenger1985}
and is possible only below a threshold on
$p$~\cite{HajekWeller1991,EvansSchulman1999}. The amplification $F^{(r)}$ is
Pierce's interwoven redundant logic~\cite{Pierce1964,Pierce1965} with threshold
units. The restoring organ is absorbed into every unit, which sums the $r$ copies
of each input with one weight, and the identity rows of $\Lev$ become majority
units. Amplification does not enlarge the tolerance to systematic errors: for $F^{(r)}$
the conditions of Proposition~\ref{prop:mismatch} are those for $F$ with $\eta/r$ in
place of $\eta$.
\end{remark}

\begin{construction}[The amplified host]\label{con:amplified}
$\Lev(\Nhost)$ is a clocked netlist: its rows are units, each reading rows of the
layer before it, and the rows of its last layer are the next-state values. Let
$N_r=\Lev(\Nhost)^{(r)}$, a clocked netlist on $r\cdot2057$ state bits with
$41{,}967\,r$ units and $75{,}501\,r^2$ edges. Its rows are again in layer order,
each reading only the layer before, so $\Lev(N_r)=N_r$, and
Proposition~\ref{prop:amplify} describes every row of $\Lev(N_r)$. Its device
reads each bit of a register as the majority of the $r$ copies and writes every
copy.
\end{construction}

\begin{corollary}[Exact self-reproduction of the amplified host]\label{cor:amplifiedselfrep}
$N_r$ realizes the transducer of Definitions~\ref{def:transducer}
and~\ref{def:subleq} on replicated states. Hence, for
$\tau^\ast_r=\describe(\field{\sigma(N_r)}\cat\field{m^\ast})$, the instance
$\Omega^\ast_r=(N_r,m^\ast,\tau^\ast_r)$ is exactly self-reproducing and
Proposition~\ref{prop:steps} gives its running time.
\end{corollary}

\begin{proof}
From a replicated state $N_r$ produces the replicated next state
(Proposition~\ref{prop:amplify}(i)), and the device reads the majority of equal
copies, which is their common value, and writes every copy, so the run of $N_r$
from $\rho_r(x_0(m))$ visits the replications of the states $\Nhost$ visits from
$x_0(m)$ on the same tape and emits the same bytes. The proofs of
Theorems~\ref{thm:univ} and~\ref{thm:selfrep} use the transducer alone, so they
hold with $N_r$ in place of $\Nhost$. The same holds for every odd $r$ and every
$A\ge8$ with $\Lev(H_A)$ in place of $\Lev(\Nhost)$: the run of $N_r$ from a
replicated state visits the replicated states of the transducer's run with the same
output, and on $P^\ast$ and every recipe it reads its halt bit as $0$ for the steps
of Proposition~\ref{prop:steps} and as $1$ after them, with the output of
Theorem~\ref{thm:selfrep}.
\end{proof}

Run with $N_3$ in place of $\Nhost$ under Gaussian read noise on every
pre-activation, $\Omega^\ast$ emits the whole instance at $s=0.10$, $0.15$, $0.20$
and $0.25$, where the host itself fails at $s=0.10$ after $143$ bytes and at
$s=0.15$ and $0.20$ before emitting one (Appendix~\ref{app:measure}). A single wrong
copy is outvoted in the next layer, and the computation fails only when two wrong
copies meet, an event of second order in the probability of a flip. The following
bound counts these events.

\begin{proposition}[Restoration]\label{prop:restore}
Let the perturbations of $N_3$ be independent centred Gaussians of standard
deviation $s$ on every pre-activation. The probability that the state of $N_3$ read
by majority differs from the orbit of $x_0$ after one of the first $T$ steps is at
most $\sum_{t<T}b(t)$, where $b(t)$, given in Definition~\ref{def:restore}, is
computed layer by layer along the orbit from the distance of each row to the
comparator and from the inputs whose errors move it toward the comparator.
\end{proposition}

Appendix~\ref{app:restore} defines $b(t)$ and proves the proposition. The artifact
computes the bound along the exact orbit, in C, by \texttt{src/restore\_bound.c}.
The matrix $K$ is held sparse, and the
entries it drops are carried by rank-one majorants that contribute less than
$10^{-12}$ of $b(t)$. The run is split into $36$ chunks, each restarted twenty steps
early from $P=1$ and $K=0$, values that bound the probabilities they stand for; the
chunks reproduce the uninterrupted recursion to within $10^{-14}$.

Along the self-reproducing run, $T=1{,}664{,}939$, the bound $\sum_{t<T}b(t)$ is
$6.13\times10^{-7}$ at $s=0.20$, where Theorem~\ref{thm:amplify}(ii) gives
$1.34\times10^{-2}$. It is $3.11\times10^{-4}$ at $s=0.22$ and $0.271$ at $s=0.25$,
where the first-order bound exceeds one. Along the first $20{,}000$ steps of the
self-describing instance, $b(t)$ lies within the exact $95\%$ Poisson interval of
the measured failure rate at $s=0.27$, $0.30$ and $0.33$, and $1.47$ times above the
rate at $s=0.36$, the highest level measured (Table~\ref{tab:bound3},
Figure~\ref{fig:noise}(b)). The failure rates are $21$ and $12$ times below the
rates at which some copy of the replicated state left its orbit at $s=0.27$ and
$0.30$.

\subsection{A generation of the stored host under noise}\label{sec:storednoise}

The ring machine is a ternary threshold network like the host, so the bounds of this
section apply to a whole generation of the stored host, in which the host decodes its
recipe into the next bank, copies the recipe and hands over. The following corollary
gives the noise under which such a generation runs exactly.

\begin{corollary}[Self-reproduction under noise]\label{cor:storednoise}
For the ring of three banks in which Section~\ref{sec:storedhost} runs $\mathcal S$,
let $M=470{,}604{,}761$ be the number of rows of $\Lev(N^{15}_{3,S,L})$ and
$T=95{,}762{,}842{,}100$ the number of applications in a generation. Under
independent centred Gaussian perturbations of deviation $s$ on every pre-activation,
a generation of $\mathcal S$ runs exactly with probability at least
$1-TM\erfc\bigl(1/(2\sqrt2\,s)\bigr)$, which exceeds $0.999$ for $s\le0.050$. The threefold and fivefold amplifications of
Definition~\ref{def:amplify} reach the same probability for $s\le0.148$ and
$s\le0.247$.
\end{corollary}

\begin{proof}
Proposition~\ref{prop:random} and Theorem~\ref{thm:amplify}(ii) over the $T$
applications; Appendix~\ref{app:sizes} gives the rows of $\Lev(N^\alpha_{K,S,L})$
in closed form.
\end{proof}

\section{Related work}\label{sec:related}

\emph{Threshold computers.} Von Neumann's EDVAC report describes the
stored-program computer with McCulloch--Pitts elements~\cite{vonNeumann1945}.
Goto's parametron computes the majority of three inputs, an input inverted by
reversing the polarity of its coupling~\cite{Goto1959}, so parametron circuits are
networks of ternary threshold units; MUSASINO-1~\cite{MurogaTakashima1959} and
PC-1, of about $4{,}200$ parametrons~\cite{Okoma2002}, were stored-program machines
built from them, and DONUT a general-purpose computer of about $500$ threshold
gates~\cite{CoatesLewis1964}; see also~\cite{Muroga1971}.
Weighted threshold circuits of depth $d$ are simulated by majority circuits of depth
$d+1$~\cite{GoldmannHastadRazborov1992,GoldmannKarpinski1998,SiuRoychowdhuryKailath1995},
and Remark~\ref{rem:ternarydepth} is the smallest instance of that exchange.

\emph{Networks as programmable computers.} Recurrent networks of saturated-linear
units with rational weights are
Turing-universal~\cite{SiegelmannSontag1995,SimaOrponen2003}, and programs have been
compiled into neural networks~\cite{GruauEtAl1995,LindnerEtAl2023}. A looped
transformer with fixed weights executes SUBLEQ programs placed in its
input~\cite{GiannouEtAl2023}, as does a looped ReLU perceptron~\cite{LiangEtAl2025};
a fixed transformer simulates any transformer of a class from a description in its
input~\cite{LiuAndoniHsu2026}, and Loom executes compiled C programs in a looped
transformer whose state holds the program~\cite{Turkcan2026}. A programmable neural
network, a fixed-weight network with auxiliary inputs, behaves as any network of a
class when those inputs carry a code of its weight
matrix~\cite{DonnarummaPreveteTrautteur2012}. These use real-valued weights. The ring machine is a fixed ternary network of this kind, which evaluates
any network stored in its state. Valiant's universal circuits simulate any circuit
of size $n$ from its program bits in size $\Theta(n\log n)$~\cite{Valiant1976}, and
intrinsic simulation between automata networks~\cite{RiosWilsonTheyssier2024a}
builds the simulating network from the simulated one; here the simulating network
is fixed.

\emph{Finiteness and noise.} Every network here is a finite
automaton~\cite{Kleene1956,Minsky1967}. Power beyond finite
automata comes from unbounded precision~\cite{SiegelmannSontag1995}, growing numbers
of units~\cite{ChungSiegelmann2021} or changing weights~\cite{CabessaSiegelmann2014},
and von Neumann's cellular automaton is universal because its tape can be
extended~\cite{vonNeumann1966,Burks1970}. Analog noise reduces discrete-time analog
models to finite automata~\cite{MaassOrponen1998}.

\emph{Self-reproduction.} Von Neumann's universal constructor uses a description
once as instructions for construction and once as data to be
copied~\cite{vonNeumann1966,Burks1970}. Moore and Myhill gave abstract
formulations~\cite{Moore1962,Myhill1964}, Lee and Thatcher constructed Turing
machines that print their own description~\cite{Lee1963,Thatcher1963}, cellular
constructions were given by Arbib, Codd, Langton, Pesavento, and Perrier, Sipper
and Zahnd, the last with an attached program and a growing data
tape~\cite{Arbib1966,Codd1968,Langton1984,Pesavento1995,PerrierSipperZahnd1996},
and Ib\'a\~nez et al.\ reproduce by self-inspection~\cite{IbanezEtAl1995}.
Self-replicating field-programmable arrays copy configurations between
cells~\cite{MangeEtAl1998}, and Cell Matrix cells read and write the truth tables
of their neighbours~\cite{DurbeckMacias2001}. Program replicators include
Tierra~\cite{Ray1991} and those that emerge in~\cite{AgueraYArcas2024}. Cotler,
Hongler and Hudcov\'a give an informal condition for self-replication in cellular
automata and place the automata that meet it strictly between the globally and the
locally universal ones, so that Turing universality does not imply self-replication;
they require the encoder and the decoder between symbolic computation and the
dynamics to be computationally simpler than the computation attributed to the
system~\cite{CotlerHonglerHudcova2025}. In those terms the host running $P^\ast$
(Section~\ref{sec:selfrep}) is a replicator relative to an instantiation map. The
organism and the stored host need no such map: each is a configuration of a fixed
threshold network, the only map between the computation and the dynamics is the
encoding of a state of the ring machine in the bits of that network
(Theorem~\ref{thm:physnet}), and what is reproduced includes the weights of the
network that computes.

\emph{Networks that produce their own parameters.} Neural networks have been
trained to output their own weights~\cite{ChangLipson2018,Gabor2022}, to reproduce
over generations~\cite{RandazzoVersariMordvintsev2021,ShvartzmanRam2024} and to
emit the program text that trains them~\cite{Schmidgall2021}; hypernetworks and
self-referential weight matrices generate or modify
weights~\cite{HaDaiLe2017,Schmidhuber1993,IrieEtAl2022}. In each case the relation
between a network and its parameters is learned and approximate.

\emph{Machine-checked universal machines.} A universal Turing machine has been
programmed and verified in Coq, with polynomial time and constant-factor space
overhead~\cite{ForsterKunzeWuttke2020}, and the Game of Life has been simulated in
the Game of Life by circuits verified in HOL4~\cite{MyreenCarneiro2025}. Here one
Rocq development covers the interpreter, the universal computer, construction,
reproduction and noise.

\section{Conclusion}\label{sec:conclusion}

A fixed network of ternary threshold units holds threshold networks as records and
runs them, and a stored network can build any network that fits a bank, copy its
description and hand control to what it built. The organism and the stored SUBLEQ
host rebuild themselves in this way, weights and biases included, in every ring, and
integer pre-activations keep every generation exact while no perturbation of a
pre-activation reaches one half.

The fixed network itself is not reproduced. Its netlist is ternary and compiles into
records by Lemma~\ref{cor:compile}, but its state, which contains the window, the
banks and the tapes, cannot fit in its own window, so a network that also rebuilds
the network that runs it needs another construction. Other questions are open:
whether a ternary step of SUBLEQ can have depth below $11$, since
Lemma~\ref{lem:muxdepth} and Remark~\ref{rem:ternarydepth} bound only a read and the
branch; how few records a self-reproducing stored network needs; and how much noise
a reproducing run tolerates when the perturbations are correlated or drift, which
the independence assumed in Section~\ref{sec:robust} excludes.

\appendix
\section{Serialization and file formats}\label{app:sigma}

\begin{definition}[Serialization]\label{def:sigma}
Let $N$ be a clocked netlist on $n$ state bits. Number its signals with the state
bits first, in the order fixed by the state layout, and the units after them in the
order of $N$. The \emph{serialization} $\sigma(N)$ is the \texttt{safetensors}
encoding of four integer vectors:
\begin{enumerate}[label=(\roman*)]
\item the bias of each unit, in the order of $N$;
\item the number of predecessors of each unit, in the same order;
\item one signed entry per edge, its magnitude the index of the predecessor plus
one and its sign the weight, the edges of a unit consecutive and in the order the
unit lists them;
\item for each state bit, the index of the signal computing its next value.
\end{enumerate}
The file is an eight-byte little-endian header length, a JSON header object, and
the concatenated tensor data in little-endian order. The biases are stored as
eight-bit integers, and the fan-ins, the edges and the next-state indices as
sixteen-bit integers when their values fit and as thirty-two-bit integers
otherwise. Seven thirty-two-bit scalars record the address width, the data width,
the number of memory words, the width of the program counter, the number of state
bits, the number of units and the version of the format. Four choices that a
general-purpose serializer leaves open are fixed. The header object begins with the
key \texttt{\_\_metadata\_\_}, whose value lists the machine name, the quantization
tag, the state layout and the device addresses of Table~\ref{tab:device}, in that
order. The remaining entries are the tensors, ordered by element type and then by
name, with the data in the same order. The JSON has compact separators and no
non-ASCII escaping. Spaces pad the header so that the tensor data begins at an
eight-byte boundary.
\end{definition}

The fixed choices make $\sigma$ a function. Every predecessor carries the weight
$\pm1$, so the sign of an edge entry gives its weight, and the fan-in vector cuts
the edge vector into the predecessor list of each unit; every index refers to a
signal that these vectors define. So $N$, the order of its units included, is
recovered from $\sigma(N)$ alone.

For $\Nhost$ the metadata has \texttt{machine} \texttt{subleq8} and
\texttt{weight\_quantization} \texttt{ternary}. Its \texttt{signal\_registry} is
the first string below, the state layout of Construction~\ref{con:host} by first
coordinate and width, with the number of words for the memory, and its \texttt{io}
is the second, the addresses of Table~\ref{tab:device}:
\begin{center}\ttfamily\small
\{"halt":[8,1],"mem":[9,256,8],"pc":[0,8]\}\\
\{"c\_eot":250,"c\_rd":252,"c\_rw":251,"c\_wr":249,"halt\_pc":255,"r\_in":253,"r\_out":254\}
\end{center}
The scalars are the \texttt{I32} tensors \texttt{manifest.addr\_bits},
\texttt{data\_bits}, \texttt{memory\_words}, \texttt{pc\_width},
\texttt{state\_bits}, \texttt{units} and \texttt{version}, holding $8$, $8$, $256$,
$8$, $2057$, $17{,}598$ and $6$, and the vectors are \texttt{net.fanin},
\texttt{net.next} and \texttt{net.pred} of type \texttt{I16} and \texttt{net.bias}
of type \texttt{I8}.

The netlists $N^\alpha_{K,S,L}$ are written as files of little-endian $32$-bit
integers: the numbers of state bits, units and edges; the bias of each unit; the
fan-in of each unit; one signed entry per edge, as in $\sigma$; and, for each state
bit, the index of the signal computing its next value. The state bits are numbered
first, in the order of Theorem~\ref{thm:physnet}, and the units after them, in the
order of the Rocq definition. An image of records, such as the organism or the
stored host, is the sequence of its records, each laid out as in
Definition~\ref{def:twohead}, in the order of the Rocq definition, with the window
cells of Construction~\ref{con:organism} or~\ref{con:stored}; in the stored host the
cells of the units, of the levels and of the device follow the state bits, in that
order.

\section{Sizes of the constructions}\label{app:sizes}

This appendix gives the numbers of units, edges, rows, records and steps to which
the main text refers. The closed forms are proved in Rocq (Table~\ref{tab:rocq}),
and the census of a particular byte string is computed by the artifact.

\subsection{The ring machine}

With $W=2^\alpha$, $\beta=8+3\alpha$ and $N=S\beta$, the netlist $N^\alpha_{K,S,L}$ of
Theorem~\ref{thm:physnet} has
\begin{gather*}
5KN+4KS+7K\beta+9KL+22K+8W+\beta+22\quad\text{units and}\\
14KN+10KS+16K\beta+23KL+43K+3\alpha W+12W+46\quad\text{edges},
\end{gather*}
and $\Lev(N^\alpha_{K,S,L})$ has depth $19$ and
$34W+\beta+43K+24K\beta+30KN+23KS+51KL-56$ rows.

\subsection{The host}

The netlist $H_A$ has
\[
(8A+4)2^A+\tfrac12(A^2+39A+4)\quad\text{units and}\quad
(49A+4)2^{A-1}+\tfrac16(A^3+9A^2+194A+24)\quad\text{edges},
\]
the sums over the blocks of Table~\ref{tab:blocks}. Of the units,
$(8A+4)2^A+5A$ form the memory, that is, the decode lines, the fetch, the operand
reads and the write, and the other $\tfrac12(A^2+29A+4)$ the arithmetic and
control datapath. For $A=8$ these are $17{,}598$ units, $17{,}448$ of the memory and
$150$ of the datapath, carrying $68{,}730$ parameters: $51{,}132$ weights of $\pm1$
and one bias per unit in $[-8,1]$. The serialization $\sigma(\Nhost)$ has
$160{,}240$ bytes.

$\Lev(H_A)$ has depth $13$ and $(19A+10)2^A+\tfrac12(A^2+109A+54)$ rows, of which
$(11A+6)2^A+35A+25$ are identity units and the others the units of $H_A$. For $A=8$
these are $M=41{,}967$ rows on $n=2057$ state bits, $24{,}369$ of them identity
units, in thirteen layers of $2313$, $8210$, $2090$, $2587$, $6427$, $2347$, $2379$,
$2389$, $2609$, $2353$, $2339$, $3867$ and $2057$ rows. As dense matrices the
layers have $124{,}165{,}940$ entries, $75{,}501$ of them nonzero.

The folded host of Remark~\ref{rem:fold} has
\[
(8A+2)2^A+\tfrac12(A^2+39A+4)\quad\text{units and}\quad
(10A^2+37A+2)2^{A-1}+\tfrac16(A^3+9A^2+194A+24)\quad\text{edges},
\]
$17{,}086$ and $120{,}508$ for $A=8$.

\subsection{Programs and recipes}

The recipe $\tau_H=\describe(\sigma(\Nhost))$ of Corollary~\ref{cor:selfdesc} has
$144{,}891$ bytes, against $144{,}818$ for the shortest recipe of $\sigma(\Nhost)$.
Its census is $\nlit=1137$, $\nrep=147$, $f_{\mathrm{lit}}=143{,}459$ and
$f_{\mathrm{rep}}=16{,}781$, so Theorem~\ref{thm:univ} gives a run of $793{,}709$
steps.

The tape $\tau^\ast$ of Construction~\ref{con:selfrep} has $144{,}996$ bytes, and
$\ser(\Omega^\ast)$ has $305{,}508$: $8+160{,}240$ for the framed serialization of
the host, $8+256$ for the framed memory image and $144{,}996$ for the tape. The
census of $\tau^\ast$ is $\nlit=1141$, $\nrep=152$, $f_{\mathrm{lit}}=143{,}550$ and
$f_{\mathrm{rep}}=16{,}962$, so the run of Proposition~\ref{prop:steps} takes
$794{,}960$ steps in phase~(A), one in phase~(R) and $869{,}978$ in phase~(B),
$1{,}664{,}939$ in all.

\subsection{The stored host}

In $\mathcal S=\mathcal S_{8,20}$ the window holds the thirteen cells of
Definition~\ref{def:twohead}, $\mathrm{booted}$, $\mathrm{cnt}$, $\mathrm{dest}$
and $\mathrm{lim}$ in cells $13$ to $54$, the state bits of $\Nhost$ in cells $55$
to $2111$, and after them the cells of its $17{,}598$ units, of the levels and of
the device, $19{,}765$ cells in all, so $\alpha=15$ and $\beta=53$. Its $36{,}080$
records are $2104$ of the startup, $33{,}568$ of the step, $403$ of the device and
$1+4$ of the hand-over. Its image has $1{,}912{,}240$ bits, read as $239{,}030$
bytes, and $\tau_{\mathcal S}=\describe(\mathcal S)$ has $240{,}900$ bytes.

For a recipe $\tau'$ of $\nlit$ literal and $\nrep$ repeat tokens,
$|\tau'|=f_{\mathrm{lit}}+\nlit+2\nrep+1$, so Proposition~\ref{prop:steps} gives a
generation of Theorem~\ref{thm:storedrep}
\[
P=11f_{\mathrm{lit}}+4f_{\mathrm{rep}}+13\nlit+21\nrep+16
\]
passes. For $\tau'=\describe(\mathcal S_{A,w})$ of an image of $B$ bytes,
$P=11B+13\lceil B/127\rceil+16$ when the image has no four equal consecutive bytes,
since $\describe$ then emits $\lceil B/127\rceil$ literal tokens and nothing else,
and $P\le\tfrac{25}{2}B+29$ in every case, since the proof of
Lemma~\ref{lem:describe} gives $\nlit\le f_{\mathrm{lit}}/127+\nrep+1$ and
$4\nrep\le f_{\mathrm{rep}}$. The census of $\tau_{\mathcal S}$ is $\nlit=1888$, $\nrep=8$,
$f_{\mathrm{lit}}=238{,}995$ and $f_{\mathrm{rep}}=35$, which gives the
$P=2{,}653{,}813$ passes stated after Theorem~\ref{thm:storedrep}; there the tapes
have $L=8|\tau_{\mathcal S}|+64=1{,}927{,}264$ cells.

For every $A\ge8$ and $w\ge2$ the stored host $\mathcal S_{A,w}$ has
$A2^A+A+2w+8$ startup records, $11A+12w+75$ device records,
\[
\tfrac{33}{2}A2^A-2^{A+1}+\tfrac16(A^3-A)+A^2+17A+4
\]
step records and $4+\pi$ hand-over records, $\pi\le7$ of them empty. The step
records count the records of Lemma~\ref{cor:compile} over the blocks of
Table~\ref{tab:blocks}: a conjunction or a disjunction of $k\ge2$ literals takes
$k-1$ records, a unit of threshold $2$ on three or four literals four or eight, and
the cubic term comes from the carry terms. The stored host $\mathcal S_A$ of
Theorem~\ref{thm:widths} therefore has
\[
G_A=35A2^{A-1}-2^{A+1}+\tfrac16(A^3-A)+A^2+29A+14w_A+91+\pi_A
\]
records. Its window satisfies
$A+\lfloor\log_2A\rfloor+3\le\alpha_A\le A+\lfloor\log_2A\rfloor+8$, its counter
width satisfies $w_A\le2A+10$, and a generation in banks of $G_A$ slots takes between
$2^8A^34^A$ and $2^{13}A^34^A$ applications of $\mathcal M^{\alpha_A}_{K,G_A,L}$.

For $A=8$ to $12$, $\mathcal S_A$ has $36{,}080$, $80{,}463$, $178{,}112$,
$391{,}220$ and $853{,}176$ records, windows of $2^{15}$ to $2^{19}$ cells,
counters of $20$, $21$, $22$, $23$ and $24$ bits and images of $239{,}030$ to
$6{,}932{,}055$ bytes, and a generation takes $2{,}653{,}813$, $6{,}253{,}470$,
$14{,}583{,}670$, $33{,}661{,}915$ and $76{,}962{,}164$ passes, from
$9.57\times10^{10}$ to $6.57\times10^{13}$ applications of $\mathcal M$ with $S=G_A$.

\subsection{The amplified host}

For $r=3$ the amplified host $N_3$ of Construction~\ref{con:amplified} has $6171$
state bits, $125{,}901$ units, $679{,}509$ edges, depth $13$ and biases in
$[-23,4]$, and $|\sigma(N_3)|=3{,}121{,}539$ bytes. The tape $\tau^\ast_3$ has
$3{,}022{,}111$ bytes, $\ser(\Omega^\ast_3)$ has $6{,}143{,}922$, and the run of
$\Omega^\ast_3$ takes $33{,}796{,}449$ steps.

\section{The records of the organism}\label{app:organism}

The organism $\rho$ of Construction~\ref{con:organism} consists of the $143$ records
below, in the order of the bank, each written as the gate of
Definition~\ref{def:code} that it computes into its output cell. Cells are named by
their role: $\mathrm{rd}$, $\mathrm{arm}$ and $\mathrm{cpy}$ are the cells under the read head,
the arm and the copy head, $\mathrm{AR}$ to $\mathrm{HO}$ the requests, $m$ the mode,
\phstart, \phcodon, \phgate{} and \phaddr{} the phase, $\xi_0,\dots,\xi_3$ the codon
register, $z_0,\dots,z_4$ the counter, and $t_{27},\dots,t_{44}$ the temporaries,
each reused once its value has been read. Every record is a gate, so $\code(\rho)$
has no escape. The records are generated by \texttt{src/twohead.py}, which also
writes them to \texttt{rocq/Organism.v}.

{\footnotesize
\emph{Record 1, the start flag:} 1~$t_{27}\gets \textsc{start}$. \emph{Records 2--5, the counter at zero:} 2~$t_{28}\gets \neg z_0\wedge \neg z_1$; 3~$t_{28}\gets t_{28}\wedge \neg z_2$; 4~$t_{28}\gets t_{28}\wedge \neg z_3$; 5~$t_{28}\gets t_{28}\wedge \neg z_4$. \emph{Records 6--14, the codon after this pass's shift, and its end and escape at completion:} 6~$t_{29}\gets \xi_2\wedge \xi_1$; 7~$t_{30}\gets \xi_0\wedge \mathrm{rd}$; 8~$t_{31}\gets t_{29}\wedge t_{30}$; 9~$t_{29}\gets t_{29}\wedge \neg t_{30}$; 10~$t_{30}\gets \textsc{codon}\wedge t_{28}$; 11~$t_{32}\gets t_{30}\wedge t_{31}$; 12~$t_{30}\gets t_{30}\wedge \neg t_{31}$; 13~$t_{31}\gets t_{32}\wedge \neg m$; 14~$t_{33}\gets t_{32}\wedge m$. \emph{Records 15--17, the escape of the codon held in \textsc{gate}:} 15~$t_{34}\gets \xi_3\wedge \xi_2$; 16~$t_{35}\gets \xi_1\wedge \xi_0$; 17~$t_{34}\gets t_{34}\wedge \neg t_{35}$. \emph{Records 18--24, consume, write under the arm, copy:} 18~$t_{35}\gets \textsc{gate}\wedge t_{34}$; 19~$t_{36}\gets \textsc{codon}\vee \textsc{addr}$; 20~$t_{35}\gets t_{36}\vee t_{35}$; 21~$t_{36}\gets \textsc{gate}\vee \textsc{addr}$; 22~$t_{36}\gets t_{36}\wedge \neg m$; 23~$t_{34}\gets \textsc{gate}\wedge \neg t_{34}$; 24~$t_{37}\gets m\wedge t_{35}$. \emph{Records 25--63, the gate bit $7-z$ that the held codon determines:} 25~$t_{38}\gets \xi_3\wedge \neg \xi_2$; 26~$t_{39}\gets \xi_1\vee \xi_0$; 27~$t_{40}\gets \neg \xi_1\vee \neg \xi_0$; 28~$t_{39}\gets t_{39}\wedge t_{40}$; 29~$t_{40}\gets \xi_1\wedge \xi_0$; 30~$t_{41}\gets \xi_1\vee \xi_0$; 31~$t_{40}\gets \xi_2\wedge t_{40}$; 32~$t_{41}\gets \neg \xi_2\wedge t_{41}$; 33~$t_{40}\gets t_{40}\vee t_{41}$; 34~$t_{41}\gets \xi_3\wedge \xi_1$; 35~$t_{40}\gets \neg \xi_3\wedge t_{40}$; 36~$t_{40}\gets t_{41}\vee t_{40}$; 37~$t_{41}\gets \xi_2\vee t_{39}$; 38~$t_{39}\gets \neg \xi_2\vee \neg t_{39}$; 39~$t_{39}\gets t_{41}\wedge t_{39}$; 40~$t_{41}\gets \xi_3\wedge \xi_1$; 41~$t_{39}\gets \neg \xi_3\wedge t_{39}$; 42~$t_{39}\gets t_{41}\vee t_{39}$; 43~$t_{41}\gets \neg t_{38}\wedge \neg \xi_0$; 44~$t_{38}\gets \neg t_{38}\wedge \xi_0$; 45~$t_{42}\gets z_0\wedge \neg \xi_1$; 46~$t_{43}\gets \neg z_0\wedge \xi_1$; 47~$t_{42}\gets t_{42}\vee t_{43}$; 48~$t_{41}\gets z_0\wedge t_{41}$; 49~$t_{38}\gets \neg z_0\wedge t_{38}$; 50~$t_{38}\gets t_{41}\vee t_{38}$; 51~$t_{41}\gets z_1\wedge t_{42}$; 52~$t_{38}\gets \neg z_1\wedge t_{38}$; 53~$t_{38}\gets t_{41}\vee t_{38}$; 54~$t_{41}\gets z_1\vee z_0$; 55~$t_{40}\gets t_{41}\wedge t_{40}$; 56~$t_{39}\gets \neg t_{41}\wedge t_{39}$; 57~$t_{39}\gets t_{40}\vee t_{39}$; 58~$t_{38}\gets z_2\wedge t_{38}$; 59~$t_{39}\gets \neg z_2\wedge t_{39}$; 60~$t_{38}\gets t_{38}\vee t_{39}$; 61~$t_{38}\gets t_{34}\wedge t_{38}$; 62~$t_{34}\gets \neg t_{34}\wedge \mathrm{rd}$; 63~$t_{34}\gets t_{38}\vee t_{34}$. \emph{Records 64--77, the next phase:} 64~$t_{38}\gets \textsc{codon}\wedge \neg t_{28}$; 65~$t_{38}\gets t_{38}\vee t_{31}$; 66~$t_{39}\gets \textsc{addr}\wedge t_{28}$; 67~$t_{38}\gets t_{38}\vee t_{39}$; 68~$t_{38}\gets t_{38}\vee \textsc{start}$; 69~$t_{29}\gets \neg m\vee t_{29}$; 70~$t_{39}\gets t_{30}\wedge t_{29}$; 71~$t_{40}\gets \textsc{gate}\wedge \neg t_{28}$; 72~$t_{39}\gets t_{39}\vee t_{40}$; 73~$t_{29}\gets t_{30}\wedge \neg t_{29}$; 74~$t_{30}\gets \textsc{gate}\wedge t_{28}$; 75~$t_{29}\gets t_{29}\vee t_{30}$; 76~$t_{30}\gets \textsc{addr}\wedge \neg t_{28}$; 77~$t_{29}\gets t_{29}\vee t_{30}$. \emph{Records 78--107, the counter: load $3$, $7$ or $17$, or count down:} 78~$t_{28}\gets t_{28}\vee \textsc{start}$; 79~$t_{30}\gets \neg z_0$; 80~$t_{40}\gets \neg z_0$; 81~$t_{41}\gets z_1\vee t_{40}$; 82~$t_{42}\gets \neg z_1\vee \neg t_{40}$; 83~$t_{41}\gets t_{41}\wedge t_{42}$; 84~$t_{40}\gets t_{40}\wedge \neg z_1$; 85~$t_{42}\gets z_2\vee t_{40}$; 86~$t_{43}\gets \neg z_2\vee \neg t_{40}$; 87~$t_{42}\gets t_{42}\wedge t_{43}$; 88~$t_{40}\gets t_{40}\wedge \neg z_2$; 89~$t_{43}\gets z_3\vee t_{40}$; 90~$t_{44}\gets \neg z_3\vee \neg t_{40}$; 91~$t_{43}\gets t_{43}\wedge t_{44}$; 92~$t_{40}\gets t_{40}\wedge \neg z_3$; 93~$t_{44}\gets z_4\vee t_{40}$; 94~$t_{40}\gets \neg z_4\vee \neg t_{40}$; 95~$t_{40}\gets t_{44}\wedge t_{40}$; 96~$t_{44}\gets t_{38}\vee t_{39}$; 97~$z_0\gets t_{28}\vee t_{30}$; 98~$t_{30}\gets t_{28}\wedge t_{44}$; 99~$t_{41}\gets \neg t_{28}\wedge t_{41}$; 100~$z_1\gets t_{30}\vee t_{41}$; 101~$t_{30}\gets t_{28}\wedge t_{39}$; 102~$t_{41}\gets \neg t_{28}\wedge t_{42}$; 103~$z_2\gets t_{30}\vee t_{41}$; 104~$z_3\gets \neg t_{28}\wedge t_{43}$; 105~$t_{30}\gets t_{28}\wedge t_{29}$; 106~$t_{28}\gets \neg t_{28}\wedge t_{40}$; 107~$z_4\gets t_{30}\vee t_{28}$. \emph{Records 108--119, the codon register shifts in \textsc{codon}:} 108~$t_{28}\gets \textsc{codon}\wedge \xi_2$; 109~$t_{30}\gets \neg \textsc{codon}\wedge \xi_3$; 110~$\xi_3\gets t_{28}\vee t_{30}$; 111~$t_{28}\gets \textsc{codon}\wedge \xi_1$; 112~$t_{30}\gets \neg \textsc{codon}\wedge \xi_2$; 113~$\xi_2\gets t_{28}\vee t_{30}$; 114~$t_{28}\gets \textsc{codon}\wedge \xi_0$; 115~$t_{30}\gets \neg \textsc{codon}\wedge \xi_1$; 116~$\xi_1\gets t_{28}\vee t_{30}$; 117~$t_{28}\gets \textsc{codon}\wedge \mathrm{rd}$; 118~$t_{30}\gets \neg \textsc{codon}\wedge \xi_0$; 119~$\xi_0\gets t_{28}\vee t_{30}$. \emph{Records 120--125, the bits under the arm and the copy head:} 120~$t_{28}\gets t_{36}\wedge t_{34}$; 121~$t_{30}\gets \neg t_{36}\wedge \mathrm{arm}$; 122~$\mathrm{arm}\gets t_{28}\vee t_{30}$; 123~$t_{28}\gets t_{37}\wedge \mathrm{rd}$; 124~$t_{30}\gets \neg t_{37}\wedge \mathrm{cpy}$; 125~$\mathrm{cpy}\gets t_{28}\vee t_{30}$. \emph{Records 126--133, the mode and the phase flags:} 126~$t_{28}\gets m\vee t_{32}$; 127~$t_{30}\gets \neg m\vee \neg t_{32}$; 128~$t_{28}\gets t_{28}\wedge t_{30}$; 129~$m\gets t_{28}\wedge \neg \textsc{start}$; 130~$\textsc{codon}\gets t_{38}$; 131~$\textsc{gate}\gets t_{39}$; 132~$\textsc{addr}\gets t_{29}$; 133~$\textsc{start}\gets t_{33}$. \emph{Records 134--143, the requests:} 134~$\mathrm{AA}\gets t_{36}$; 135~$\mathrm{AC}\gets t_{37}$; 136~$\mathrm{AR}\gets t_{35}$; 137~$\mathrm{RR}\gets t_{27}\vee t_{31}$; 138~$t_{28}\gets t_{27}\vee t_{33}$; 139~$\mathrm{RA}\gets t_{28}$; 140~$\mathrm{RC}\gets t_{28}$; 141~$\mathrm{CB}\gets t_{27}$; 142~$\mathrm{CT}\gets t_{27}$; 143~$\mathrm{HO}\gets t_{33}$.
}

\section{The restoration bound}\label{app:restore}

\begin{definition}[The restoration bound]\label{def:restore}
\allowdisplaybreaks
Let $s$ be the deviation of Proposition~\ref{prop:restore}, and unroll
$\Lev(\Nhost)$ along the orbit of $x_0$, the rows of layer~$1$ of each step reading
the rows of layer~$d$ of the step before. For a row $v$ of a step, let $z_v$ be its pre-activation and $y_v$ its
value on the orbit, and put $D_v=3|z_v+\tfrac12|$, $q_j(v)=Q\bigl((D_v-j)/s\bigr)$
for $j=0,1,2,3$ and $\Delta_v=q_1(v)-q_0(v)$. Let $U(v)$ be the set of rows $u$ that
$v$ reads with a weight $w$ satisfying $w(1-2y_u)(z_v+\tfrac12)<0$, the inputs
whose error moves $v$ toward the comparator. Define, layer by layer,
\[
\begin{gathered}
P_v=3q_0(v)+3\Delta_vX_v,\qquad
P^{(2)}_v=3q_0(v)^2+3\bigl(q_1(v)^2-q_0(v)^2\bigr)X_v,\\
P^{(3)}_v=q_0(v)^3+\bigl(q_1(v)^3-q_0(v)^3\bigr)X_v,
\end{gathered}
\]
where $X_v=\sum_{u\in U(v)}P_u$, and for distinct rows $v_1,v_2$ of one layer
\[
K_{v_1v_2}=9\Delta_{v_1}\Delta_{v_2}\Bigl(\sum_{u\in U(v_1)\cap U(v_2)}P_u
+\sum_{\substack{u_1\in U(v_1),\ u_2\in U(v_2)\\ u_1\neq u_2}}K_{u_1u_2}\Bigr),
\]
with $P$ and $K$ zero on the state of step~$0$ and on the device cells of every
state, and $P^{(2)}$ and $P^{(3)}$ zero on every state. Let
\begin{multline*}
b(t)=\sum_{w}\Bigl(\min\bigl(1,3q_2(w)\bigr)\sum_{u\in U(w)}P^{(2)}_u
+\min\bigl(1,3q_3(w)\bigr)\sum_{u\in U(w)}P^{(3)}_u\\*
+\sum_{\substack{u_1,u_2\in U(w)\\ u_1<u_2}}\bigl(P_{u_1}P_{u_2}+K_{u_1u_2}\bigr)\Bigr)
+\sum_{v}P^{(2)}_v,
\end{multline*}
the first sum over the rows of the layers of step~$t$ and the last over the rows
of its last layer.
\end{definition}

\begin{proof}[Proof of Proposition~\ref{prop:restore}]
Call a copy wrong when its output differs from the replicated orbit. Every copy
of a row $v$ reads the same copies, and a wrong copy of an input $u$ moves the
pre-activation of a copy of $v$ by $w(1-2y_u)=\pm1$, toward the comparator
exactly when $u\in U(v)$. Let $S_v$ be the set of wrong copies of the rows of
$U(v)$. By Proposition~\ref{prop:amplify}(ii) a copy of $v$ is wrong only if its
perturbation, on the far side of the comparator, has magnitude at least
$D_v-|S_v|$; given the perturbations of the earlier layers the three copies of
$v$ therefore fail independently, each with probability at most $q_{|S_v|}(v)$.
Call a layer a \emph{stop} when it holds a wrong copy of a row $v$ with
$|S_v|\ge2$, or when it is the last layer of a step and one of its rows has two
wrong copies. The state read by majority is wrong after a step only if the last
layer of that step is a stop, and while no layer is a stop the device reads
every cell correctly and writes every copy of it, so the probability to be
bounded is at most the sum, over the layers of the first $T$ steps, of the
probability that the first stop is that layer.

For a row $v$ let $W_v$ be the event that $v$ has a wrong copy and neither its
layer nor an earlier one is a stop. On $W_v$, $|S_v|\le1$, and $|S_v|=1$ only on
the union of the $W_u$ with $u\in U(v)$, so
$\Pr(W_v)\le3q_0(v)+3\Delta_v\sum_{u\in U(v)}\Pr(W_u)$ and, by induction over the
layers, $\Pr(W_v)\le P_v$; on the device cells of a state $W_v$ is empty, the
device having written every copy. Two of three independent copies fail with
probability at most $3q_j^2$ and all three with probability at most $q_j^3$, so
the probability that no earlier layer is a stop, $|S_v|\le1$, and $v$ has two or
three wrong copies is at most $P^{(2)}_v$ or $P^{(3)}_v$; for a row of a state it
is zero as far as the layers after it are concerned, since two wrong copies of a
state row make the last layer of the step before a stop. For distinct rows
$v_1,v_2$ of one layer the copies of $v_1$ and of $v_2$ fail independently given
the earlier layers, so
\[
\Pr(W_{v_1}\cap W_{v_2})\le9\,\mathbb E\Bigl[\bigl(q_0(v_1)+\Delta_{v_1}\mathbf 1_{|S_{v_1}|=1}\bigr)\bigl(q_0(v_2)+\Delta_{v_2}\mathbf 1_{|S_{v_2}|=1}\bigr)\Bigr],
\]
the expectation taken on the event that no earlier layer is a stop. Both
indicators are $1$ only on the union of the $W_u$ with $u\in U(v_1)\cap U(v_2)$
and of the $W_{u_1}\cap W_{u_2}$ with $u_1\in U(v_1)$, $u_2\in U(v_2)$,
$u_1\neq u_2$; expanding the product and bounding the latter events by induction
gives $\Pr(W_{v_1}\cap W_{v_2})\le P_{v_1}P_{v_2}+K_{v_1v_2}$, since
$\sum_{u_1\neq u_2}P_{u_1}P_{u_2}\le X_{v_1}X_{v_2}$.

If the first stop holds a wrong copy of a row $w$ with $|S_w|\ge2$, then either
$S_w$ consists of two or of three copies of one row $u\in U(w)$, events of
probability at most $P^{(2)}_u$ and $P^{(3)}_u$ after which a copy of $w$ fails
with probability at most $\min(1,3q_2(w))$ or $\min(1,3q_3(w))$, or two distinct
rows $u_1,u_2\in U(w)$ have wrong copies, with probability at most
$P_{u_1}P_{u_2}+K_{u_1u_2}$. Otherwise the first stop is the last layer of a step
and one of its rows $v$ has two wrong copies with $|S_v|\le1$, with probability
at most $P^{(2)}_v$. Summing over the rows of the layers of the first $T$ steps
gives $\sum_{t<T}b(t)$.
\end{proof}

\section{Measurements under perturbation}\label{app:measure}

This appendix gives in full the measurements that Sections~\ref{sec:certified}
and~\ref{sec:amplify} summarize; every run is recorded in the artifact.

\subsection{The host}

Table~\ref{tab:curve} compares the rate at which noisy copies first leave the
exact orbit with $b'$, $b$ and the bound $Mp$ of Proposition~\ref{prop:random}.

Restarting $1024$ noisy copies from the exact state at each of the first $4000$
steps, $4{,}096{,}000$ copy-steps in all, separates the failures of rows from the
deviations of the state. At $s=0.08$, $0.09$, $0.10$, $0.11$, $0.12$ and $0.15$ some
row failed in $31$, $1739$, $35{,}940$, $333{,}003$, $1{,}556{,}080$ and
$4{,}095{,}990$ copy-steps, against $26.0$, $1750.3$, $36{,}109.5$, $332{,}517.2$,
$1{,}554{,}804.7$ and $4{,}095{,}992.7$ under $b'$, and $1$, $100$, $2269$,
$19{,}939$, $76{,}719$ and $11$ of them ended on the orbit; at $s=0.20$ every
copy-step failed and none did.

\begin{table}[ht]
\centering\small
\setlength{\tabcolsep}{5pt}
\begin{tabular}{@{}lrrllll@{}}
\hline
$s$ & Deviated & Exposure & Measured rate, $95\%$ interval & $b'$ & $b$ & $Mp$\\
\hline
$0.06$ & 0/256 & 768{,}000 & $0$, $[0,\,4.80\times10^{-6}]$ & $1.21\times10^{-12}$ & $1.21\times10^{-12}$ & $3.30\times10^{-12}$\\
$0.08$ & 7/256 & 1{,}008{,}916 & $6.94\times10^{-6}$, $[2.79,\,14.3]\times10^{-6}$ & $6.34\times10^{-6}$ & $6.34\times10^{-6}$ & $1.72\times10^{-5}$\\
$0.09$ & 202/256 & 496{,}131 & $4.07\times10^{-4}$, $[3.53,\,4.67]\times10^{-4}$ & $4.27\times10^{-4}$ & $4.27\times10^{-4}$ & $1.16\times10^{-3}$\\
$0.10$ & 256/256 & 29{,}384 & $8.71\times10^{-3}$, $[7.68,\,9.85]\times10^{-3}$ & $8.81\times10^{-3}$ & $8.85\times10^{-3}$ & $2.41\times10^{-2}$\\
$0.11$ & 256/256 & 3{,}325 & $7.70\times10^{-2}$, $[6.79,\,8.70]\times10^{-2}$ & $8.07\times10^{-2}$ & $8.42\times10^{-2}$ & $2.30\times10^{-1}$\\
$0.12$ & 256/256 & 714 & $3.59\times10^{-1}$, $[3.16,\,4.05]\times10^{-1}$ & $3.77\times10^{-1}$ & $4.73\times10^{-1}$ & $1.30$\\
$0.15$ & 256/256 & 256 & $1.00$, $[0.88,\,1]$ & $1.00$ & $1.30\times10^{1}$ & $3.60\times10^{1}$\\
$0.20$ & 256/256 & 256 & $1.00$, $[0.88,\,1]$ & $1.00$ & $1.89\times10^{2}$ & $5.21\times10^{2}$\\
\hline
\end{tabular}
\caption{The per-step probability that $\Lev(\Nhost)$ leaves the exact orbit under
read noise of standard deviation $s$. One noise-free copy and $256$ noisy copies
were stepped together from the initial state of the self-describing instance,
$3000$ steps for $s\le0.06$ and $4000$ otherwise, the full state compared after
every step; the exposure is the number of step--copy pairs before a first
deviation, and the interval is the exact $95\%$ interval for a Poisson count over
it. $b'$ and $b$ are the product form and the sum of Proposition~\ref{prop:profile}
averaged over the same exposure; at $s=0.08$ and $s=0.09$ $b'$ bounds the
expected number of deviating copies by $6.4$ and $212.0$, against $7$ and $202$
observed.}
\label{tab:curve}
\end{table}

The complete self-reproducing run was executed under the same noise on every
pre-activation of every layer at every step, a run ending when it emits a byte that
differs from $\ser(\Omega^\ast)$, when the halt bit is set, or when the whole
instance is emitted. On $\Lev(\Nhost)$ it emitted the whole instance at $s=0.05$,
and set the halt bit after $10{,}517$ steps and $143$ bytes at $s=0.10$, after
$135$ steps at $s=0.15$ and after $2$ at $s=0.20$. With $N_3$ in place of $\Nhost$,
the device reading each register bit as the majority of its three copies, it
emitted the whole instance and halted on the step of the exact orbit at $s=0.10$,
$0.15$, $0.20$ and $0.25$, where Theorem~\ref{thm:amplify}(ii) bounds the
probability of a deviation over the run by $1.54\times10^{-39}$,
$3.19\times10^{-12}$, $1.34\times10^{-2}$ and $4.14\times10^{2}$.

\subsection{The amplified host}

Repeated on $N_3$, the measurement of Table~\ref{tab:curve}
gives the rate at which some copy of some state bit leaves the orbit:
$9.79\times10^{-7}$, $8.47\times10^{-5}$, $2.46\times10^{-3}$ and
$2.58\times10^{-2}$ per step at $s=0.24$, $0.27$, $0.30$ and $0.33$. These rates
lie below the profile bound of Theorem~\ref{thm:amplify}(ii) by factors between
$9.8$ and $19$, since a flip in a layer before the last is absorbed unless a second
error meets it (Proposition~\ref{prop:amplify}(iii)). Restricted to the rows of the
last layer, whose wrong copies nothing outvotes, the bound is $8.54\times10^{-5}$
at $s=0.27$, inside the $95\%$ interval $[6.65,10.6]\times10^{-5}$ of the measured
rate, which comes from $74$ events. At $s=0.30$ to $0.36$ the measured rate exceeds
this restricted bound by factors $1.4$ to $1.6$, a wrong copy of an earlier layer
then meeting a second error more often.

Table~\ref{tab:bound3} measures the failure itself: $64$ noisy copies of $N_3$ were
stepped beside an exact one from the initial state of the self-describing instance
for up to $20{,}000$ steps, each stopped at its first wrong majority-read state. At
$s=0.27$ the restoration bound gives probability $0.091$ that a copy fails within the
$20{,}000$ steps, against $5$ of the $64$ copies.

\begin{table}[ht]
\centering\small
\begin{tabular}{@{}lrrlll@{}}
\hline
$s$ & Failed & Exposure & Measured rate, $95\%$ interval & Restoration bound & $rM\erfc(r/2\sqrt2s)$\\
\hline
$0.10$ & --- & --- & --- & $8.62\times10^{-59}$ & $9.24\times10^{-46}$\\
$0.15$ & --- & --- & --- & $4.02\times10^{-25}$ & $1.92\times10^{-18}$\\
$0.20$ & --- & --- & --- & $3.59\times10^{-13}$ & $8.03\times10^{-9}$\\
$0.22$ & --- & --- & --- & $1.82\times10^{-10}$ & $1.16\times10^{-6}$\\
$0.25$ & --- & --- & --- & $1.58\times10^{-7}$ & $2.48\times10^{-4}$\\
$0.27$ & 5/64 & 1{,}237{,}505 & $4.04\times10^{-6}$, $[1.31,\,9.43]\times10^{-6}$ & $4.57\times10^{-6}$ & $3.48\times10^{-3}$\\
$0.30$ & 63/64 & 295{,}163 & $2.13\times10^{-4}$, $[1.64,\,2.73]\times10^{-4}$ & $2.26\times10^{-4}$ & $7.22\times10^{-2}$\\
$0.33$ & 64/64 & 17{,}472 & $3.66\times10^{-3}$, $[2.82,\,4.68]\times10^{-3}$ & $4.34\times10^{-3}$ & $6.90\times10^{-1}$\\
$0.36$ & 64/64 & 2{,}188 & $2.93\times10^{-2}$, $[2.25,\,3.74]\times10^{-2}$ & $4.30\times10^{-2}$ & $3.89$\\
\hline
\end{tabular}
\caption{The per-step probability that the state of $N_3$ read by majority leaves the orbit under read noise of standard deviation $s$: $64$ noisy copies stepped beside an exact one from the initial state of the self-describing instance for up to $20{,}000$ steps, each stopped at its first wrong majority-read state, with the exact $95\%$ interval for a Poisson count over the exposure, against the bound $b(t)$ of Proposition~\ref{prop:restore} along the same orbit, weighted by the steps each copy ran where copies ran and averaged over the $20{,}000$ steps elsewhere, and the first-order bound of Theorem~\ref{thm:amplify}(ii).}
\label{tab:bound3}
\end{table}

\section{Verification}\label{sec:verification}

The artifact~\cite{Norton2026} holds the Rocq development, the generators, the
evaluators and the record of every run. Everything was run with Python 3.12.3,
PyTorch 2.14.0 with CUDA 12.6, NumPy 2.5.3, safetensors 0.8.0 and GCC 13.3.0 on
Ubuntu 24.04, on an Intel Core i9-12900H with $20$ hardware threads and $31$ GiB of
memory and an NVIDIA GeForce RTX 3070 Ti Laptop GPU with $8$ GiB of memory.

\subsection{The Rocq development}

The Rocq development, \texttt{rocq/} in the artifact, has $66$ files and $24{,}230$
lines. With Rocq 9.1.1, its standard library 9.2.0, Coquelicot 3.4.5, Interval
4.11.4 and MathComp 2.4.0 it builds in about two minutes. Table~\ref{tab:rocq} gives
the Rocq theorems behind each result. \texttt{rocq/Assumptions.v} names them and the lemmas they rest on,
$144$ theorems in all, and \texttt{src/paper.py rocq} rebuilds the development and
prints the axioms of each.

The theorems on netlists give every unit its value in terms of the state and show
by induction over the units that each unit fires to it; those on programs follow a
run step by step in $\mathrm{SUBLEQ}_A$ with the device of
Definition~\ref{def:transducer}; and those on stored networks relate each pass of
the ring machine to one step of the machine stored in it, the organism's decoder or
the host with its device.

Section~\ref{sec:robust} is proved on a probability space given by its properties.
A function $P$ on the events of the perturbations $(e_g)$ is monotone and
subadditive, is $0$ on the empty event and $1$ on the sure one, multiplies on events
that depend on disjoint sets of perturbations, and is additive and complementary on
the events built from the threshold events $\{c\le e_g\}$; each perturbation has the
tails $P(e_g\ge c)=Q(c/s)$ and $P(e_g\le c)=Q(-c/s)$, and in
Proposition~\ref{prop:flip} each output flips with probability $\varepsilon$. The
product of Gaussian measures with its outer measure is such a space: on events that
depend on disjoint sets of perturbations the outer measure is multiplicative, as
integrating the sections of a measurable cover shows.
Corollary~\ref{cor:synapse} assumes in addition that a sum of $m$ independent
centred Gaussians of deviation $s_{\mathrm w}$ is a centred Gaussian of deviation
$s_{\mathrm w}\sqrt m$. In Proposition~\ref{prop:restore} the device writes every
copy of its cells with the values of the run, as the device of $N_3$ does while the
state read by majority is the run's. The identity
$\bigl(\int_0^xe^{-t^2}dt\bigr)^2+\int_0^1e^{-x^2(1+t^2)}/(1+t^2)\,dt=\pi/4$ is proved
from its derivative, which is $0$, with the integrals and derivatives of
Coquelicot~\cite{BoldoLelayMelquiond2015}. With
$\erfc(x)=1-\tfrac2{\sqrt\pi}\int_0^xe^{-t^2}dt$ it gives $\erfc(x)\le e^{-x^2}$ for
$x\ge0$ and, for $x>0$,
\[
\frac{e^{-x^2}}{x\sqrt\pi}\Bigl(1-\frac1{2x^2}\Bigr)\le\erfc(x)
\le\frac{e^{-x^2}}{x\sqrt\pi}\Bigl(1-\frac1{2x^2}+\frac3{4x^4}\Bigr),
\]
from which the values of $\erfc$ in Sections~\ref{sec:certified}
and~\ref{sec:amplify} and in Corollary~\ref{cor:storednoise} are evaluated with the
interval tactic of the Interval library~\cite{Melquiond2008}: the failure
probability at $s=0.05$, the bound $TMp$ over the self-reproducing run, the
redundancies of Corollary~\ref{cor:synapse}, the threshold of
Proposition~\ref{prop:flip}, the first-order bound for $N_3$ over the run and the
noise levels of Corollary~\ref{cor:storednoise}. The profile bound of
Proposition~\ref{prop:profile}, the restoration bound of
Proposition~\ref{prop:restore} and the per-step values of Tables~\ref{tab:curve}
and~\ref{tab:bound3} are computed by the artifact in floating point. The results on real
numbers rest on the classical axioms of Rocq's real numbers and on functional
extensionality, and the evaluated values also on the specifications of its
primitive integers and floating-point numbers.

{\footnotesize
\setlength{\LTcapwidth}{\textwidth}
\begin{longtable}{@{}>{\raggedright\arraybackslash}p{0.27\textwidth}>{\raggedright\arraybackslash}p{0.69\textwidth}@{}}
\caption{The Rocq theorems behind each result.}\label{tab:rocq}\\
\hline
Result & Rocq theorems\\
\hline
\endfirsthead
\hline
Result & Rocq theorems\\
\hline
\endhead
\hline
\endfoot
Lemma~\ref{lem:lev} & \rocq{Lev.lev\_ok}\\
Theorem~\ref{thm:physnet} & \rocq{PhysSem.netlist\_step}, \rocq{PhysFacts.phys\_ternary}, \rocq{PhysLev.phys\_dmax}\\
Lemmas~\ref{lem:feedback} and~\ref{cor:compile} & \rocq{Schedule.netlist\_pass}, \rocq{Compile.compile\_unit}, \rocq{Compile.compile\_records}\\
Lemma~\ref{lem:code}, Proposition~\ref{prop:pass} & \rocq{Complete.describe\_complete}, \rocq{PassLemma.machine\_pass}\\
Theorem~\ref{thm:construct} & \rocq{Run.universal\_construction}, \rocq{Main.builds\_every\_image}\\
Theorem~\ref{thm:organism} & \rocq{Main.copies}, \rocq{Main.periodic}, \rocq{NetMain.network\_reproduces}, \rocq{NetMain.network\_periodic}\\
Lemmas~\ref{lem:sub} and~\ref{lem:mem} & \rocq{Adder.sub\_value}, \rocq{Gates.sAND\_fire}, \rocq{Gates.sOR\_fire}\\
Lemma~\ref{lem:branch}, Remark~\ref{rem:ternarydepth} & \rocq{Lower.branch\_unit}, \rocq{Lower.branch\_not\_ternary}\\
Lemma~\ref{lem:muxdepth} & \rocq{Lower.read\_needs\_two\_levels}, \rocq{Lower.onehot\_needs\_two\_levels}\\
Proposition~\ref{prop:host} & \rocq{HostSem.host\_step}, \rocq{HostLev.host\_dmax}, \rocq{HostCount.host\_weights\_biases}\\
Remark~\ref{rem:fold} & \rocq{FoldHost.fold\_step}, \rocq{FoldHost.fold\_dmax}\\
Lemmas~\ref{lem:complete}, \ref{lem:describe}, Prop.~\ref{prop:optimal} & \rocq{Recipe.recipes\_complete}, \rocq{Recipe.describe\_spec}, \rocq{Recipe.recipe\_optimal}\\
Theorem~\ref{thm:univ} & \rocq{Pstar.constructor\_run}\\
Theorem~\ref{thm:selfrep}, Proposition~\ref{prop:steps} & \rocq{Pstar.selfrep\_run}\\
Proposition~\ref{prop:storedpass} & \rocq{StoredProp.stored\_pass}, \rocq{StoredWbr.wbr\_stored}\\
Theorem~\ref{thm:storedconstruct} & \rocq{StoredRun.stored\_construct}\\
Theorem~\ref{thm:storedrep} & \rocq{StoredMain.copies}, \rocq{StoredMain.periodic}, \rocq{StoredMain.network\_reproduces}, \rocq{StoredMain.network\_periodic}\\
Theorem~\ref{thm:widths} & \rocq{StoredWidth.width\_exists}, \rocq{StoredCount.stored\_records}\\
Theorem~\ref{thm:margin}, Proposition~\ref{prop:mismatch} & \rocq{Margin.margin}, \rocq{Margin.margin\_tight}, \rocq{Margin.mismatch\_all}, \rocq{Margin.mismatch\_run}\\
Propositions~\ref{prop:random} and~\ref{prop:profile} & \rocq{Noise.union\_gauss}, \rocq{Noise.profile\_step}, \rocq{Noise.copy\_bound}\\
Proposition~\ref{prop:amplify}, Theorem~\ref{thm:amplify} & \rocq{Margin.amap\_rep}, \rocq{Margin.amp\_margin}, \rocq{Noise.amp\_gauss}, \rocq{Noise.amp\_profile}\\
Corollary~\ref{cor:synapse} & \rocq{Noise.synapse\_bound}\\
Proposition~\ref{prop:flip} & \rocq{Noise.flip\_bound}, \rocq{FlipLower.flip\_fails}\\
Corollary~\ref{cor:amplifiedselfrep} & \rocq{AmpHost.amp\_selfrep}\\
Proposition~\ref{prop:restore} & \rocq{Restore.restoration}\\
Corollary~\ref{cor:storednoise} & \rocq{NoiseNum.generation\_1}, \rocq{NoiseNum.generation\_3}, \rocq{NoiseNum.generation\_5}\\
The other values of Section~\ref{sec:robust} & the theorems of \rocq{NoiseNum}, from the bounds on $\erfc$ in \rocq{Erfc}\\
Appendix~\ref{app:sizes} & the counts in \rocq{PhysFacts}, \rocq{PhysLev}, \rocq{HostCount}, \rocq{HostLev}, \rocq{FoldHost} and \rocq{StoredCount}\\
\end{longtable}}

\subsection{Evaluators and runs}

Four evaluators of the host were used: (a) an integer reference of
Definitions~\ref{def:transducer} and~\ref{def:subleq}; (b) the netlist that
$\sigma(\Nhost)$ denotes, unit by unit; (c) $\Lev(\Nhost)$; and (d)
\texttt{src/netlist\_eval.c}, with its own reference, device, parser and
levelization. The self-describing instance ran on (a), (b) and (d), $793{,}709$
steps emitting $\sigma(\Nhost)$, and $\Omega^\ast$ for three generations on (a), (b)
and (d) and one on (c), each started only from the output of the one before, split
by $\inst$ and rebuilt by $\sigma^{-1}$; (d) compared the full state and the output
with its reference after every step. One step of $N_3$ agreed with
Definition~\ref{def:subleq} on $65{,}536$ replicated random states, and on~(d)
$\Omega^\ast_3$ halted after the $33{,}796{,}449$ steps of
Proposition~\ref{prop:steps} and emitted $\ser(\Omega^\ast_3)$. Separate code in
\texttt{src/independent.py}, written from the definitions of this paper, rebuilds
$H_A$ for $A=3$ to $12$, levelizes $\sigma(\Nhost)$ and repeats the measurements of
Section~\ref{sec:certified} and Appendix~\ref{app:measure}, and rebuilding the host
from source reproduces $\sigma(\Nhost)$ byte for byte. One step of the folded host
agreed with $\mathrm{SUBLEQ}_A$ on $2^{16}$ random states for each $A$ from $3$ to
$10$, and on $2^{15}$ and $2^{14}$ for $A=11$ and $12$, and none of the $137{,}781$
ternary units on eight inputs with biases in $[-10,10]$ computes the branch
predicate.

The encoder satisfies Lemma~\ref{lem:describe} on the $29$ serialized machines of
the artifact, $971{,}266{,}155$ bytes, and on $68{,}061$ strings straddling every
boundary of the token format; the shortest recipe attains the bound of
Lemma~\ref{lem:complete} on $3{,}203$ strings without equal adjacent bytes; the step
counts of Theorem~\ref{thm:univ} and Proposition~\ref{prop:steps} agree with the
reference on $880$ recipes; and $\inst(\ser(\Omega))=\Omega$ on $310$ instances. The
constructor emitted all $29$ machines on the reference, and four of them on the
netlist of $\sigma(\Nhost)$.

For $N^6_{3,160,3200}$, of $151{,}892$ units on $32{,}301$ state bits, and for
$N^{12}_{2,3,40}$ and $N^6_{5,2,7}$, \texttt{src/twohead.py} writes byte for byte the
file that the extracted netlist definition writes, and one step of the netlist
agreed with a transcription of Definition~\ref{def:twohead} on $700$ random states of
three other members. Six generations of the organism, $1{,}191{,}183$ steps each,
ran in $\mathcal M_{3,160,3200}$ with the outcome Theorem~\ref{thm:organism} states,
the first also on $N^6_{3,160,3200}$ beside the transcription, the full state
compared after every step. The extracted definition of $\mathcal S_{A,w}$ writes for
$A=8$ and $9$ the images that \texttt{src/twohead.py} builds, and
\texttt{src/ring\_run.c} ran six generations of $\mathcal S$ in a ring of three banks,
$S=36{,}085$ and $L=1{,}927{,}264$, about eight minutes each, with the outcome
Theorem~\ref{thm:storedrep} states. One pass of $\mathcal S$ agreed with the
startup, one step of $\Nhost$ and the device of Definition~\ref{def:ring} from $300$
random states, as did one pass of $\mathcal S_9$ from $100$ and of $\mathcal S_{10}$
from $30$, and reduced ordered decision diagrams show that from every window each
block of $\mathcal S_8$ to $\mathcal S_{10}$ leaves what Construction~\ref{con:stored}
and Definition~\ref{def:ring} prescribe. The counts of Theorem~\ref{thm:widths} agree
with the construction for $A=8$ to $12$.

\subsection{Digests}

Every run of $\Omega^\ast$ produced the same output. Table~\ref{tab:digests} gives
the SHA-256 digests of three files of the artifact, in the formats of
Appendix~\ref{app:sigma}, the organism's image taken as $3718$ bytes of value $0$ or
$1$; the digests of the others are recorded with the runs that produce them, in
\texttt{paper/runs/}.
\begin{table}[htb]\centering\ttfamily\footnotesize
\begin{tabular}{@{}l@{\quad}l@{}}
\normalfont $\ser(\Omega^\ast)$ & c5950d3fec6bcf70e8111586a030d21b2dfa55bd69a1c7fe5f016cb06e7fb523\\
\normalfont organism $\rho$ & d4913db0b5ef6d3f122c44974cf9a58015fe7dabacb7176cf377d5cd0743f00b\\
\normalfont stored host $\mathcal S$ & 9a7821a540fed2ed58fbb263213f65aead6bbbd8cca6344abcabcc5f4289903c
\end{tabular}
\normalfont\caption{The SHA-256 digests of the artifact's files.}\label{tab:digests}
\end{table}


\end{document}